\documentclass{llncs} 
\usepackage{amsmath,amsfonts,amssymb,xspace}
\usepackage{wasysym} 
\usepackage{color}

\usepackage[inline]{enumitem}
\setlist{nosep}

\usepackage{array}
\newcolumntype{x}[1]{>{\centering\let\newline\\\arraybackslash\hspace{0pt}}p{#1}}

\usepackage{tikz}

\usetikzlibrary{arrows,shapes,snakes,automata,backgrounds,positioning}

\tikzstyle{player1}=[draw,rounded rectangle, minimum size=5mm]
\tikzstyle{player2}=[draw,rectangle,minimum size=5mm]
\tikzstyle{widget}=[draw=gray,rectangle, rounded rectangle=10pt,dashed,minimum size=6mm]
\tikzset{every loop/.style={looseness=7}}

\newcommand{\set}[1]{\left\{ #1 \right\}}

\newcommand{\seq}[1]{\langle #1 \rangle}

\renewcommand\geq{\geqslant}
\renewcommand\leq{\leqslant}

\newcommand\clocks{\mathcal X}
\newcommand\valuation{\nu}

\newcommand\Nat{\mathbb{N}}
\newcommand\Int{\mathbb{Z}}

\newcommand\Real{\mathbb{R}}
\newcommand\Rplus{\mathbb{R}_{\geq 0}}

\newcommand\Rpos{\Rplus}
\newcommand\R{\Real}
\newcommand\Z{\Int}
\newcommand\N{\Nat}

\newcommand\arena{\mathcal A}
\newcommand\locations{L}
\newcommand\minlocations{L_{1}}
\newcommand\maxlocations{L_{2}}
\newcommand\invariants{\mathrm{Inv}}
\newcommand\labels{\Sigma}

\newcommand\transitions{\delta}
\newcommand\prices{\omega}
\newcommand\goals{\locations_f}
\newcommand\powerset[1]{2^{#1}}
\newcommand\location{\ell}

\newcommand\pricedarena{\mathcal G}
\newcommand\vertices{V}
\newcommand\edges{E}
\newcommand\actions{A}
\newcommand\finalvertices{\vertices_f}

\newcommand\Price{\mathsf{Cost}}
\newcommand\Weight{\Price}
\newcommand\edgeweights{\pi}
\newcommand\edgeprices{\pi}
\newcommand\fstrategy{\xi}
\newcommand\minfstrategy{\fstrategy_1}
\newcommand\maxfstrategy{\fstrategy_2}

\newcommand\timedprices{\kappa}

\newcommand\sstop{\textsf{Stop}}

\newcommand{\sem}[1]{ [ \! [ {#1}  ]  \! ]} 
\newcommand\states{S}
\newcommand\timedactions{\Gamma}

\newcommand\timedtransitions{\Delta}
\newcommand\finalstates{\states_f}

\newcommand\StepsToReach{\mathsf{Length}}

\newcommand\ReachOb{\mathsf{Reach}}
\newcommand\BReachOb{\mathsf{TBReach}}
\newcommand\RReachOb{\mathsf{RReach}}

\newcommand\oneBPTG{\ensuremath{\mathrm{1BPTG}}\xspace}
\newcommand\onePTG{\ensuremath{\mathrm{1PTG}}\xspace}

\newcommand\strategies{\mathsf{Strat}}
\newcommand\minstrategies{\strategies_{1}}
\newcommand\maxstrategies{\strategies_{2}}
\newcommand\strategy{\sigma}
\newcommand\minstrategy{\strategy_1}
\newcommand\maxstrategy{\strategy_2}
\newcommand\outcomes{\mathsf{Play}}

\newcommand\tadelay[1]{\mathsf{del}(#1)}
\newcommand\taaction[1]{\mathsf{lab}(#1)}

\newcommand\Last{\textsf{Last}}
\newcommand\Play{\textsf{Play}}
\newcommand\FPlay{\textsf{FPlay}}

\newcommand\uppervalue{\overline{\mathsf{Val}}}
\newcommand\lowervalue{\underline{\mathsf{Val}}}

\newcommand\Value{\mathsf{Val}}
\newcommand\pricestrategy{\Price}

\newcommand\intervals{\mathcal I}
\newcommand\leqinterval{\preceq} 

\newcommand\uniformminstrategies{\mathsf{UStrat}_{1}}
\newcommand\uniformmaxstrategies{\mathsf{UStrat}_{2}}
\newcommand\uniformuppervalue{\overline{\mathsf{UVal}}}
\newcommand\uniformlowervalue{\underline{\mathsf{UVal}}}
\newcommand\convergentminstrategies{\mathsf{CStrat}_{1}}
\newcommand\convergentmaxstrategies{\mathsf{CStrat}_{2}}

\newcommand\convergentlowervalue{\underline{\mathsf{CVal}}}

\newcommand\abstrarena{\ensuremath{\tilde{\arena}}}

\renewcommand\paragraph[1]{\noindent \textbf{#1.}}

\usepackage[disable]{todonotes}

\title{\texorpdfstring{Adding Negative Prices to Priced Timed
    Games\thanks{The research leading to these results has received
      funding from the European Union Seventh Framework Programme
      (FP7/2007-2013) under Grant Agreement number 601148
      (CASSTING)}}{Adding Negative Prices to Priced Timed Games} }

\author{
  Thomas Brihaye\inst{1}, 
  Gilles Geeraerts\inst{2}\thanks{Supported by a `Cr\'edit aux chercheurs' 
    number 1808881 of the F.R.S./FNRS.}, 
  Shankara Narayanan Krishna\inst{3}, \\
  Lakshmi Manasa\inst{3}, 
  Benjamin Monmege\inst{2}, and 
  Ashutosh Trivedi\inst{3}}
\institute{
  Universit\'e de Mons, Belgium, \email{thomas.brihaye@umons.ac.be} 
  \and 
  Universit\'e libre de Bruxelles, Belgium, \email{gigeerae,benjamin.monmege@ulb.ac.be}  
  \and 
  IIT Bombay, India,
  \email{krishnas,manasa,trivedi@cse.iitb.ac.in}} 

\usepackage[pdftex,bookmarks=true,bookmarksopen=true,bookmarksopenlevel=2,bookmarksdepth=2]{hyperref}

\begin{document}

\maketitle

\begin{abstract}
  Priced timed games (PTGs) are two-player zero-sum games played on
  the infinite graph of configurations of priced timed automata where
  two players take turns to choose transitions in order to optimize
  cost to reach target states.  Bouyer et al. and Alur, Bernadsky, and
  Madhusudan independently proposed algorithms to solve PTGs with
  nonnegative prices under certain divergence restriction over prices.
  Brihaye, Bruy\`ere, and Raskin later provided a justification for
  such a restriction by showing the undecidability of the optimal
  strategy synthesis problem in the absence of this divergence
  restriction.  This problem for PTGs with one clock has long been
  conjectured to be in polynomial time, however the current best known
  algorithm, by Hansen, Ibsen-Jensen, and Miltersen, is exponential.
  We extend this picture by studying PTGs with both negative and
  positive prices. We refine the undecidability results for optimal
  strategy synthesis problem, and show undecidability for several
  variants of optimal reachability cost objectives including
  reachability cost, time-bounded reachability cost, and repeated
  reachability cost objectives.  We also identify a subclass with
  bi-valued price-rates and give a pseudo-polynomial\todo{G: removed
    `(polynomial when prices are
    nonnegative)'} 
  algorithm to partially answer the conjecture on the complexity of
  one-clock PTGs.

\end{abstract}

\section{Introduction}
\label{sec:introduction}
Timed automata~\cite{AluDil94} equip finite automata with a finite
number of real-valued variables---aptly called clocks---that evolve
with a uniform rate.  The syntax of timed automata also permits
specifying \emph{transition guards} and \emph{location (state)
invariants} using the constraints over clock valuations, and resetting
the clocks as a means to remember the time since the execution of a
transition.  Timed automata is a well-established formalism to specify
time-critical properties of real-time systems.  Priced timed
automata~\cite{AluLa-04,BehFeh01} (PTAs) extend timed automata with
price information by augmenting locations with price-rates and
transitions with discrete prices.  The natural reachability-cost
optimization problem for PTAs is known to be decidable with the same
complexity~\cite{BouBri07} as the reachability problem
(PSPACE-complete), and forms the backbone of many applications of
timed automata including scheduling and planning.

Priced timed games (PTGs) extend the reachability-cost optimization
problem to the setting of competitive optimization problem, and form
the basis of optimal controller synthesis \cite{RW89} for real-time
systems.  We study turn-based variant of these games where the game
arena is a PTA with a partition of the locations between two players
Player~1 and Player~2.  A play of such a game begins with a token in
an initial location, and at every step the player controlling the
current location proposes a valid timed move, i.e., a time delay and a
discrete transition, and the state of the system is modified
accordingly.  
The play stops if the token reaches a location from a distinguished set of
\emph{target locations}, and the payoff of the play is equal to the cost
accumulated before reaching the target location.
If the token never reaches a target location
then the game continues forever, and the payoff in this case is
$+\infty$ irrespective of actual cost of the infinite play.  We
characterize a PTG according to the objectives of Player~1.  Since we
study zero-sum games, the objective of Player~2 is also implicitly
defined. We study PTGs with the following objectives:
\begin{enumerate*}[label=(\emph{\roman*})]
\item \textit{Constrained-price reachability} objective $\ReachOb({\bowtie} K)$
is to achieve a payoff $C$ of the play such that $C \bowtie K$ where
${\bowtie} \in \set{\leq, <, = , >, \geq}$ and $K \in \Nat$; 
\item \textit {Bounded-time reachability} objective $\BReachOb(K, T)$
is to keep the payoff of the play less than $K$ while keeping the
total time elapsed within $T$ units; and 
\item \textit{Repeated reachability} objective $\RReachOb(\eta)$ is to
visit target infinitely often with a payoff in the interval $[-\eta, \eta]$.
\end{enumerate*}
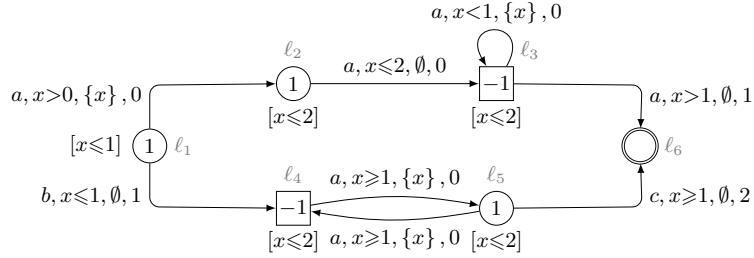
\begin{figure}[t]
  \centering
  \scalebox{.9}{
    \begin{tikzpicture}[node distance=3cm,auto,-{>},{>}=latex]
      \node[player1](1){\makebox[0mm][c]{$1$}}; 
      \node()[right of=1,node distance=5mm,color=gray]{$\location_1$};
      \node()[left of=1,node distance=8mm]{$[x{\leq} 1]$};
      \node[player1](2)[above right of=1,yshift=-1.2cm]{\makebox[0mm][c]{$1$}}; 
      \node()[above of=2,node distance=5mm,color=gray]{$\location_2$};
      \node()[below of=2,node distance=5mm]{$[x{\leq} 2]$};
      \node[player2](3)[right of=2]{\makebox[0mm][c]{$-1$}};
      \node()[above right of=3,node distance=7mm,color=gray]{$\location_3$};
      \node()[below of=3,node distance=5mm]{$[x{\leq} 2]$};
      \node[player2](4)[below right of=1,yshift=1.2cm]{\makebox[0mm][c]{$-1$}};
      \node()[above of=4,node distance=5mm,color=gray]{$\location_4$};
      \node()[below of=4,node distance=5mm]{$[x{\leq} 2]$};
      \node[player1](5)[right of=4]{\makebox[0mm][c]{$1$}}; 
      \node()[above of=5,node distance=5mm,color=gray]{$\location_5$};
      \node()[below of=5,node distance=5mm]{$[x{\leq} 2]$};
      \node[player1,double](6)[below right of=3,yshift=1.2cm]{};
      \node()[right of=6,node distance=5mm,color=gray]{$\location_6$};

      \draw[rounded corners]
      (1) -- node[left,near end]{$a, x{>}0, \set{x},0$} (0, 0.9) -- (2);

      \draw[rounded corners]
      (1) -- node[left,near end]{$b, x{\leq}1, \emptyset, 1$} (0, -0.9) -- (4);

      \draw[rounded corners]
      (3) -- (7.27, 0.9) -- node[right,,near start]{$a, x{>}1, \emptyset, 1$} (6);

      \draw[rounded corners]
      (5) -- (7.27, -0.9) -- node[right,,near start]{$c, x{\geq}1, \emptyset, 2$} (6);
      
      \path
      (2) edge node[above]{$a,x{\leq} 2, \emptyset, 0$} (3)
      (3) edge[in=120,out=60,loop] node[above]{$a,x{<}1, \set{x},0$} (3)
      (4) edge[bend left=10] node[above]{$a,x{\geq} 1, \set{x},0$} (5) 
      (5) edge[bend left=10] node[below]{$a,x{\geq} 1, \set{x},0$} (4);

    \end{tikzpicture}}
  \caption{A price timed game arena with one clock}
  \label{fig:BPTG}
\end{figure}

  An example of PTG with clock variable $x$ and six locations is
  given in Fig.~\ref{fig:BPTG}.  We depict Player~1 locations as
  circles and Player~2 locations as boxes.  The numbers inside
  locations denote their price-rates, while the clock constraints next
  to a location depicts its invariant.  We denote a transition, as
  usual, by an arrow between two location annotated by a tuple $a, g,
  r, c$ where $a$ is the label, $g$ is the guard, $r$ is the clocks
  reset set, and $c$ is the cost of the transition.

\paragraph{Related work} 
PTGs with constrained-price reachability objective $\ReachOb({\leq}
K)$ were independently introduced in \cite{BouCas04} and
\cite{AluBer04}, with semi-algorithms to decide the existence of
winning strategy for Player~1 in PTGs with nonnegative prices. They
also showed that under the \emph{strongly non-Zeno assumption} on
prices the proposed semi-algorithms always terminate.
This assumption was justified in \cite{BriBru05} by showing that, in
the absence of non-Zeno assumption, the problem of deciding the
existence of winning strategy for the objective $\ReachOb(\leq K)$ is
undecidable for PTGs with five or more clocks. This result has been
later refined in \cite{BouBri06} by showing that the problem is
undecidable for PTGs with three or more clocks and nonnegative prices.
In \cite{BCJ09} is showed the undecidability of the existence of
winning strategy problem for $\ReachOb({\leq} K)$ objective over PTGs
with both positive and negative price-rates and two or more clocks.

On a positive side, the existence of winning strategy for
$\ReachOb({\leq} K)$ problem for PTGs with one clock when the
price-rates are restricted to values $0$ and $d \in \Nat$ has been
shown decidable in \cite{BriBru05}, by proving that the
semi-algorithms in \cite{BouCas04,AluBer04} always terminate. However,
the authors did not provide any complexity analysis of their
algorithm.
One-clock PTGs with nonnegative prices are reconsidered
in~\cite{BouLar06}, and a 3-EXPTIME algorithm is given to solve the
problem, while the best known lower bound is PTIME.
A tighter analysis of the problem is presented in \cite{Rut11} that
lowered the known complexity of this problem to EXPTIME, namely
$2^{O(n^2+m)}$ where $n$ is the number of locations and $m$ is the
number of transitions. A significant improvement over the complexity
($m12^{n}n^{O(1)}$) was given in \cite{DueIbs13} by improving the
analysis of the semi-algorithms by \cite{BouCas04,AluBer04}.

\paragraph{Contributions}
We consider PTGs with both negative and positive prices.  We show that
deciding the existence of a winning strategy for reachability
objective $\ReachOb({\bowtie} K)$ is undecidable for PTGs with two or
more clocks. In \cite{OW10}, a theory of time-bounded verification has
been proposed, arguing that restriction to bounded-time domain
reclaims the decidability of several key verification problems.  As an
example, we cite~\cite{BDGORW13} where authors recovered the
decidability of the reachability problem for hybrid automata under
time-bounded restriction.  We begin studying PTGs with bounded
reachability objective $\BReachOb(K, T)$ hoping that the problem may
be decidable due to time-bounded restriction.  However, we answer this
question negatively by showing undecidability of the existence of
winning strategy problem for PTGs with six or more clocks.  We also
show the undecidability for the corresponding problem for repeated
reachability objective $\RReachOb(\eta)$ for PTGs with three or more
clocks.

On the positive side, we introduce a previously unexplored subclass of
one-clock PTGs, called one-clock bi-valued priced timed games
(1BPTGs), where the price-rates of locations are taken from a set of
two integers from $\{-d,0,d\}$ (with $d$ any positive integer).  None
of the previously cited algorithms can be applied in this case since
we do not assume non-Zenoness of prices and consider both positive and
negative prices.  After showing a determinacy result for 1BPTGs, we
proceed to give a pseudo-polynomial time algorithm to compute the
value and $\varepsilon$-optimal strategy for both players with
$\ReachOb({\leq} K)$ objective. \todo{G: removed `The complexity drops
  to polynomial for 1BPTGs if the price-rates are non-negative
  integers.  This gives a polynomial time algorithm for the
one-clock PTG problem studied in~\cite{BriBru05}.'}

\paragraph{Remark on a previous version of this paper} Note that the
five first arXiv revisions of the present paper (up to
\url{https://arxiv.org/abs/1404.5894v5}), as well as the published
version of the paper \cite{concur14}, contained a flawed claim that
has been kindly pointed out to us by the authors of
\cite{FIS20}. Namely, we were claiming that our results entail that
`the complexity [of computing the value and $\varepsilon$-optimal
strategy for both players with $\ReachOb({\leq} K)$ objective] drops
to \emph{polynomial} for 1BPTGs if the price-rates are non-negative
integers.' While our theorems and corollaries remain correct, this
conclusion was not, and we are indebted to those authors for kindly
letting us know. Such claims have been removed from the present
version.

\section{Reachability-cost games on priced game graphs}
\label{sec:preliminaries}

PTGs can be considered as a succinct representation of some games on
uncountable state space characterized by the configuration graph of
timed automata. 

We begin by introducing the concepts and notations related to such
more general game arenas that we call priced game graphs.

\begin{definition}
  A \emph{priced game graph} is a tuple $\pricedarena=
  (\vertices,\actions,\edges,\edgeprices,\finalvertices)$ where:
  \begin{itemize}
  \item $\vertices=\vertices_1\uplus \vertices_2$ is the set of
    vertices partitioned into the sets $\vertices_1$ and
    $\vertices_2$;
  \item $\actions$ is a set of labels called actions;
  \item $\edges\colon \vertices\times\actions\to \vertices$ is the
    edge function defining the set of labeled edges;
  \item $\edgeprices\colon \vertices\times\actions \to \R$ is the
    price function that assigns prices to edges; and
  \item $\finalvertices\subseteq \vertices$ is the set of target
    vertices.
\end{itemize}
We call a game graph \emph{finite} if both $\vertices$ and $\actions$
are finite and with rational prices.
\end{definition}

A reachability-cost game begins with a token placed on some initial
vertex $v_0$.  At each round, the player who controls the current
vertex $v$ chooses an action $a \in \actions$ and the token is moved
to the vertex $\edges(v, a)$.  The two players continue moving the
token in this fashion, and give rise to an infinite sequence of
vertices and actions called a play of the game.
Formally, a finite play $r$ is a finite sequence of vertices and
actions $\seq{v_0, a_0, v_1, a_1, \ldots, a_{n-1}, v_n}$ where for
each $0 \leq i < n$ we have that $v_{i+1} = \edges(v_i, a_{i})$; we
write $\Last(r)$ for the last vertex of a finite play, here $\Last(r)
= v_n$.  An infinite play is defined analogously.  We write
$\FPlay_\pricedarena$ ($\FPlay_\pricedarena(v)$) for the set of finite
plays (starting from the vertex $v$) of the game graph $\pricedarena$.
We often omit the subscript when the game arena is clear form the
context.  We similarly define $\Play$ and $\Play(v)$ for the set of
infinite plays.  For all $k\geq 0$, we let $r[k]$ be the prefix
$\seq{v_0, a_0, \ldots, a_{k-1}, v_k}$ of $r$, and we denote by
$\Price(r[k])=\sum_{i=0}^{k-1}\edgeprices(v_{i}, a_i)$ its cost.  We
write $\sstop(r)$ for the index of the first target vertex in $r$,
i.e., $\sstop(r) = \inf \set{k\::\: v_k \in F}$.  We define the cost
of an infinite run $r = \seq{v_0, a_1, v_1, \ldots}$ as $\Price(r) =
+\infty$ if $\sstop(r) = \infty$ and $\Price(r) =
\Price(r[\sstop(r)])$, otherwise.

A strategy for a Player $i$ (for $i \in \set{1, 2}$) is a partial
function $\sigma: \FPlay \to \actions$ that is defined for a run $r =
\seq{v_0, a_0, v_1, \ldots, a_{n-1}, v_n}$ if $v_n \in V_i$ and is
such that $\edges(v_n, \sigma(r))$ is defined, i.e., there is a
$\sigma(r)$-labeled outgoing transition from $v_n$.  We denote by
$\strategies_i(\pricedarena)$ (or $\strategies_i$ when the game arena
is clear) the set of strategies for Player~i. Given a strategy profile
$(\minstrategy,\maxstrategy)\in\minstrategies \times \maxstrategies$
for both players, and an initial vertex $v\in\vertices$, the unique
infinite play $\outcomes(v,\minstrategy,\maxstrategy)= \seq{v_0,
  a_0,v_1,\ldots v_{k}, a_k,v_{k+1},\ldots}$ is such that for all
$k\geq 0$ if $v_{k}\in \vertices_i$, for $i=1,2$, then
$a_{k+1}=\sigma_i(r[k])$ and $v_{k+1}=\edges(v_k, a_{k+1})$.  A
strategy $\sigma$ is said to be \emph{memoryless} (or
\emph{positional}) if, for all finite plays $r, r' \in \FPlay$ with
$\Last(r) = \Last(r')$ we have that $\sigma(r) = \sigma(r')$.
Similarly, \emph{finite-memory strategies} can be defined as
implementable with Moore machines, see \cite{priced-games} for a
formal definition.

We consider optimal reachability-cost games on priced game graphs, where the
goal of Player~1 is to minimize the reachability-cost, while the goal of
Player~2 is the opposite. 
The standard concepts of upper value and lower value of the optimal
reachability-cost game are defined in straightforward manner. 
Formally, the upper-value $\uppervalue_\pricedarena(v)$ and lower
value $\lowervalue_\pricedarena(v)$ of a game starting from a vertex
$v$ is defined as $\uppervalue_\pricedarena(v) =
\inf_{\minstrategy\in\minstrategies}
\sup_{\maxstrategy\in\maxstrategies}
\Price(\outcomes(v,\minstrategy,\maxstrategy))$ and
$\lowervalue_\pricedarena(v) = \sup_{\maxstrategy\in\maxstrategies}
\inf_{\minstrategy\in\minstrategies}
\Price(\outcomes(v,\minstrategy,\maxstrategy))$.
It is easy to see that $\lowervalue_\pricedarena(v)\leq
\uppervalue_\pricedarena(v)$ for every vertex $v$.  We say that a game
is \emph{determined} if the lower and the upper values match for every
vertex $v$, and in this case, we say that the optimal value of the
game exists and we let $\Value_\pricedarena(v)=
\lowervalue_\pricedarena(v)= \uppervalue_\pricedarena(v)$.
The determinacy of these games follow from Martin's determinacy
theorem, and an alternative proof is given in~\cite{priced-games}.

In the following, we write $\pricestrategy(v,\minstrategy)$ for the
value of the strategy $\minstrategy$ of Player~1 from vertex $v$,
i.e., $\pricestrategy(v,\minstrategy) =
\sup_{\maxstrategy\in\maxstrategies}
\Price(\outcomes(v,\minstrategy,\maxstrategy))\,.$ A strategy
$\minstrategy^*$ of Player 1 is said to be optimal from $v$ if
$\pricestrategy(v,\minstrategy^*) = \uppervalue_\pricedarena(v)\,.$
Optimal strategies do not always exist, hence we also define
$\varepsilon$-optimal strategies. For $\varepsilon>0$, a strategy
$\minstrategy$ is an $\varepsilon$-optimal strategy if for all vertex
$v\in\vertices$, $\pricestrategy(v,\minstrategy) \leq
\uppervalue_\pricedarena(v)+\varepsilon\,.$ In this paper we exploit
the following result from~\cite{priced-games}.
\vspace{-1em}

\begin{theorem}[\cite{priced-games}]\label{thm:optimal-strategy}
  Let $\pricedarena$ be a finite priced game graph.
  \begin{enumerate}
  \item Deciding $\Value_\pricedarena(v) = +\infty$ is in Polynomial Time.
  \item Deciding $\Value_\pricedarena(v) = -\infty$ is in
    $\mathrm{NP}\cap\mathrm{co\text{-}NP}$, can be achieved in
    pseudo-polynomial time\footnote{
      polynomial time if the prices are encoded in unary.} and is as hard as
    solving mean-payoff games~\cite{ZP96}.
  \item Given $-\infty < \Value_\pricedarena(v) < +\infty$ for every vertex $v$,
    optimal strategies exist for both players. 
    In particular, Player~2 has optimal memoryless strategies,
    while Player~1 has optimal finite-memory strategies. 
    Moreover, the values $\Value_\pricedarena(v)$, as well as optimal strategies, 
    can be computed in pseudo-polynomial time.
  \end{enumerate}
\end{theorem}

It must be noticed that, in the presence of negative costs, and even
when every vertex $v$ has a finite value
$\Value_\pricedarena(v)\in\R$, memoryless optimal strategies may not
exist for Player~1, as pointed out in~\cite[Example~1]{priced-games}.

\section{Priced timed games}
\label{sec:definition}
In order to formally introduce priced timed games, we need to define the
concepts of clocks, clock valuations, constraints, and zones. 
Let $\clocks$ be a finite set of real-valued variables
called \emph{clocks}. A clock valuation on $\clocks$ is a function
$\valuation\colon \clocks\to\Rpos$ and we write $V(\clocks)$ for the
set of clock valuations. Abusing notation, we also treat a valuation
$\valuation$ as a point in $\R^{|\clocks|}$. If $\valuation\in
V(\clocks)$ and $t \in \Rpos$ then we write $\valuation+t$ for the
clock valuation defined by $(\valuation+t)(c) = \valuation(c)+t$ for
all $c \in \clocks$. For $C \subseteq \clocks$, we write
$\valuation[C:=0]$ for the valuation where $\valuation[C:=0](c)$
equals~$0$ if $c \in C$ and $\valuation(c)$ otherwise.
A clock constraint over $\clocks$ is a conjunction of simple
constraints of the form $c \bowtie i$ or $c-c' \bowtie i$, where
$c,c'\in\clocks$, $i\in \N$ and ${\bowtie} \in \{<,>,=,\leq,\geq\}$. 
A clock zone is a finite set of clock constraints that
defines a convex set of clock valuations. 
We write $Z(\clocks)$ for the set of clock zones over the set of clocks
$\clocks$. 
\begin{definition}
  A priced timed game is a tuple $\arena = (\locations,
  \clocks, \invariants, \labels, \transitions, \prices, \goals)$ where:
  \begin{itemize}
  \item 
    $\locations=\locations_1\uplus \locations_2$ is a finite set
    of locations, partitioned into the sets $\locations_1$ and $\locations_2$;
  \item 
    $\clocks$ is a finite set of clocks;
  \item 
    $\invariants \colon \locations \to Z(\clocks)$ associates an invariant to
    each location;
  \item
    $\labels$ is a finite set of labels;
  \item
    $\transitions \colon \locations\times \labels \to Z(\clocks) \times
    \powerset\clocks \times \locations$ is a transition function that
    maps a location $\location\in\locations$ and label
    $a\in\labels$ to a clock zone $\zeta\in Z(\clocks)$ representing
    the guard on the transition, a set of clocks $R\subseteq \clocks$
    to be reset and successor location $\location'\in\locations$;
  \item
    $\prices\colon \locations\cup\labels \to \Z$ is the price
    function; and
  \item 
    and $\goals \subseteq \locations$ is the set of target
    locations.
  \end{itemize}
\end{definition}

A configuration of a PTG is a tuple $(\location,\valuation)
\in\locations\times V$ where $\location$ is a location, $\valuation$
is a clock valuation and $\valuation\in\invariants(\location)$.  A
timed action is a tuple $\tau = (t,a) \in \Rpos\times\labels$ where
$t$ is a time delay and $a$ is a label.  In the following, for a timed
move $\tau=(t,a)\in\Rpos\times \labels$, we let $\tadelay\tau=t$ be
the delay part and $\taaction\tau=a$ be the label part.
The semantics of a PTG is given as an infinite priced game graph.

\begin{definition}[Semantics]
The semantics of a PTG $\arena = (\locations, \clocks, \invariants,
\labels, \transitions, \prices, \allowbreak \goals)$ is given as a
priced game graph $\sem\arena=(\states, \timedactions,
\timedtransitions, \timedprices, \finalstates)$ where
 \begin{itemize}
 \item
$\states=\{(\location,\valuation)\in \locations\times V \mid
\valuation \in \invariants(\location)\}$ is the set of configurations
of the PTG;
 \item
$\timedactions = \Rpos \times \labels$ is the set of timed moves;
 \item
$\timedtransitions\colon \states \times \timedactions \to \states$ is
the transition function defined by
$(\location',\valuation')=\timedtransitions((\location,\valuation),
(t,a))$ if $\transitions(\location,a)=(\zeta,R,\location')$ 
such that $\valuation+t\in\zeta$,
$\valuation+t'\in\invariants(\location)$ for all $0\leq t'\leq t$, and
$\valuation'=(\valuation+t)[R:=0]$;
\item $\timedprices: \states \times \timedactions \to \Real$ is such that 
$\timedprices((\location,\valuation), (t,a)) = \prices(\location)\times
t+\prices(a)$; and 
 \item
   $\finalstates \subseteq \states$ is such that
   $(\location,\valuation) \in\finalstates$ iff
   $\location\in\goals$.
 \end{itemize}
 \end{definition}

 The concepts of a play, its cost, and strategies of players for a PTG
 $\arena$ is defined via corresponding objects for its semantic priced
 game graph $\sem\arena$.  In the previous section we introduced games
 with reachability-cost objective for priced game graphs.  We also
 study the following winning objectives for Player~1 in the context of
 priced timed games; the objective for Player~2 is the opposite.
\begin{enumerate}
\item \textbf{Constrained-price reachability}.  The constrained-price
  reachability objective $\ReachOb({\leq} K)$ is to keep the payoff
  within a given bound $K \in \Nat$.  Objectives
  $\ReachOb({\bowtie}K)$ for constrains ${\bowtie} \in \set{<, =, >,
    \geq}$ are defined analogously.
\item \textbf{Bounded-time reachability}.  Given constants $K, T \in
  \Nat$, the bounded-time reachability objective $\BReachOb(K, T)$ is
  to keep the payoff of the play less than or equal to $K$ while
  keeping the total time elapsed within $T$ units.
\item \textbf{Repeated reachability}.  For this objective, we consider
  slightly different semantics of the game where the play continues
  forever, and the repeated reachability objective $\RReachOb(\eta)$,
  $\eta \in \Rpos$ is to visit target locations infinitely often each
  time with a payoff in a given interval $[-\eta, \eta]$.
\end{enumerate}

In Section~\ref{sec:undecidability}, we sketch the proof of the
following negative result regarding the decidability of PTGs with
these objectives. This result is particularly surprising for
bounded-time reachability objective, since bounded-time restriction
has been shown to recover decidability in many related
problems~\cite{OW10,BDGORW13}.
\begin{theorem}
\label{key-thm}
Let $\arena$ be a priced timed game arena. 
The decision problems corresponding to the existence of winning strategy for
following objectives are undecidable:
\begin{enumerate}
\item $\ReachOb({\bowtie} K)$ objective for PTGs with two or more
  clocks and arbitrary prices;
\item $\BReachOb(K, T)$ objective for PTGs with five or more clocks;
  and prices 0,1;
\item $\RReachOb(\eta)$ objective for PTGs with three or more clocks
  and arbitrary prices.
\end{enumerate}
\end{theorem}

To recover decidability, we consider a subclass of one-clock PTGs.  In
this subclass, the set of clocks $\clocks$ is a singleton $\{x\}$, and
price-rates of the locations come from a doubleton set $\{p^-,p^+\}$
with $p^-<p^+$ two distinct elements of $\{-1,0,1\}$ (no condition is
made on the prices $\prices(a)$ of labels $a\in\Sigma$).  We call
these restricted games \emph{one-clock bi-valued priced timed games},
abbreviated as $\onePTG(p^-,p^+)$, or \oneBPTG if $p^-$ and $p^+$ do
not matter. All our results may easily be extended to the case where
$p^-$ and $p^+$ are taken from the set $\{-d,0,d\}$ with $d\in\N$. We
devote Section~\ref{sec:single-clock} to the proof of the following
decidability results.
\begin{theorem}\label{thm:overall}
  We have the following results:
  \begin{enumerate}
  \item $\oneBPTG$s are determined.
  \item The value of a $\oneBPTG$ can be computed in pseudo-polynomial
    time.
  \item Given that a \oneBPTG has a finite value, an
    $\varepsilon$-optimal strategy for Player~1 can be computed in
    pseudo-polynomial time.
  \item Aforementioned complexities drop to polynomial time for
    $\onePTG(0,1)$ with prices of labels taken from $\N$.
\end{enumerate}
\end{theorem}

\section{Undecidability results}
\label{sec:undecidability}
In this section we provide a proof sketch of our undecidability result
(Theorem~\ref{key-thm}) by reducing the halting problem for two
counter machines (see \cite{Min67}) to the existence of a winning
strategy for Player~1 for the desired objective.  For all the three
objectives, given a two counter machine, we construct a PTG $\arena$
whose building blocks are the modules for instructions.  In these
reductions the objective of Player 1 is linked to a faithful
simulation of various increment, decrement, and zero-test instructions
of the machine by choosing appropriate delays to adjust the clocks to
reflect changes in counter values.  The goal of Player 2 is then to
verify the simulation performed by Player~1. Proofs of correctness of
the reductions, as well as more details can be found in the appendix.
    
\smallskip \paragraph{Constrained-price reachability objectives
  $\ReachOb({\bowtie} K)$} The result in the case $\ReachOb({\leq}K)$
is a consequence of the result in \cite{BCJ09}. Undecidability for
other comparison operators $\bowtie$ is a new contribution. We only
consider the objective $\ReachOb({=}1)$ in this section, since proofs
for other constraints are similar. Our reduction uses a PTG with two
clocks $x_1$ and $x_2$, arbitrary price-rates for locations and no
prices for labels. Each counter machine instruction (increment,
decrement, and test for zero value) is specified using a PTG
module. The main invariant in our reduction is that upon entry into a
module, we have that $x_1=\frac{1}{5^{c_1}7^{c_2}}$ and $x_2=0$ where
$c_1$ (respectively, $c_2$) is the value of counter $C_1$
(respectively, $C_2$).  We outline the simulation of a decrement
instruction for counter $C_1$ in Fig.~\ref{fig:equalone}.
\begin{figure}[tbp]
\centering
\scalebox{.9}{
 \begin{tikzpicture}[->,>=stealth',shorten >=1pt,auto,node distance=1cm,
       semithick,scale=0.8]
       \node[initial,initial text={}, player1] (lk) {$0$} ;
       \node()[above of=lk,node distance=5mm,color=gray]{$\ell_k$};
       \node[player2] at (3,0) (chk){$-1$} ;
       \node()[above of=chk,node distance=5mm,color=gray]{$\text{Check}$};

       \node[player1] at (6,0) (lk1){$0$ } ;  
       \node()[above of=lk1,node distance=5mm,color=gray]{Go};

       \node[fill=gray!20,rounded corners,fill opacity=0.9] at   (14,-0.65)(nxt){\texttt{$\ell_{k+1}$}}; 

       \node[player1] at (9, 0) (B) {$1$};
       \node()[above of=B,node distance=5mm,color=gray]{$\text{Abort}$};

       \node[accepting, player1] at (12, 0) (T1) {$0$};
       \node()[above of=T1,node distance=5mm,color=gray]{$T_1$};
       
       \path (lk) edge node {} (chk);
       \path (chk) edge node {$\set{x_2}$} (lk1);
       \path (lk1) edge node {$x_2{=}0$} (B);
       
       \draw[rounded corners] (lk1) -- (6, -0.65) --  node[above, near start] {$x_2{=}0$}(nxt);
       \path (B) edge node {$x_2 {>} 1$} (T1);
        
       \node[player1] (A) at (3, -1.5) {$-1$};
       \node()[below of=A,node distance=5mm,color=gray]{$L$};
       
       \node[player1] at (6, -1.5) (B){$-5$};
       \node()[below of=B,node distance=5mm,color=gray]{$M$};
       
       \node[player1] at (9, -1.5) (C){$2$};
       \node()[below of=C,node distance=5mm,color=gray]{$N$};

       \node[accepting, player1] at (12, -1.5) (T) {$0$};
       \node()[below of=T,node distance=5mm,color=gray]{$T$};
       
       \path (A) edge node {$x_1{=}1$} node[below] {$\set{x_1}$} (B);
       \path (B) edge node {$x_2 {=} 1$} node[below] {$\set{x_2}$} (C);
       \path (C) edge node {$x_2 {=} 1$} node[below]{$\set{x_2}$} (T);
     
       \path (chk) edge node[yshift=1mm] {$x_1{\leq} 1$} (A);

       \node[rotate=90,color=gray] at (1.8, -1.5) (N) {$\text{WD}_1$};

       \draw[dashed,draw=gray,rounded corners=10pt] (2.3,-1) rectangle (12.6,-2.4);
     \end{tikzpicture}}
\caption{Decrement module for the objection $\ReachOb({=}1)$}
\label{fig:equalone}
\end{figure}
Let us denote by $x_{old}=\frac{1}{5^{c_1}7^{c_2}}$ the value of $x_1$
while entering the module.  At the location $\ell_{k+1}$ of the
module, $x_1 = x_{new}$ should be $5x_{old}$ to correctly decrement
counter $C_1$.  At location $\ell_k$, Player 1 spends a
non-deterministic amount of time $t_k= x_{new} - x_{old}$ such that
$x_{new} = 5x_{old}+\varepsilon$. To correctly decrement $C_1$,
$\varepsilon$ should be 0, and $t_k$ must be
$\frac{4}{5^{c_1}7^{c_2}}$.  At location $\text{Check}$, Player~2
could choose to go to Go (in order to continue the simulation of the
machine) or go to the widget $\text{WD}_1$, if he suspects that
$\varepsilon \neq 0$. If Player~2 spends time $t>0$ in the location
$\text{Check}$ before proceeding to $\text{Go}$, then Player~1 can
enter the location Abort (to abort the simulation), rather than going
to $\ell_{k+1}$. Player 1 spends $1+t$ time in location Abort and
reaches a target $T_1$ with cost $1$ (and thus achieve his
objective). However, if $t=0$ then entering location Abort will make
the cost to be greater than $1$ (which is losing for Player 1). If
Player~2 decides to enter widget $\text{WD}_1$, then the cost upon
reaching the target in the widget $\text{WD}_1$ is $1+\varepsilon$
which is $1$ iff $\varepsilon=0$.

\smallskip \paragraph{Bounded-time reachability objective} We sketch
the reduction for objective $\BReachOb(K, T)$. Our reduction uses a
PTG with price-rates 0 or 1 on locations, and zero prices on labels,
along with five clocks $x_1,x_2,z,a,b$.  On entry into a module for
the $(k+1)$th instruction, we always have one of the two clocks
$x_1,x_2$ with value $\frac{1}{2^{k+c_1}{3^{k+c_2}}}$ and other is
$0$.  Clock $z$ keeps track of the total time elapsed during
simulation of an instruction: we always have $z=1-\frac{1}{2^k}$ at
the end of simulating $k$th instruction. Thus, time $\frac{1}{2}$ is
spent simulating the first instruction, $\frac{1}{4}$ for the second
instruction and so on, so that the total time spent in simulating the
main modules is less than $1$. The main challenge here is to ensure
that only a bounded time is spent along the entire simulation, along
with updating the counter values correctly. Clocks $a,b$ are used for
rough work.  For instance, if the $(k+1)$th instruction $\ell_{k+1}$
is an increment of $C_1$, and we have $x_1 =
\frac{1}{2^{k+c_1}3^{k+c_2}}$, while $a=b=x_2=0$, and
$z=1-\frac{1}{2^k}$, then at the end of the module simulating
$\ell_{k+1}$, we want $x_2= \frac{1}{2^{k+1+c_1+1}3^{k+1+c_2}}$ and
$x_1 = 0$ and $z = 1 - \frac{1}{2^{k+1}}$.
 
\smallskip \paragraph{Repeated reachability objective} Finally, we consider the
repeated reachability objective $\RReachOb(\eta)$.  Our reduction uses
a PTG with 3 clocks, and arbitrary price-rates, but zero prices for
labels.  On entry into a module, we have $x_1 =
\frac{1}{5^{c_1}7^{c_2}}$, $x_2 = 0$ and $x_3=0$, where $c_1,c_2$ are
the values of $C_1$ and $C_2$.  Fig.~\ref{fig:repeated} shows module
to simulate decrement $C_1$.
\begin{figure}[tbp]
\centering
\scalebox{.9}{
\begin{tikzpicture}[->,>=stealth',shorten >=1pt,auto,node distance=1cm,
  semithick,scale=0.9]
  \node[initial,initial text={}, player1] at (0,-.5) (lk) {$0$ } ;
   \node()[above of=lk,node distance=5mm,color=gray]{$\ell_k$};

  \node[player2] at (3,-.5) (chk){$0$} ;
     \node()[above of=chk,node distance=5mm,color=gray]{$\text{Check}$};

  \node () [below right of=chk,node distance=6.5mm,xshift=-3mm] {$[x_3=0]$};

   \node[fill=gray!20,rounded corners,fill opacity=0.9] at (6,-.5)(lk1){$\ell_{k+1}$};

  \node[ player2] at (0,-2)(A) {$-1$};
   \node()[above of=A,node distance=5mm,color=gray]{$A$};

  \node[player2] at (2.5,-2) (B){$4$};
   \node()[above of=B,node distance=5mm,color=gray]{$B$};

  \node[player2] at (5,-2) (C){$-3$};
   \node()[above of=C,node distance=5mm,color=gray]{$C$};

  \node[player2] at (7,-2) (D) {$0$};
   \node()[above of=D,node distance=5mm,color=gray]{$D$};

  \node[player2] at (9,-2) (E) {$0$};
   \node()[above of=E,node distance=5mm,color=gray]{$E$};

  \node[player2,accepting] at (11,-2) (F) {$0$} ;
   \node()[above of=F,node distance=5mm,color=gray]{$F$};

  \node () [below of=C,node distance=7mm] {$x_3=0$};
    
\path (lk) edge node {$x_1 {\leq} 1$} node[below] {$\set{x_3}$}(chk);
\path (chk) edge node[below] {$\set{x_2}$} (lk1);
\path (chk) edge node {} (A);

    \path (A) edge node {$ x_2 {=} 1$} node[below]{$\set{x_2}$} (B);
    \path (B) edge node {$x_1 {=} 2$} node[below]{$\set{x_1}$} (C);
    \path (C) edge node {$x_1 {=} 1$} node[below]{$\set{x_1}$} (D);
    \path (D) edge node {$x_2 {=} 2$} node[below]{$\set{x_2}$} (E);
    \path (E) edge node {$x_3 {=} 3$} node[below]{$\set{x_3}$} (F);
    \draw[rounded corners] (F) -- (11, -3) -- (0, -3) --  (A);
           
    \node[rotate=90,color=gray] at (-.8, -2) (N) {$\text{WD}_1$};
    \draw[dashed,draw=gray,rounded corners=10pt] (-.5,-3.2) rectangle (11.5,-1.25);

 \end{tikzpicture}}
\caption{Decrement module for Repeated reachability objective.}
\label{fig:repeated}
\end{figure}
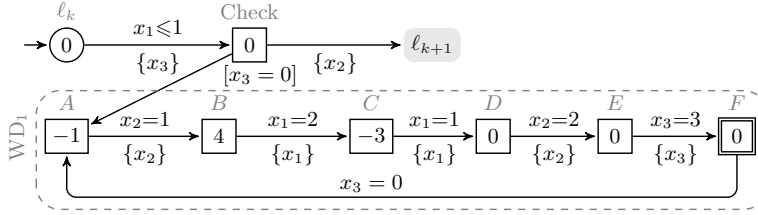
Location $\ell_k$ is entered with $x_1=\frac{1}{5^{c_1}7^{c_2}}, x_2 =
0$ and $x_3=0$.  To correctly decrement $C_1$, Player 1 should choose
a delay of $\frac{4}{5^{c_1}7^{c_2}}$ in location $\ell_k$.  At
location Check, no time can elapse because of the invariant. If Player
1 makes an error, and delays $\frac{4}{5^{c_1}7^{c_2}}+\varepsilon$ at
$\ell_k$ ($\varepsilon \neq 0$) then Player 2 can jump in widget
$\text{WD}_1$.  The cost of going from location $A$ to $F$ is
$\varepsilon$; each time we come back to $A$, the clock values with
which $A$ was entered are restored.  Clearly, if $\varepsilon \neq 0$,
Player 2 can incur a cost that is not in $[-\eta, \eta]$ by taking the
loop from $A$ to $F$ a large number of times.

\section{One-clock bi-valued priced timed games}
\label{sec:single-clock}
This section is devoted to the proof of
Theorem~\ref{thm:overall}. First of all, let us assume that all
$\oneBPTG$s $\arena$ we consider are \emph{bounded}, i.e., that there
is a global invariant in every location, of the form $x\leq M_K$
(where $M_K$ denotes the greatest constant appearing in the clock
guards and invariants of $\arena$). This restriction comes w.l.o.g
since every $\oneBPTG$ arena can be made bounded with a polynomial
algorithm.\footnote{By introducing auxiliary states in order to reset
  the clock $x$ at every time unit once its value goes beyond $M_k$.
  The polynomial complexity holds only for one-clock PTGs.}

Our proof of Theorem~\ref{thm:overall} is based on an extension of the
classical notion of regions in timed automata, in the spirit of the
regions introduced to define the corner point abstraction
\cite{BBL-fmsd08}.  Indeed, to take the price into account,
$\varepsilon$-optimal strategies do not take uniform decisions on the
classical regions.  That is why we need to subdivide each classical
region into three parts: two small parts around the corners of the
region (that we will call \emph{borders} in the following, considering
our one-clock setting), and a big part in-between.  We will show that
considering only strategies that never jump into those big parts is
sufficient (Lemma~\ref{lem:uniform}).  Lemma~\ref{lem:convergent},
later, shows a stronger result that one can restrict attention to
strategies that play closer and closer to the borders of the regions
as time elapses.  Finally, we combine these results to show that a
finite abstraction of \oneBPTG{s} is sufficient to compute the value
as well as $\varepsilon$-optimal strategies
(Lemma~\ref{lem:translation}).
This not only yields the desired result, but also provides us further
insight into the shape of $\varepsilon$-optimal strategies for both
players.

\subsection{\texorpdfstring{Reduction to $\eta$-region-uniform
     strategies}{Reduction to region-uniform strategies}}
 Since we only consider one-clock PTGs, we need not consider the
 standard Alur-Dill regional equivalence relation.  Instead, we
 consider special region equivalence relation characterized by the
 intervals with constants appearing in guards and invariants of
 $\arena$ inspired by Laroussinie, Markey, and Schnoebelen
 construction~\cite{LarMar04}.  Let $0{=}M_0{<}M_1{<}\cdots {<}M_K$ be
 the integers appearing in guards and invariants of
 $\arena$. 
 We say that two valuations $\valuation, \valuation' \in \Rplus$ are
 region-equivalent (or lie in the same region), and we write
 $\valuation \sim \valuation'$, if for every $k\in\{0,\ldots,K\}$,
 $\valuation\leq M_k$ iff $\valuation'\leq M_k$, and $\valuation\geq
 M_k$ iff $\valuation'\geq M_k$.  We define the set of regions to be
 the set of equivalence classes of $\sim$.
We extend the equivalence relation $\sim$ from valuations to configurations in a
straightforward manner. 
We also generalize the regional equivalence relation to the plays.
For two (finite or infinite) plays
$r=\seq{(\location_0,\valuation_0),(t_0,a_0),\ldots}$ and
$r'=\seq{(\location'_0,\valuation'_0),(t'_0,a'_0),\ldots}$ we say that
$r\sim r'$ if the lengths of $r$ and $r'$ are equal, and they define
sequences of regional equivalent states (i.e.,
$(\location_i,\valuation_i)\sim (\location'_i,\valuation'_i)$ for all
$i\geq 0$) and follow equivalent timed actions (i.e., $a_i=a'_i$ and
$\valuation_i+t_i\sim \valuation'_i+t'_i$ for all $i\geq 0$).  We also
consider a refinement of region equivalence relation that we call the
$\eta$-region equivalence relation, and we write $\sim_\eta$, for a
given $\eta\in(0,\frac 1 3)$.  Intuitively,
$\valuation\sim_\eta\valuation'$ if both valuations are close or far
from any borders of the regions, with respect to the distance $\eta$.
\begin{definition}[$\eta$-regions]
  For valuations $\valuation, \valuation'\in\Rpos$ we say that
  $\valuation \sim_\eta \valuation'$ if $\valuation \sim \valuation'$
  and for every $k\in\{0,\ldots,K-1\}$, $|\valuation-M_k|\leq\eta$ iff
  $|\valuation'-M_k|\leq\eta$, and $\valuation\geq M_K-\eta$ iff
  $\valuation'\geq M_K-\eta$.  We assume the natural order
  $\leqinterval$ over $\eta$-regions by their lower bounds.  We call
  $\eta$-regions the equivalence classes of $\sim_\eta$.  We also
  extend the relation $\sim_\eta$ to configurations and runs.
\end{definition}
 
For instance, if $M_1=2$ and $M_2=3$, the set of $\eta$-regions is
given by $ \{\{0\}, (0,\eta], (\eta,2-\eta), [2-\eta,2), \{2\},
(2,2+\eta], (2+\eta,3-\eta), [3-\eta,3), \{3\}, (3,+\infty)\}$. We
next introduce the strategies of a restricted shape with the
properties that
they depend only on the $\eta$-region abstraction of runs;
their decision is uniform over each $\eta$-region; and 
    they play $\eta$-close to the borders of the regions.

\begin{definition}[$\eta$-region uniform strategies]
  Let $\eta\in(0,\frac 1 3)$ be a constant. A strategy
  $\strategy\in\minstrategies\cup\maxstrategies$ is said to be
  \emph{$\eta$-region-uniform}
  if
  \begin{itemize}
  \item for all finite run $r\sim_\eta r'$ ending respectively in
    $(\location,\valuation)$ and $(\location,\valuation')$ (in
    particular $\valuation\sim_\eta\valuation'$) we have
    $\valuation+\tadelay{\strategy(r)}\sim_\eta
    \valuation'+\tadelay{\strategy(r')}$ and
    $\taaction{\strategy(r)}=\taaction{\strategy(r')}$;
  \item for every finite run $r$ ending in $(\location,\valuation)$,
    if 
    $\valuation+\tadelay{\strategy(r)}\in(M_k,M_{k+1})$, we have
    $\valuation+\tadelay{\strategy(r)}\in (M_k,M_k+\eta]\cup
    [M_{k+1}-\eta,M_{k+1})$.
  \end{itemize}
  We write $\uniformminstrategies^\eta$ and $\uniformmaxstrategies^\eta$
  for the set of $\eta$-region-uniform strategies for
  Players~1 and~ 2.
  We also define upper-value $\uniformuppervalue^\eta(s)$ when both players are
  restricted to use only $\eta$-region-uniform strategies. 
  Formally, 
  \[ 
  \uniformuppervalue^\eta(s) =
  \inf_{\minstrategy\in\uniformminstrategies^\eta}
  \sup_{\maxstrategy\in\uniformmaxstrategies^\eta}
  \Weight(\outcomes(s,\minstrategy,\maxstrategy)), \text{ for all
  $s\in\states$}.
\]
\end{definition}

\begin{example}
  Consider PTG $\arena_1$ shown in Fig.~\ref{fig:CE-0+1-1} (that is
  not a $\oneBPTG$ since there are three distinct price-rates). A
  strategy of Player~2 is entirely described by the time spent in the
  initial location with initial valuation $0$. For example, Player~2
  can choose to delay $1/2$ time units before jumping in the next
  location. Indeed, the lower and upper value of the game is $-\frac 1
  2$. However, this strategy is not $\eta$-region-uniform. Instead, an
  $\eta$-region-uniform strategy will delay $t$ time units with $t\in
  [0,\eta]\cup [1-\eta,1]$. Hence, the upper value when players can
  only use $\eta$-region-uniform strategies is equal to $-1$.
\end{example}
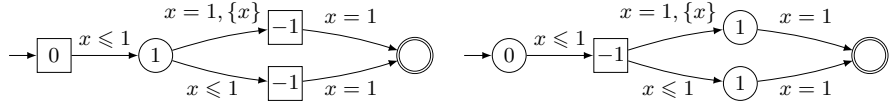
\begin{figure}[tbp]
  \begin{minipage}{0.495\linewidth}
    \scalebox{.9}{
      \begin{tikzpicture}[node distance=2cm,auto,->,>=latex] 
        \node[player2,initial,initial
        text=](1){\makebox[0mm][c]{$0$}};%
        \node[player1](2)[right of=1,node
        distance=1.5cm]{\makebox[0mm][c]{$1$}};%
        \node[player2](4)[below right
        of=2,xshift=5mm,yshift=1cm]{\makebox[0mm][c]{$-1$}};%
        \node[player2](5)[above right
        of=2,xshift=5mm,yshift=-1cm]{\makebox[0mm][c]{$-1$}};%
        \node[player1,double](7)[below right
        of=5,xshift=5mm,yshift=1cm]{};%
        
        \path (1) edge node[above]{$x\leq 1$} (2)%
        (2) edge[bend left=5] node[above,node
        distance=0mm,xshift=-1mm,yshift=1mm]
        {$x=1,\{x\}$} (5) %
        edge[bend right=5] node[below,xshift=-1mm]{$x\leq 1$} (4)%
        (4) edge[bend right=5] node[below]{$x=1$} (7) %
        (5) edge[bend left=5] node[above,yshift=1mm]{$x=1$} (7);%
      \end{tikzpicture}}
  \end{minipage}
  \begin{minipage}{0.495\linewidth}
    \scalebox{.9}{\begin{tikzpicture}[node
        distance=2cm,auto,->,>=latex] 
        \node[player1,initial,initial
        text=](1){\makebox[0mm][c]{$0$}};%
        \node[player2](2)[right of=1,node
        distance=1.5cm]{\makebox[0mm][c]{$-1$}};%
        \node[player1](4)[below right
        of=2,xshift=5mm,yshift=+1cm]{\makebox[0mm][c]{$1$}};%
        \node[player1](5)[above right
        of=2,xshift=5mm,yshift=-1cm]{\makebox[0mm][c]{$1$}};%
        \node[player1,double](7)[below right
        of=5,xshift=5mm,yshift=+1cm]{};%
        
        \path (1) edge node[above]{$x\leq 1$} (2)%
        (2) edge[bend left=5] node[above,node
        distance=0mm,xshift=-1mm,yshift=1mm]
        {$x=1,\{x\}$}(5)%
        edge[bend right=5] node[below,xshift=-1mm]{$x\leq 1$} (4)%
        (4) edge[bend right=5] node[below]{$x=1$} (7) %
        (5) edge[bend left=5] node[above,yshift=1mm]{$x=1$} (7);
      \end{tikzpicture}}
  \end{minipage}
  \caption{The value in the left-side one-clock PTG $\arena_1$ with
    price-rates in $\{-1,0,1\}$ is $-\frac 1 2$, while the value in
    the right-side PTG $\arena_2$ is $\frac 1 2$.}
  \label{fig:CE-0+1-1}
\end{figure}

Contrary to this example, the next lemma shows that, in $\oneBPTG$s,
the upper value of the game increases when we restrict ourselves to
$\eta$-region-uniform strategies.  Intuitively, every cost that
Player~2 can secure with general strategies, it can also secure it
with $\eta$-region-uniform strategies against $\eta$-region-uniform
strategies of Player~1.
\begin{lemma}\label{lem:uniform} 
$\uppervalue(s)\leq \uniformuppervalue^\eta(s)\,$, for every 
$\oneBPTG$ $\arena$, $s\in\states$ and $\eta\in(0,\frac 1 3)$, 
\end{lemma}

\subsection{\texorpdfstring{Reduction to $\eta$-convergent
    strategies}{Reduction to convergent strategies}}

A similar result concerning the lower values of the games can be shown
in case of $\eta$-region-uniform strategies.  In subsequent proofs, we
need a stronger result to avoid situations detailed in
Example~\ref{ex:convergence}, where player 2 needs infinite precision
to play incrementally closer to borders (as well as an infinite
memory).  For this reason, we restrict the shape of strategies to
force them to play at distance $\frac \eta{2^n}$ of borders when
playing the $n$th round of the game.  The slight asymmetry in the
definitions for the two players is exploited in proving subsequent
results.
\begin{definition}[$\eta$-convergent strategies]
  Let $\eta\in(0,\frac 1 3)$ be a constant. A strategy
  $\strategy\in\minstrategies\cup\maxstrategies$ is said to be
  \emph{$\eta$-convergent} if $\strategy$ is $\eta$-region-uniform and
  for all finite run $r$ of length $n$ ending in
  $(\location,\valuation)$:
  \begin{itemize}
  \item if $\strategy\in\minstrategies$, there exists $k$ such that
    either $|\valuation+\tadelay{\strategy(r)}-M_k|\leq \frac
    \eta{2^{n+1}}$, or $\tadelay{\strategy(r)}=0$ and $\valuation\in
    (M_k+\frac \eta{2^{n+1}},M_k+\eta]$;
  \item if $\strategy\in\maxstrategies$, there exists $k$ such that
    either $\valuation+\tadelay{\strategy(r)}\in\{M_k+\frac
    \eta{2^{n+1}}\}\cup [M_k-\frac \eta{2^{n+1}},M_k)$, or
    $\tadelay{\strategy(r)}=0$ and $\valuation\in (M_k+\frac
    \eta{2^{n+1}},M_k+\eta]$.
  \end{itemize}
  We let $\convergentminstrategies^\eta$ and
  $\convergentmaxstrategies^\eta$ be respectively the set of
  $\eta$-convergent strategies for Player~1 and Player~2, and we
  define, for every configuration $s\in\states$,
  $\convergentlowervalue^\eta(s) =
  \sup_{\maxstrategy\in\convergentmaxstrategies^\eta}
  \inf_{\minstrategy\in\convergentminstrategies^\eta}
  \Weight(\outcomes(s,\minstrategy,\maxstrategy))\,.$
\end{definition}

\begin{example}\label{ex:convergence}
  Consider the $\oneBPTG$ $\arena_3$ composed of a vertex per player,
  on top of the target vertex. In its vertex, having price-rate $0$,
  Player~1 must choose between going to the target vertex, or going to
  the vertex of Player~2 by resetting clock~$x$. In its vertex, having
  price-rate $-1$, Player~2 must go back to the vertex of Player~1,
  with a guard $x>0$: hence, Player~2 would like to exit as soon as
  possible, but because of the guard, he must spend some time before
  exiting. If Player~2 plays according to a finite-memory strategy,
  there must be a bound $\varepsilon$ such that Player~2 always stays
  in his state for a duration bounded from below by $\varepsilon$, and
  Player~1 can exploit it by letting the game continue for an
  arbitrarily long time to achieve an arbitrarily small
  payoff. 
  On the other hand, if Player~2 plays an infinite-memory
  $\eta$-convergent strategy by staying in his location for a duration
  $\varepsilon/{2^n}$ in his $n$-th visit to its location, Player~2
  ensures a payoff $-\varepsilon$ for an arbitrarily small
  $\varepsilon>0$, resulting in the value $0$ of the game.
\end{example}

It is clear from the previous example that Player~2 needs
infinite-memory strategies to optimize his objective.
The following lemma formalizes our intuition that the lower value of the game
decreases when we restrict ourselves to $\eta$-convergent strategies.
Intuitively, every cost that Player~1 can secure with general strategies, it can
also secure it with $\eta$-convergent strategies against an $\eta$-convergent
strategy of Player~2. 
\begin{lemma}\label{lem:convergent} 
$\convergentlowervalue^\eta(s)\leq \lowervalue(s)\,$, for every 
$\oneBPTG$ $\arena$, $s\in\states$ and $\eta\in(0,\frac 1 3)$.
\end{lemma}
Observe that this lemma fails to hold when location price-rates can take more
than two values as exemplified by arena $\arena_2$ in Fig.~\ref{fig:CE-0+1-1}. 
It shows a game with three distinct prices with lower and upper value
equal to $1/2$. However, when restricted to $\eta$-convergent
strategies, the lower value equals $1$.

Our next goal is to find a common bound being both a lower bound on
$\convergentlowervalue^\eta(s)$ and an upper bound on
$\uniformuppervalue^\eta(s)$ by studying the value of a reachability-cost game 
on a finitary abstraction of \oneBPTG{s}. 

\subsection{\texorpdfstring{Finite abstraction of \oneBPTG{s}}{Finite
    abstraction of 1BPTGs}}
We now construct a finite price game graph $\abstrarena$ from any
\oneBPTG $\arena$, as a finite abstraction of the infinite weighted
game $\sem\arena$, based on $\eta$-regions.
Since we have learned that $\eta$-region-uniform strategies suffice,
we limit ourselves to playing at a distance at most $\eta$ from the
borders of regions.  Observe that only $\eta$-regions close to the
borders are of interest, and moreover $\eta$-regions after the maximal
constant $M_K$ are not useful since $\arena$ is bounded.  Let
$\intervals^\eta_\arena$ be the set of remaining ``useful''
$\eta$-regions.  For example, if constant appearing in the PTG are
$M_1=2$ and $M_2=3$, we have $\intervals^\eta_\arena=\{\{0\},
(0,\eta], [2-\eta,2), \{2\}, (2,2+\eta], [3-\eta,3), \{3\}\}$.  We
next define the \emph{delay} between two such $\eta$-regions
$I\leqinterval J$, denoted by $d(I,J)$, as the closest integer of
$q'-q$, where $q$ (respectively, $q'$) is the lower bound of interval
$I$ (respectively, $J$).  For example, $d((2,2+\eta],[3-\eta,3))=1$
and $d(\{0\},[2-\eta,2))=2$.

\begin{definition}
  For every \oneBPTG $\arena$ we define its border abstraction as a finite
  priced game graph 
  $\abstrarena=(\vertices=\vertices_1\uplus
  \vertices_2,\actions,\edges,\edgeweights,\finalvertices)$ where:
  \begin{itemize}
  \item
    $\vertices_{i}=\{(\location,I)\mid \location \in
    \locations_i, I\in\intervals^\eta_\arena,
    I\subseteq\invariants(\location)\}$ for $i\in\{1,2\}$;
  \item 
    $\actions = \intervals^\eta_\arena\times\labels$;
  \item 
    $E$ is the set of tuples
    $((\location,I),(J,a),(\location',J'))$ such that $I \leqinterval
    J$ and for all $I\leqinterval K \leqinterval J$ we have $K
    \subseteq \invariants(\location)$ and $J \subseteq \zeta$ and
    $J'=J[R:=0]$ with
    $(\zeta,R,\location')=\transitions(\location,a)\}$;
  \item 
    $\edgeweights((\location,I),(J,a),(\location',J'))=
    \prices(\location) \times d(I,J) + \prices(a)$; and 
  \item 
    $\finalvertices = \{(\location,I)\mid \location\in\goals,
    I\in\intervals^\eta_\arena\}$.
  \end{itemize}
\end{definition}
In a border abstraction game $\abstrarena$, the meaning of action $(J,a)$ is
that the player wants to let time elapse until it reaches the $\eta$-region
$J$, then playing label $a$. 
It simulates any timed move $(t,a)$ with $t$ any delay reaching a point in $J$. 

\begin{example}
  Consider the border abstraction of the $\oneBPTG$ of Fig.~\ref{fig:BPTG} shown
  in Fig.~\ref{fig:BPTG-finite}.
  Observe that we depict only a succinct representation of the real abstraction,
  since we only show the reachable part of the game from $(\location_1,0)$, and
  we have removed multiple edges (introduced due to label hiding) and kept only
  the most useful ones for the corresponding player. 
  For example, consider the location $(\location_5,\{0\})$.
  There are edges labelled by $(J,a)$ for every interval $J\in
  \intervals^\eta_\arena$, all directed to 
  $(\location_4,\{0\})$ due to a reset being performed there.
  We only show the best possible edge---the one with lowest price---since
  location $\location_5$ belongs to Player~1, who seeks to minimise cost. 
  Each vertex contains the $\eta$-region it represents.
\begin{figure}[tbp]
  \centering
  \scalebox{.8}{
  \begin{tikzpicture}[node distance=1cm,auto,->,>=latex]
    \node[player1](1a){\makebox[0mm][c]{$\{0\}$}};

    \node[player1,above of=1a,node distance=1.5cm,xshift=1.5cm](2){\makebox[0mm][c]{$\{0\}$}};

    \node[player2,right of=2,node distance=2.5cm](3a){\makebox[5mm][c]{\footnotesize{$[0,\eta]$}}};
    \node[player2,right of=3a,node distance=2cm](3b){\makebox[9mm][c]{\footnotesize{$[1\mathrm{-}\eta,1)$}}};
    \node[player2,right of=3b,node distance=1.5cm](3c){\makebox[9mm][c]{\footnotesize{$[1,1\mathrm{+}\eta]$}}};
    \node[player2,right of=3c,node distance=1.5cm](3d){\makebox[9mm][c]{\footnotesize{$[2\mathrm{-}\eta,2]$}}};

    \node[player2,below of=1a,node distance=1.5cm,xshift=1cm](4a){\makebox[0mm][c]{$\{0\}$}};
    \node[player2,right of=4a](4b){\makebox[5mm][c]{\footnotesize{$(0,\eta]$}}};
    \node[player2,right of=4b,node distance=1.5cm](4c){\makebox[9mm][c]{\footnotesize{$[1\mathrm{-}\eta,1)$}}};
    \node[player2,right of=4c,node
    distance=1.3cm](4d){\makebox[0mm][c]{$\{1\}$}};

    \node[player1,right of=4d,node distance=4cm](5a){\makebox[0mm][c]{$\{0\}$}};

    \node[draw,rounded rectangle,minimum size=5mm,below of=3d,node
    distance=1.5cm,xshift=1.5cm,double](6){};

    \path (1a) edge node[above left]{$0$} (2)
     (1a) edge[bend right=10] node[left]{$1$} (4a)
     (1a) edge node[left,yshift=-1mm]{$1$} (4b)
     (1a) edge node[above right]{$2$} (4c)
     (1a) edge[bend left=15] node[above right]{$2$} (4d)
     (2) edge node[above]{$0$} (3a)
         edge[bend left=25] node[above right,yshift=-1mm,xshift=10mm]{$1$} (3b)
         edge[bend left=30] node[right,xshift=12mm]{$1$} (3c)
         edge[bend left=35] node[below,xshift=15mm,yshift=-2mm]{$2$} (3d)
     (3a) edge[out=-100,in=-150,loop] node[below left]{$0$} (3a)
          edge[bend right=10] node[below left]{$0$} (6)
     (3b) edge node[above]{$0$} (3a)
          edge[bend right=5] node[left,xshift=-3mm]{$1$} (6)
     (3c) edge node[right,yshift=2mm]{$1$} (6)
     (3d) edge[bend left=5] node[right,yshift=1mm]{$1$} (6)
     (4d) edge node[above]{$0$} (5a)
     (4c) edge[bend right=18] node[above]{$0$} (5a)
     (4b) edge[bend right=28] node[above]{$-1$} (5a)
     (4a) edge[bend right=38] node[above]{$-1$} (5a)
     (5a) edge[bend left=45] node[below]{$1$} (4a)
     (5a) edge node[below right]{$3$} (6);

     \tikzset{color=gray,node distance=5mm}
     \node[left of=1a](){$\location_1$};
     \node[above left of=2](){$\location_2$};
     \node[right of=6](){$\location_6$};
     \node[below right of=5a](){$\location_5$};
     \draw[dashed,draw=gray,rounded corners=10pt] (.5,-2) rectangle (5.3,-1);
     \node at (.3,-1.5) {$\location_4$};
     \draw[dashed,draw=gray,rounded corners=10pt] (3.3,1) rectangle (9.8,2);
     \node at (10,1.5) {$\location_3$};
  \end{tikzpicture}}
\caption{Finite weighted game associated with the $\oneBPTG$ of
  Fig.~\ref{fig:BPTG}.}
  \label{fig:BPTG-finite}
\end{figure}
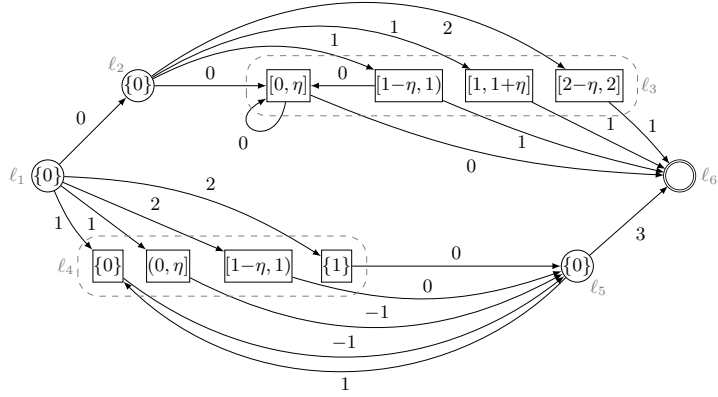
Thanks to Theorem~\ref{thm:optimal-strategy}, it is possible to
compute the optimal value as well as optimal strategies for both
players. 
Here, the value of state $(\location_1,0)$ is $1$, and an optimal strategy for
Player~1 is to follow action $(\{0\},a)$ (i.e., jump to $\location_2$
immediately), and then action $(\{1\},a)$ (i.e., to delay 1 time unit, before
jumping in $\location_3$). 
\end{example}

\begin{lemma}\label{lem:translation}
  Let $\arena$ be a \oneBPTG and $\abstrarena$ be its border
  abstraction.  Suppose that for all $0\leq k \leq K$ and
  $\location\in\locations$ we have that
  $\Value_{\abstrarena}((\location,\{M_k\}))$ is finite. Then, for
  all $\varepsilon>0$, there is $\eta>0$ s.t.
  $\uniformuppervalue^\eta_\arena((\location,M_k))-\varepsilon \leq
  \Value_{\abstrarena}((\location,\{M_k\})) \leq
  \convergentlowervalue^\eta_\arena((\location,M_k))+\varepsilon$.
\end{lemma}
Combining this result with Theorem~\ref{thm:optimal-strategy} we obtain the
following. 
\begin{corollary}\label{cor:pseudo}
  \oneBPTG{s} are determined and we can compute their
  values in pseudo-polynomial time. Moreover, in case the values are
  finite, $\varepsilon$-optimal strategies exist for both players:
  Player 2 may require infinite memory strategies, whereas finite
  memory is sufficient for Player 1. Finally, $\varepsilon$-optimal
  strategies can also be computed in pseudo-polynomial time.
\end{corollary}
\begin{proof}
  In case of infinite values
  $\Value_{\abstrarena}((\location,\{M_k\}))$, we can show directly
  that $\uppervalue_\arena((\location,M_k))=
  \Value_{\abstrarena}((\location,\{M_k\}))
  =\lowervalue_\arena((\location,M_k))\,.$
  Otherwise, let $\varepsilon>0$. By Lemma~\ref{lem:translation}, we
  know that there exists $\eta>0$ such that for every location
  $\location\in\locations$ and integer $0\leq k\leq K$: 
  \[
  \uniformuppervalue^\eta_\arena((\location,M_k))-\varepsilon\leq
  \Value_{\abstrarena}((\location,\{M_k\})) \leq
  \convergentlowervalue^\eta_\arena((\location,M_k))+\varepsilon\,.
  \]
  Moreover Lemma~\ref{lem:uniform} and \ref{lem:convergent} show that: 
  \[
  \convergentlowervalue^\eta((\location,M_k)) \leq
  \lowervalue((\location,M_k))\leq \uppervalue((\location,M_k))\leq
  \uniformuppervalue^\eta((\location,M_k))\,.
  \]
  Both inequalities combined permit to obtain
  \[ 
  \Value_{\abstrarena}((\location,\{M_k\}))-\varepsilon
  \leq \lowervalue((\location,M_k))\leq \uppervalue((\location,M_k))
  \leq
  \Value_{\abstrarena}((\location,\{M_k\}))+\varepsilon\,.
  \]
  Taking the limit when $\varepsilon$ tends to $0$, we obtain that
  $\lowervalue((\location,M_k)) = \uppervalue((\location,M_k)) =
  \Value_{\abstrarena}((\location,\{M_k\}))$. 
  Therefore, \oneBPTG are determined.  
  Moreover, in case of finite values, the proof of Lemma~\ref{lem:translation}
  permits to construct  $\varepsilon$-optimal $\eta$-region-uniform strategies
  $\minstrategy^*$ (with finite memory) and $\maxstrategy^*$ (which is
  moreover $\eta$-convergent).\qed
\end{proof}

\begin{figure}[t]
  \centering
  \scalebox{.9}{\begin{tikzpicture}[node distance=.5cm,auto,->,>=latex]
    \node[player1,initial,initial text=](1){$0$}; %
    \node[player2](2)[right of=1,node distance=3cm]{$0$}; %
    \node[player1](3)[above right of=2,xshift=2cm]{$1$};%
    \node[player1](4)[below right of=2,xshift=2cm]{$0$};%
    \node[player1](5)[right of=3, node distance=2cm]{$0$}; %
    \node[player1](6)[right of=4, node distance=2cm]{$1$}; %
    \node[player1,double](7)[below right of=5,xshift=2cm]{$0$};%
    
    \path %
    (1) edge node[above]{$x\leq 1, \{y\}$} (2) %
    (2) edge[bend left=5] node[above]{$y=0$} (3)%
    edge[bend right=5] node[below]{$y=0$} (4) %
    (3) edge node[above]{$x=1$} (5)%
    (4) edge node[above]{$x=1$} (6) %
    (5) edge[bend left=5] node[above]{$y=1$} (7) %
    (6) edge[bend right=5] node[below]{$y=1$} (7);
  \end{tikzpicture}}
  \caption{A two-clock PTG with prices of locations in $\{0,+1\}$ and
    value $1/2$}
  \label{fig:2clocks}
\end{figure}
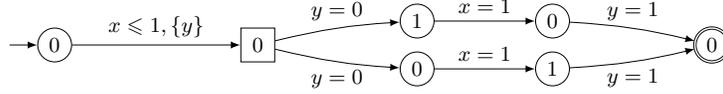

In the case of \oneBPTG{s}, the finite values are integers.  This
property fails if we allow more than one clock, as shows
Fig.~\ref{fig:2clocks} with a two-clock PTG with price-rates in
$\{0,1\}$ and optimal value $\frac 1 2$.  It also fails if we allow
more than two price-rates as was shown in Fig.~\ref{fig:CE-0+1-1}.
However for $\onePTG(0,1)$ with prices of labels in $\N$, the value of
the game is necessarily nonnegative disallowing the case $-\infty$.
The case $+\infty$ can be detected in polynomial time.  If the value
is not $+\infty$, the exact computation in the finite abstraction
$\abstrarena$ can be performed in polynomial time (see
\cite{priced-games} or \cite{DueIbs13}), resulting in a polynomial
algorithm for PTGs.  The sketch of Theorem~\ref{thm:overall} is now
complete.  Contrary to what it has been said in \cite{concur14}, our
results do not allow to ensure that the number of cutpoints is
polynomial in this subclass. It can indeed be exponential as shown in
\cite{FIS20}.  \todo{G: Removed `Notice that our proof shows that
  optimal value functions (as defined in
  \cite{BouLar06,Rut11,DueIbs13}) of such games have a polynomial
  number of line segments, and hence algorithms presented
  in~\cite{BouLar06,Rut11,DueIbs13} are indeed polynomial time.'}

\section{Conclusion}
We revisited games with reachability objective on PTGs with both
positive and negative price-rates.  We showed undecidability of all
classes of constrained-price reachability objectives with two or more
clocks. We also observed that adding bounded-time restriction does not
recover decidability, even with nonnegative prices. We also partially
answer the question regarding polynomial-time algorithm for one-clock
PTGs by showing that for a bi-valued variant the problem is in
pseudo-polynomial time. However, the existence of a polynomial-time
algorithm for multi-priced one-clock PTGs with nonnegative
price-rates, and the existence of algorithm for computing
$\varepsilon$-optimal strategies for PTGs with arbitrary number of
clocks remain open problems.

\clearpage
\appendix
\section{Motivation: A case-study from Project Cassting} 
Our work is partly motivated by possible applications to real case
study, like energy-aware houses, taken from the EU FP7 project
Cassting. It consists in houses equiped with solar panels and energy
storage capacities (e.g., with a water tank), that can produce some
energy, use or store it, and possibly sell it on a local grid, for use
by other houses. Priced timed games may permit to model such house,
and strategy synthesis allows us to build optimal controller for the
house, aiming at energy and cost savings: player 1 models the house
and wants to minimize its cost, whereas player 2 models other houses
and the environment, modelling the worst possible situation. Pricing
policy implies that selling energy on the grid is more profitable
during the day than the night, whereas it is more profitable to use
energy from the grid during the night than the day. Moreover, weather
conditions (sun or clouds, e.g.) imply that the solar panels are not
constantly producing energy. Hence, we can model the situation with
three independant phases: sunny day, cloudy day and night. In each of
these phases, we can suppose that only two rates are indeed
available. During a sunny day, selling energy will be rewarded
$\alpha$ euros per time unit (assuming that the energy production or
consumption is constant), whereas consumption or storing of energy
does not cost anything. During a cloudy day, solar panels are off:
because of storage capacities, it is however possible to sell energy
at the same cost than previously, but the consumption may now cost
$\alpha$ euros per time unit. Finally, during the night, selling
energy rewards $\beta$ euros per time unit, whereas consuming costs
$\beta$ euros per time unit. Notice that in each of the phases, we can
use bi-valued priced timed games to model the arena.
 
\section{Detailed decidability proofs}

\subsection{Proof of Lemma~\ref{lem:uniform}}
\label{app:lem:uniform}
With respect to rechability of the target locations,
($\eta$-)equivalent plays are indistinguishable. Moreover, with
respect to weights, we show that plays may avoid regions
$(M_k+\eta,M_{k+1}+\eta)$ without loss of generality. Formally, a play
$r=(\location_0,\valuation_0),(t_0,a_0),\ldots$ is said to \emph{stay
  $\eta$-close to borders} if for all $i\geq 0$, there exists $k$ such
that $|\valuation_i+t_i - M_k|\leq \eta$. Before proving
Lemma~\ref{lem:uniform}, we first study more precisely the
relationship between general plays and plays that stay $\eta$-close to
borders. More precisely, we now explain how to construct, from any
finite play $r$, a play $r^+$ such that
\begin{enumerate*}[label=(\emph{\roman*})]
\item $r^+$ stays $\eta$-close to borders; 
\item $r$ and $r^+$ are region-equivalent; and
\item $\Weight(r)\leq \Weight(r^+)$.
\end{enumerate*}
Intuitively, the idea is to consider the steps of play $r$ and to
shift them towards one of the closest borders, unless the current
step is already $\eta$-close to a border: when the current location
has price $p^+$, since we want the weight of $r^+$ to be greater than
or equal to the weight of $r$, we will spend more time in this
location, hence shifting the step `to the right', and symetrically in
case of a location of price $p^-$. Moreover, the construction will
trivially verify that if $r'$ and $r$ are two plays that coincide on
their prefix of length $n$, i.e., $r'[n]=r[n]$, then
$(r')^+[n]=r^+[n]$.

The construction is by induction on the length of the
play. Henceforth, let $r=(\location_0,\valuation_0),(t_0,a_0),\ldots,
(\location_{i+1},\valuation_{i+1})$ and suppose that $\valuation_0$ is
$\eta$-close to some $M_k$. We construct a play
$r^+=(\location^+_0,\valuation^+_0),(t^+_0,a^+_0),\ldots,
(\location^+_{i+1},\valuation^+_{i+1})$. In case $i+1=0$, we simply
let $\location^+_0=\location_0$ and
$\valuation^+_0=\valuation_0$. Otherwise, we consider the play $r^+$
constructed by induction up to its configuration of index $i$, and we
now explain how to construct the next timed action
$(t^+_{i},a^+_{i})$. First, we let $a^+_{i}=a_{i}$. Then, we
distinguish between three cases:
\begin{itemize}
\item if there exists $k$ such that $|\valuation_i+t_i-M_k|\leq \eta$,
  then we let $t^+_i= \valuation_i+t_i-\valuation^+_i$;
\item if there exists $k$ such that $\valuation_i+t_i\in
  (M_k+\eta,M_{k+1}-\eta)$ and $\prices(\location_i)=p^-$, then we let
  $t^+_i=\max(M_{k}+\eta-\valuation^+_i,0)$. Indeed, in this case, we
  want to spend as little time as possible in the location in order to
  ensure $\Weight(r)\leq \Weight(r^+)$;
\item if there exists $k$ such that $\valuation_i+t_i\in
  (M_k+\eta,M_{k+1}-\eta)$ and $\prices(\location_i)=p^+$, then there
  are two cases. In case $\location_i\in\maxlocations$ or $t_i>0$ or
  $\valuation^+_i\neq M_k+\eta$, we let
  $t^+_i=M_{k+1}-\eta-\valuation^+_i$. The intution is that we try to
  spend as much time as possible in the location to enforce
  $\Weight(r)\leq \Weight(r^+)$. Otherwise, i.e., if
  $\location_i\in\minlocations$ and $t_i=0$ and $\valuation^+_i=
  M_k+\eta$, then we let $t^+_i=0$. Because $t_i=0=t^+_i$, we will
  still ensure that $\Weight(r)\leq \Weight(r^+)$ in that case.
\end{itemize}
Notice that this definition is purely syntactic, and we will show
later that all $t_i^+$'s are non-negative, ensuring that we are indeed
constructing a play of the game. We then let
$\location^+_{i+1}=\location_{i+1}$, and valuation
$\valuation^+_{i+1}$ is defined according to the semantics of the
game.

An example of construction of $r^+$ is given in Fig.~\ref{fig:ex1},
for a play $r$ without reset. We have supposed that the sequence of
borders is $0, 1, 3$ and $5$. 

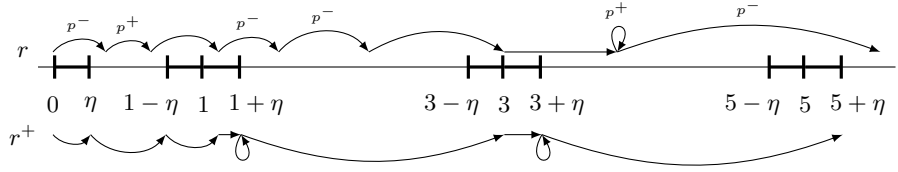
\begin{figure}[tbp]
  \centering
  \begin{tikzpicture}
    \draw(-.2,0) -- (11.2,0);
    \draw[very thick,|-|] (0,0) -- (.5,0);
    \draw[very thick,|-|] (1.5,0) -- (2,0);
    \draw[very thick,-|] (2,0) -- (2.5,0);
    \draw[very thick,|-|] (5.5,0) -- (6,0);
    \draw[very thick,-|] (6,0) -- (6.5,0);
    \draw[very thick,|-|] (9.5,0) -- (10,0);
    \draw[very thick,-|] (10,0) -- (10.5,0);
    \node at (0,-.5) {$0$};
    \node at (.5,-.5) {$\eta$};
    \node at (1.3,-.5) {$1-\eta$};
    \node at (2,-.5) {$1$};
    \node at (2.7,-.5) {$1+\eta$};
    \node at (5.3,-.5) {$3-\eta$};
    \node at (6,-.5) {$3$};
    \node at (6.7,-.5) {$3+\eta$};
    \node at (9.3,-.5) {$5-\eta$};
    \node at (10,-.5) {$5$};
    \node at (10.7,-.5) {$5+\eta$};

    \node at (-.4,.2) {$r$};
    {
    \tikzset{-latex,bend left=50};
    \draw (0,.2) to node[above]{\tiny$p^-$} (.7,.2);
    \draw (.7,.2) to node[above]{\tiny$p^+$} (1.3,.2);
    \draw (1.3,.2) to (2.2,.2);
    \draw (2.2,.2) to node[above]{\tiny$p^-$} (3,.2);
    \draw (3,.2) to node[above]{\tiny$p^-$} (4.2,.2);
    \draw (4.2,.2) to[bend left=30] (6,.2);
    \draw (6,.2) to[bend left=0] (7.5,.2);
    \draw (7.5,.2) to[bend left=20] node[above]{\tiny$p^-$} (11,.2);
    }
    \draw (7.5,.2) edge[-latex,out=120,in=60,loop] node[above]{\tiny$p^+$} (7.5,.2);
    
    \node at (-.4,-.9) {$r^+$};
    {
    \tikzset{-latex,bend left=-50};
    \draw (0,-.9) to (.5,-.9);
    \draw (.5,-.9) to (1.5,-.9);
    \draw (1.5,-.9) to (2.2,-.9);
    \draw (2.2,-.9) to[bend left=0] (2.5,-.9);
    \draw (2.5,-.9) to[bend left=-20] (6,-.9);
    \draw (6,-.9) to[bend left=0] (6.5,-.9);
    \draw (6.5,-.9) to[bend left=-20] (10.5,-.9);
    }
    \draw (2.5,-.9) edge[-latex,out=-120,in=-60,loop] (2.5,-.9);
    \draw (6.5,-.9) edge[-latex,out=-120,in=-60,loop] (6.5,-.9);
    
  \end{tikzpicture}
  \caption{A play $r$ and its associated play $r^+$. The price of the
    location in play $r$, when it matters, is denoted on the
    transition exiting this location: for instance the first location
    of $r$ has price $p^-$, whereas the second has price $p^+$. On the
    first transition, the second rule of the definition applies, and
    less time is spent in the location of price $p^-$. On the second
    transition, the first rule applies, and more time is spent in the
    location of price $p^+$. On the third transition, the first rule
    applies and both plays synchronize. The transition of time
    duration $t=0$ (denoted as a loop) in $r$ is supposed to be taken
    in a location owned by player 1, of price $p^+$. In particular, in
    $r^+$, it is taken when the valuation is $3+\eta$, which implies a
    time duration $t^+=0$ as prescribed in the second case of the last
    rule of the definition.}
  \label{fig:ex1}
\end{figure}

\begin{lemma}\label{lem:eta-runs}
  Let $r$ be a play that starts from some state $\eta$-close to a
  border. Then, the play $r^+$ verifies the following properties:
  \begin{enumerate}
  \item\label{item:near} if for some $j$, there exists $k$ such that
    $|\valuation_j+t_j-M_k|\leq \eta$, then
    $\valuation^+_j+t^+_j=\valuation_j+t_j$;
  \item\label{item:far} if for some $j$, there exists $k$ such that
    $\valuation_j+t_j\in (M_k+\eta,M_{k+1}-\eta)$, then
    $\valuation^+_j+t^+_j\in\{M_k+\eta,M_{k+1}-\eta\}$;
  \item $r^+$ is a play that stays $\eta$-close to borders;
  \item\label{item:sim} $r^+\sim r$;
  \item $\Weight(r)\leq \Weight(r^+)$.
  \end{enumerate}
\end{lemma}
In the rest of this section, we call the five points of this Lemma
`property~\ref{item:near}', `property~\ref{item:far}', and so forth.
\begin{remark}
  Before starting the proof, notice that in case
  $\transitions(\location_j,a_j)=(\zeta,\emptyset,\location_{j+1})$,
  i.e., the clock is not reset at step $j$, we have
  $\valuation_{j+1}=\valuation_j+t_j$ (and we will also have
  $\valuation^+_{j+1}=\valuation^+_j+t^+_j$). In particular, in that
  case, property~\ref{item:near} implies that if
  $|\valuation_{j+1}-M_k|\leq \eta$, then
  $\valuation^+_{j+1}=\valuation_{j+1}$, whereas
  property~\ref{item:far} implies that if $\valuation_{j+1}\in
  (M_k+\eta,M_{k+1}-\eta)$, then
  $\valuation^+_{j+1}\in\{M_k+\eta,M_{k+1}-\eta\}$.
\end{remark}
\begin{proof}
  All properties are shown by a simultaneous induction on the length
  of the play. All properties are clearly true in case of a play $r$
  reduced to a single state (which is assumed to be $\eta$-close to a
  border). We now consider a play
  $r=(\location_0,\valuation_0),\allowbreak (t_0,a_0),\ldots,
  (\location_{i+1},\valuation_{i+1})$ with $i+1\geq 1$, and prove the
  properties for the play $r^+$ whose construction has been given
  before.
  \begin{enumerate}
  \item The first property is true directly by construction.
  \item For the second property, by induction hypothesis, it is
    sufficient to prove that if $\valuation_i+t_i\in
    (M_k+\eta,M_{k+1}-\eta)$, then
    $\valuation^+_i+t^+_i\in\{M_k+\eta,M_{k+1}-\eta\}$. 

    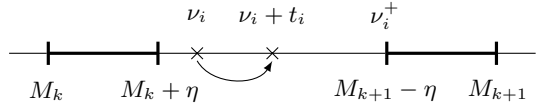
\begin{figure}[tbp]
      \centering
      \begin{tikzpicture}
        \draw(-.5,0) -- (0,0);
        \draw[very thick, |-|](0,0) -- (1.5,0);
        \draw(1,0) -- (4.5,0);
        \draw[very thick, |-|](4.5,0)--(6,0);
        \draw(6,0)--(6.5,0);
        \node at (0,-0.5) {$M_k$};
        \node at (1.5,-0.5) {$M_k+\eta$};
        \node at (4.5,-0.5) {$M_{k+1}-\eta$};
        \node at (6,-0.5) {$M_{k+1}$};
        \node[cross out,inner sep=2pt,draw] at (2,0){};
        \node at (2,0.5) {$\valuation_i$};
        \node[cross out,inner sep=2pt,draw] at (3,0){};
        \node at (3,0.5) {$\valuation_i+t_i$};
        \node at (4.5,0.5) {$\valuation^+_i$};
        \draw[-latex] (2,-.1) to[bend right=60] (3,-.1);
      \end{tikzpicture}
      \caption{Illustration for the proof of
        Lemma~\ref{lem:eta-runs}-\ref{item:far}}
      \label{fig:ex2}
    \end{figure} 
    Suppose first that $\prices(\location_i)=p^-$: then,
    $t_i^+=\max(M_k+\eta-\valuation_i^+,0)$. In case $t_i^+>0$, we
    have $t_i^+=M_{k}+\eta-\valuation^+_i$, hence,
    $\valuation_i^++t_i^+=M_k+\eta$. Otherwise, we have $t_i^+=0$ so
    that $\valuation^+_i+t^+_i=\valuation^+_i$ and
    $\valuation^+_i>M_k+\eta$: we have depicted an example of the
    situation in Fig.~\ref{fig:ex2}. Since $\valuation^+_i$ must be
    $\eta$-close to a border, we have $\valuation^+_i\geq
    M_{k+1}-\eta$. Then, $\valuation^+_i>\valuation_i+t_i\geq
    \valuation_i$. But, $\valuation^+_i$ and $\valuation_i$ are in the
    same region by induction hypothesis of property~\ref{item:sim} (no
    reset has been performed in the previous transition since
    $\valuation_i\neq \valuation^+_i$, so that they are respectively
    equal to $\valuation^+_{i-1}+t^+_{i-1}$ and
    $\valuation_{i-1}+t_{i-1}$). This implies that $\valuation^+_i\in
    [M_{k+1}-\eta,M_{k+1})$ and $\valuation_i\in
    (M_k,M_{k+1}-\eta)$. However, by induction hypothesis, if
    $\valuation_i\in (M_k,M_k+\eta]$, then
    $\valuation^+_i=\valuation_i$ which is forbidden in this
    case. Hence, we have
    $\valuation_i\in(M_k+\eta,M_{k+1}-\eta)$. Hence, by induction
    hypothesis again, we have
    $\valuation^+_i\in\{M_k+\eta,M_{k+1}-\eta\}$, which finally
    implies that $\valuation^+_i=M_{k+1}-\eta$.

    Suppose then that $\prices(\location_i)=p^+$. The result is again
    immediate in case $t^+_i=M_{k+1}-\eta-\valuation_i^+$. Otherwise,
    we have $t^+_i=0$, $\location_i\in\minlocations$, $t_i=0$ and
    $\valuation^+_i= M_k+\eta$. We directly obtain
    $\valuation^+_i+t^+_i=M_k+\eta$ which permits to conclude.
    
  \item We now prove that $r^+$ is indeed a play. The only non-trivial
    property is that $t^+_i$ is a non negative delay, especially when
    $t^+_i=\valuation_i+t_i-\valuation^+_i$ or
    $t^+_i=M_{k+1}-\eta-\valuation^+_i$. 

    Consider first the case
    $t^+_i=\valuation_i+t_i-\valuation^+_i$. It has to be shown that
    if there exists $k$ such that $|\valuation_i+t_i-M_k|\leq \eta$,
    then $\valuation^+_i$ is not greater than
    $\valuation_i+t_i$. Notice first that $\valuation_i\leq
    \valuation_i+t_i$. Hence, if $\valuation_i$ is $\eta$-close to a
    border, then by property~\ref{item:near},
    $\valuation^+_i=\valuation_i\leq \valuation_i+t_i$. Otherwise, we
    know that $\valuation_i\notin [M_k-\eta,M_k+\eta]$, hence, either
    $\valuation_i<M_k-\eta$, or $\valuation_i>M_k+\eta$. Since
    $\valuation_i\leq \valuation_i+t_i\leq M_k+\eta$, we can rule out
    the case $\valuation_i>M_k+\eta$, and conclude that
    $\valuation_i<M_k-\eta$. By property \ref{item:far} (applied to
    $\valuation_{i}^+=\valuation_{i-1}^++t_{i-1}^+$, since no reset
    has just been performed, knowing that $\valuation_i$ is not
    $\eta$-close to $0$), this implies that $\valuation^+_i\leq
    M_k-\eta\leq \valuation_i+t_i$.

    Consider then the case $t^+_i=M_{k+1}-\eta-\valuation^+_i$, which
    holds when $\valuation_i+t_i\in(M_k+\eta,M_{k+1}-\eta)$ and
    $\prices(\location_i)=p^+$. Then, $\valuation_i\leq
    \valuation_i+t_i<M_{k+1}-\eta$. Hence, either $\valuation_i$ is
    $\eta$-close to a border (in particular when the clock has just
    been reset), in which case, by property~\ref{item:near}, we have
    $\valuation^+_i=\valuation_i<M_{k+1}-\eta$, so that
    $t^+_i=M_{k+1}-\eta-\valuation^+_i>0$. Or
    $\valuation_i\in(M_{k'}+\eta,M_{k'+1}-\eta)$ with $k'\leq k$. By
    property \ref{item:far}, this implies that $\valuation^+_i\leq
    M_{k'+1}-\eta\leq M_{k+1}-\eta$, once again implying that
    $t^+_i\geq 0$.

    The fact that $r^+$ stays $\eta$-close to the borders is
    directly implied by properties \ref{item:near} and \ref{item:far}.

  \item Property $r^+\sim r$ is also a direct consequence of
    properties \ref{item:near} and \ref{item:far}.

  \item It only remains to prove that $\Weight(r)\leq
    \Weight(r^+)$. Notice that the weights $\Weight(r)$ and
    $\Weight(r^+)$ can be decomposed as sums of weights of subplays
    that start and end in the same configuration, but with
    intermediate configurations that do no match. Hence, it is
    sufficient to prove the inequality for subplays $r$ and $r^+$ that
    start in the same configuration (at step $j$,
    $\valuation_{j}+t_{j}=\valuation^+_{j}+t^+_{j}$) and that do not
    contain other identical configurations, unless possibly the last
    one. For the sake of
    simplicity, we suppose in the following that $j=0$. In particular,
    we may suppose that $r$ does not contain reset transitions or
    positions $j$ such that $\valuation_j$ is $\eta$-close to
    borders, except possibly the very last one (otherwise, the two
    plays would again contain matching configurations). Since there
    are no resets, we have $\valuation_j=\valuation_{j-1}+t_{j-1}$ and
    $\valuation^+_j=\valuation^+_{j-1}+t^+_{j-1}$ for every $j>0$.

    We now consider separately the possibility of sets $\{p^-,p^+\}$.
    \begin{itemize}
    \item As a first case, consider that $p^-=-1$ and $p^+=+1$. We
      prove by induction over $0\geq j \geq i$ that
      \[\Weight(r^+[j+1])\geq
      \Weight(r[j+1])+|\valuation_{j+1}-\valuation^+_{j+1}|\,.\] For
      the sake of brevity, we omit the weights of the actions in this
      proof, but notice that the same actions occur in $r$ and $r^+$
      since these two plays are equivalent (by property~4).
      \begin{itemize}
      \item {\bf Base case}. If $j=0$, then we have supposed that
        $\valuation_0$ is $\eta$-close to a border, so that
        $\valuation^+_0=\valuation_0$. Let $k$ be such
        that\footnote{Remember that we suppose that there is no
          synchronization for now, so that no valuation in $r$ is
          $\eta$-close to a border.}
        $\valuation_1\in(M_k+\eta,M_{k+1}-\eta)$. Then, if
        $\prices(\location_0)=-1$, $t^+_0=\max
        (M_{k}+\eta-\valuation_0,0)$. However, $\valuation_0\leq
        \valuation_1$ and $\valuation_0$ is $\eta$-close to a
        border, so that $\valuation_0\leq M_k+\eta$. Hence,
        $t_0^+=M_k+\eta-\valuation_0$, which implies
        $\valuation^+_1=\valuation_0+t^+_0=M_k+\eta$. Hence,
        \begin{align*}
          \Weight(r^+[1]) = -(\valuation^+_1-\valuation_0) &=
          \Weight(r[1]) +\valuation_1-\valuation^+_1 \\ &=
          \Weight(r[1])+|\valuation_1-\valuation^+_1|\,.
        \end{align*}
        Consider then the case where $\prices(\location_0)=+1$. Notice
        that we supposed that $\valuation_1$ is not $\eta$-close to a
        border, contrary to $\valuation_0$, so that $t_0>0$. Hence, we
        are sure that $t^+_0 = M_{k+1}-\eta-\valuation_0$. This
        implies that
        $\valuation^+_1=\valuation^+_0+t^+_0=\valuation_0+t^+_0
        =M_{k+1}-\eta > \valuation_1$ and
        \begin{align*}
          \Weight(r^+[1]) = \valuation^+_1-\valuation_0 &=
          \Weight(r[1]) +\valuation^+_1-\valuation_1\\ &=\Weight(r[1])+
          |\valuation_1-\valuation^+_1|\,.
        \end{align*}
      \item {\bf Inductive case}. Let us suppose that the property is
        proved for all indices less than or equal to $j$, and prove it
        for $j+1$. We let $k$ be such that
        $\valuation_{j+1}=\valuation_j+t_j\in(M_k+\eta,M_{k+1}-\eta)$. We
        will distinguish four possible cases depending on
        $\prices(\location_j)$ and the relative order between
        $\valuation^+_{j+1}$ and $\valuation_{j+1}$.

        \begin{enumerate}
        \item We first suppose that $\prices(\location_j)=+1$. Then,
          \begin{align*}
            &\Weight(r^+[j+1]) = \Weight(r^+[j])
            +\valuation^+_{j+1}-\valuation^+_j  \notag\\
            &\hspace{1cm}\geq \Weight(r[j]) + |\valuation_j-\valuation^+_j| +
            \valuation^+_{j+1}-\valuation^+_j\quad \text{(Ind. Hyp.)}\notag
          \end{align*}
          Hence, since $ \Weight(r[j+1]) = \Weight(r[j])
          +(\valuation_{j+1}-\valuation_j)$:
          \begin{align}
            \Weight(r^+[j+1]) & \geq
            \Weight(r[j+1]) -(\valuation_{j+1}-\valuation_j)\notag \\
            &\qquad + |\valuation_j-\valuation^+_j| +
            \valuation^+_{j+1}-\valuation^+_j\,.\label{eq:+1}
          \end{align}
          \begin{enumerate}
          \item In the case where $\valuation^+_{j+1}>
            \valuation_{j+1}$, we have
            $|\valuation_{j+1}-\valuation^+_{j+1}|=
            \valuation^+_{j+1}-\valuation_{j+1}$ so that \eqref{eq:+1}
            can be rewritten
            \begin{multline*}
              \Weight(r^+[j+1]) \geq \Weight(r[j+1]) +
              |\valuation_{j+1}-\valuation^+_{j+1}|\\
              +(\valuation_j-\valuation^+_j)+|\valuation_j-\valuation^+_j|
            \end{multline*}
            which is greater than or equal to $\Weight(r[j+1]) +
            |\valuation_{j+1}-\valuation^+_{j+1}|$ since
            $\valuation_j-\valuation^+_j\geq
            -|\valuation_j-\valuation^+_j|$.

          \item Similarly, in the case where
            $\valuation^+_{j+1}<\valuation_{j+1}$, we have
            $|\valuation_{j+1}-\valuation^+_{j+1}|=
            \valuation_{j+1}-\valuation^+_{j+1}$. Notice that this
            necessarily implies that $\location_i\in\minlocations$ and
            $t_j=0$ and $\valuation^+_j= M_{k}+\eta$: otherwise, we
            would have $t^+_{j}=M_{k+1}-\eta-\valuation^+_j$ and thus
            $\valuation^+_{j+1} = M_{k+1}-\eta>\valuation_{j+1}$ that
            contradicts the hypothesis. In particular, we have
            $t^+_{j}=0$ and $\valuation^+_{j}=\valuation^+_{j+1}<
            \valuation_{j+1}=\valuation_{j}$, so that \eqref{eq:+1}
            becomes
            \begin{align*}
              \Weight(r^+[j+1]) &\geq \Weight(r[j+1]) +
              |\valuation_j-\valuation^+_j| \\ &= \Weight(r[j+1]) +
              |\valuation_{j+1}-\valuation^+_{j+1}| \,.
            \end{align*}
          \end{enumerate}
      
        \item Suppose then that $\prices(\location_j)=-1$. Then, a
          similar calculation gives
          \begin{align*}
            &\Weight(r^+[j+1]) = 
            \Weight(r^+[j])-(\valuation^+_{j+1}-\valuation^+_j)\notag\\
            &\hspace{1cm}\geq \Weight(r[j]) + |\valuation_j-\valuation^+_j| -
            (\valuation^+_{j+1}-\valuation^+_j)
            \quad\text{(Ind. Hyp.)}\notag 
          \end{align*}
          Hence, since $ \Weight(r[j+1]) = \Weight(r[j])
          -(\valuation_{j+1}-\valuation_j)$:
          \begin{align}
            \Weight(r^+[j+1])& \geq \Weight(r[j+1])
            +(\valuation_{j+1}-\valuation_j)  \notag\\
            &\qquad +
            |\valuation_j-\valuation^+_j|
            -(\valuation^+_{j+1}-\valuation^+_j)\,.\label{eq:-1}
          \end{align}
          \begin{enumerate}
          \item Once again, if $\valuation^+_{j+1}<\valuation_{j+1}$,
            we have $|\valuation_{j+1}-\valuation^+_{j+1}|=
            \valuation_{j+1}-\valuation^+_{j+1}$ so that \eqref{eq:-1}
            becomes
            \begin{multline*}
              \Weight(r^+[j+1]) \geq \Weight(r[j+1]) +
              |\valuation_{j+1}-\valuation^+_{j+1}|\\ -
              (\valuation_j-\valuation^+_j)+|\valuation_j-\valuation^+_j|
            \end{multline*}
            which is greater than or equal to $\Weight(r[j+1]) -
            |\valuation_{j+1}-\valuation^+_{j+1}|$ since
            $\valuation_j-\valuation^+_j\leq
            |\valuation_j-\valuation^+_j|$. 

          \item Similarly, if $\valuation^+_{j+1}>\valuation_{j+1}$,
            we know by property~\ref{item:far} that
            $\valuation^+_{j+1}=M_{k+1}-\eta$. In particular, since
            $t^+_j=\max(M_k+\eta-\valuation^+_j,0)$ and
            $\valuation^+_{j+1}=\valuation^+_j+t^+_j\neq M_k+\eta$, we
            know that $t^+_j=0$, and $\valuation^+_j\geq
            M_k+\eta$. This implies
            $\valuation^+_{j+1}=\valuation^+_j=
            M_{k+1}-\eta>\valuation_{j+1}\geq \valuation_j$. Knowing
            that $|\valuation_{j+1}-\valuation^+_{j+1}|=
            \valuation^+_{j+1}-\valuation_{j+1}$, we obtain from
            \eqref{eq:-1}
            \begin{align*}
              \Weight(r^+[j+1]) &\geq \Weight(r[j+1]) +
              |\valuation_{j+1}-\valuation^+_{j+1}|\\
              &\qquad + 2(\valuation_{j+1}-\valuation^+_{j+1})-2
              (\valuation_j-\valuation^+_j)\\
              &= \Weight(r[j+1]) +
              |\valuation_{j+1}-\valuation^+_{j+1}|+
              2(\valuation_{j+1}-\valuation_j) \\
              &\geq \Weight(r[j+1]) +
              |\valuation_{j+1}-\valuation^+_{j+1}|\,.
            \end{align*}
          \end{enumerate}
        \end{enumerate}
      \end{itemize}

      We finally have proved the property by induction. Notice in
      particular that this shows that $\Weight(r^+[j+1])\geq
      \Weight(r[j+1])$ for every $j$ with
      $\valuation^+_{j+1}\neq\valuation_{j+1}$. To conclude the proof
      of $\Weight(r^+)\geq \Weight(r)$, it remains to deal with the
      case of a possible last transition ending with
      $\valuation^+_{i+1}=\valuation_{i+1}$. Unless $i=0$, in which
      case we have $\Weight(r^+)=\Weight(r)$, we know by hypothesis
      that $\valuation^+_{i}\neq \valuation_{i}$. By the previous
      property, we have $\Weight(r^+[i])\geq
      \Weight(r[i])+|\valuation_i-\valuation^+_i|$. Moreover,
      $\Weight(r^+)=\Weight(r^+[i])+
      \prices(\location_i)(\valuation^+_{i+1}-\valuation^+_i)$ and
      $\Weight(r)=\Weight(r[i])+
      \prices(\location_i)(\valuation_{i+1}-\valuation_i)$. In the
      overall (using the fact that
      $\valuation^+_{i+1}=\valuation_{i+1}$), we get
      \[\Weight(r^+)\geq \Weight(r)
      +\prices(\location_i)(\valuation_i-\valuation^+_i)+
      |\valuation_i-\valuation^+_i|\,.\] In all cases, we verify that
      $-\prices(\location_i)(\valuation_i-\valuation^+_i)\leq
      |\valuation_i-\valuation^+_i|$, so that we have proved that
      $\Weight(r^+)\geq \Weight(r)$.

    \item We now consider the case where $p^-=0$ and $p^+=+1$ (the
      case $p^-=-1$ and $p^+=0$ is very similar, and not explained in
      details here). We prove another inequality by induction over
      $0\geq j \geq i$, namely that
      \[\Weight(r^+[j+1])\geq
      \Weight(r[j+1])+\max(\valuation_{j+1}-\valuation^+_{j+1},0)\,.\]
      The proof is very similar to the previous case, and we conclude
      as previously. \qed
    \end{itemize}
  \end{enumerate}
\end{proof}

We now go to the proof of Lemma~\ref{lem:uniform}. In case,
$\uppervalue(s)=-\infty$ the Lemma is trivially true. We first
consider the case $\uppervalue(s)<+\infty$. Let
$\minstrategy'\in\uniformminstrategies^\eta$. We now explain how to
construct a strategy $\minstrategy\in\minstrategies$ such that for all
states $s$
\[\sup_{\maxstrategy'\in\uniformmaxstrategies^\eta}
\Weight(\outcomes(s,\minstrategy',\maxstrategy'))\geq
\sup_{\maxstrategy\in\maxstrategies}
\Weight(\outcomes(s,\minstrategy,\maxstrategy))\,.\] To prove such an
inequality, we will consider any strategy
$\maxstrategy\in\maxstrategies$ and construct a strategy
$\maxstrategy'\in\uniformmaxstrategies^\eta$ such that
\[\Weight(\outcomes(s,\minstrategy',\maxstrategy'))\geq
\Weight(\outcomes(s,\minstrategy,\maxstrategy))\,.\]

Strategy $\minstrategy$ follows $\minstrategy'$ in case of plays
staying $\eta$-close to borders. We must however extend it to deal
with the other plays faithfully. Let
$r=(\location_0,\valuation_0),\allowbreak (t_0,a_0),\ldots,
(\location_{i},\valuation_{i})$ be any finite play ending in a
location $\location_{i}$ of player 1, and
$r^+=(\location^+_0,\valuation^+_0),(t^+_0,a^+_0), \ldots,\allowbreak
(\location^+_{i},\valuation^+_{i})$ the play constructed as before. By
Lemma~\ref{lem:eta-runs}, we know that $r^+$ is a play that stays
$\eta$-close to borders. Hence, $\minstrategy'(r^+)=(t'_i,a)$, for
some $t'_i \in \Rpos$, with $\valuation^+_i+t'_i$ being $\eta$-close
to a border. We let $t_i= \max(\valuation^+_i+t'_i-\valuation_i,0)$
and $\minstrategy(r) = (t_i,a)$. Let $\tilde r$ (respectively, $r'$)
be the play $r$ (respectively, $r^+$) extended with the step
prescribed by $\minstrategy$ (respectively, $\minstrategy'$). Then, we
prove that $r'$ matches the construction above starting from the run
$\tilde r$, i.e., $\tilde r^+=r'$. By construction, we only have to
verify that the value of $t'_i$ is consistent with the previous
constructions, i.e., $t'_i=t^+_i$.

\begin{lemma}\label{lem:consistency}
  We have $t'_i=t^+_i$.
\end{lemma}
\begin{proof}
  In case $\valuation_i\leq \valuation^+_i+t'_i$, since
  $t_i=\max(\valuation^+_i+t'_i-\valuation_i,0)$, we have
  $\valuation_i+t_i=\valuation^+_i+t'_i$ which is $\eta$-close to
  borders, and $t'_i=\valuation_i+t_i-\valuation^+_i$ that fits
  with the definition of $t^+_i$ in $\tilde r^+$: hence $t'_i=t^+_i$
  in that case.

  Otherwise, we have $\valuation_i>\valuation^+_i+t'_i$ and
  $t_i=0$. In particular, $\valuation_i>\valuation^+_i$ so that
  $\valuation_i$ cannot be $\eta$-close to a border (by
  Lemma~\ref{lem:eta-runs}-\ref{item:near}). By
  Lemma~\ref{lem:eta-runs}-\ref{item:far}, since
  $\valuation^+_i<\valuation_i$, there exists $k$ such that
  $\valuation^+_i=M_k+\eta$ and
  $\valuation_i\in(M_k+\eta,M_{k+1}-\eta)$. We are thus in the
  situation $\location_i\in\minlocations$ and $t_i=0$ and
  $\valuation^+_i= M_k+\eta$ that prescribes a choice of the next time
  delay $t^+_i=0$. Hence, we must prove that $t'_i=0$. It is
  necessarily the case, since $M_k+\eta=\valuation^+_i\leq
  \valuation^+_i+t'_i<\valuation_i$ with $\valuation^+_i+t'_i$ being
  $\eta$-close to a border and
  $\valuation_i\in(M_k+\eta,M_{k+1}-\eta)$. Finally, we obtain
  $t'_i=0=t^+_i$. \qed
\end{proof}

We now consider any strategy $\maxstrategy\in\maxstrategies$, and
construct a strategy $\maxstrategy'\in\uniformmaxstrategies^\eta$ such
that $\outcomes(s,\minstrategy',\maxstrategy')=
\outcomes(s,\minstrategy,\maxstrategy)^+$ for every state $s$
$\eta$-close to a border. From Lemma~\ref{lem:eta-runs}, we will
then get
\[\Weight(\outcomes(s,\minstrategy',\maxstrategy'))\geq
\Weight(\outcomes(s,\minstrategy,\maxstrategy))\,,\] which will enable
us to conclude. We assume $\outcomes(s,\minstrategy,\maxstrategy)^+=
(\location^+_0,\valuation^+_0),(t^+_0,a^+_0),\allowbreak \ldots,
(\location^+_{n},\valuation^+_{n}),\ldots$ with
$s=(\location^+_0,\valuation^+_0)$ $\eta$-close to a border. Then,
we first define $\maxstrategy'$ over the finite plays
$\outcomes(s,\minstrategy,\maxstrategy)^+[n]$ with $n\in \N$, by
letting
\[\maxstrategy'(\outcomes(s,\minstrategy,\maxstrategy)^+[n]) =
(t^+_n,a^+_n)\,.\] Notice first that this strategy verifies
$\outcomes(s,\minstrategy',\maxstrategy')[n]=
\outcomes(s,\minstrategy,\maxstrategy)[n]^+$ for every state $s$
$\eta$-close to a border, by induction on $n\in\N$. In fact, in case
$\outcomes(s,\minstrategy,\maxstrategy)^+[n]$ ends with a state of
player 1, the equation holds by construction of $\minstrategy$, and in
case it ends with a state of player 2, by construction of
$\maxstrategy'$.

Once built on these finite plays, it is possible to extend
$\maxstrategy'$ as an $\eta$-region-uniform strategy defined over
every play: in particular, if
$\outcomes(s,\minstrategy,\maxstrategy)^+[n]\sim_\eta
\outcomes(s',\minstrategy,\maxstrategy)^+[n]$ (with $s$ and $s'$
different states $\eta$-close to a border), we have that
$\maxstrategy'(\outcomes(s,\minstrategy,\maxstrategy)^+[n])
\sim_\eta\maxstrategy'(\outcomes(s',\minstrategy,\maxstrategy)^+[n])$
(induced by Lemma~\ref{lem:eta-runs}-4) validating the definition of
$\eta$-region-uniform strategies.

This concludes the proof of $\uppervalue(s)\leq
\uniformuppervalue^\eta(s)$ in case $\uppervalue(s)<+\infty$.

\medskip Finally, the case $\uppervalue(s)=+\infty$, corresponds to
two possible situations: either player~2 has a way to ensure that the
goal is never reached, or he cannot have such a guarantee, but still
is able to make the price go bigger and bigger, i.e. he has a family
of strategies that do not forbid from reaching the goal but ensure a
price which is not bounded over the family\footnote{Indeed, we will
  show in Appendix~\ref{app:lem:infinite} that only the first
  alternative is possible.}. It only remains to prove that player 2
can do it so with $\eta$-region-uniform strategies too. Let
$\minstrategy\in\uniformminstrategies^\eta$ be an
$\eta$-region-uniform strategy for player 1. We know that
\[\sup_{\maxstrategy\in\maxstrategies}
\Weight(\outcomes(s,\minstrategy,\maxstrategy))=+\infty\,.\] The first
case corresponds to the one where there exists a strategy
$\maxstrategy\in\maxstrategies$ such that
$\Weight(\outcomes(s,\minstrategy,\maxstrategy))=+\infty$, i.e., in
this outcome, the goal is not reached. As previously, it is possible
to reconstruct from $\maxstrategy$ an $\eta$-region-uniform strategy
$\maxstrategy'$ achieving the very same goal (notice that the goal is
definable with regions): the only difference is the fact that
$\maxstrategy'$ must now mimic an infinite number of prefixes since
the outcome is no longer finite. The second case finally corresponds
to the one where there is no strategy $\maxstrategy\in\maxstrategies$
such that
$\Weight(\outcomes(s,\minstrategy,\maxstrategy))=+\infty$. We
construct as previously a strategy $\minstrategy\in\maxstrategies$ for
player 1 from strategy $\minstrategy'$. From the fact that
$\sup_{\maxstrategy\in\maxstrategies}
\Weight(\outcomes(s,\minstrategy,\maxstrategy))>M$, we know the
existence of a strategy $\maxstrategy\in\maxstrategies$ so that
$\Weight(\outcomes(s,\minstrategy,\maxstrategy))>M$. Since this price
is finite by hypothesis, the previous construction allows us to obtain
a region-uniform strategy $\maxstrategy'$ verifying:
\[\Weight(\outcomes(s,\minstrategy',\maxstrategy')) \geq
\Weight(\outcomes(s,\minstrategy,\maxstrategy))>M\] for all $M \in
\R$. This proves that $\uniformuppervalue(s)=+\infty$.

\subsection{Proof of Lemma~\ref{lem:convergent}}
\label{app:lem:convergent}

The proof of Lemma~\ref{lem:convergent} is a refinement of the proof
of Lemma~\ref{lem:uniform} (where, moreover, the roles of both players
are switched). To avoid the divergence phenomenon of
Example~\ref{ex:convergence}, we restrict our attention to plays
staying $\eta$-close to borders, that moreover jump closer and close
to borders. More formally, a play
$r=(\location_0,\valuation_0),(t_0,a_0),\ldots$ is said to be
\emph{$\eta$-convergent} if for all $i\geq 0$, there exists $k$ such
that either $|\valuation_i+t_i - M_k|\leq \eta/2^{i+1}$, or $t_i=0$
and $\valuation_i\in (M_k + \eta/2^{i+1},M_k+\eta]$. Notice in
particular that $\eta$-convergent runs stay $\eta$-close to
borders. Unfortunately, it is not possible to define
$\eta$-convergent runs as runs such that the first property
($|\valuation_i+t_i - M_k|\leq \eta/2^{i+1}$) always holds since it
would forbid a player to delay $0$ time units when, in the round $i$,
its valuation is $\valuation_i\in (M_k+\eta/2^{i+1},M_k+\eta]$. The
second property is there to fix this issue. 

We first study more precisely the relationship between general plays
and $\eta$-convergent plays, like we did for runs staying $\eta$-close
to borders. More precisely, we now explain how to construct from
any finite play $r$, a play $r^-$ such that
\begin{enumerate*}[label=(\emph{\roman*})]
\item $r^-$ is an $\eta$-convergent play;
\item $r$ and $r^-$ are region-equivalent; and 
\item $\Weight(r^-)\leq \Weight(r)$.
\end{enumerate*}
Moreover, the construction will trivially verify that if $r'$ and $r$
are two plays that coincide on their prefix of length $n$, i.e.,
$r'[n]=r[n]$, then $(r')^-[n]=r^-[n]$. Indeed, the construction is by
induction on the length of the play, and very similar to the
construction of $r^+$ in the previous section. The main difference
with the case of $r^+$, apart from the fact that we look for a play
with a smaller weight rather than a greater weight, belongs in the
fact that $r^-$ must jump closer and closer to the borders.

Henceforth, let $r=(\location_0,\valuation_0),(t_0,a_0),\ldots,
(\location_{i+1},\valuation_{i+1})$ and suppose that $\valuation_0$ is
$\eta$-close to some $M_k$. We construct a play
$r^-=(\location^-_0,\valuation^-_0),(t^-_0,a^-_0),
\ldots,(\location^-_{i+1},\valuation^-_{i+1})$. In case $i+1=0$, we
simply let $\location^-_0=\location^+_0=\location_0$ and
$\valuation^-_0=\valuation^+_0=\valuation_0$. Otherwise, we consider
the play $r^-$ constructed by induction up to its configuration of
index $i$, and we now explain how to construct the next timed action
$(t^-_{i},a^-_i)$. First, we let $a^-_{i}=a_{i}$. Then, we distinguish
between three cases:
\begin{itemize}
\item if there exists $k$ such that $|\valuation_i+t_i-M_k|\leq
  \eta/2^{i+1}$, then we let $t^-_i= \valuation_i+t_i-\valuation^-_i$;
\item if there exists $k$ such that $\valuation_i+t_i\in
  (M_k+\eta/2^{i+1},M_{k+1}-\eta/2^{i+1})$ and
  $\prices(\location_i)=p^+$, then we let
  $t^-_i=\max(M_k+\eta/2^{i+1}-\valuation^-_i,0)$;
\item if there exists $k$ such that $\valuation_i+t_i\in
  (M_k+\eta/2^{i+1},M_{k+1}-\eta/2^{i+1})$ and
  $\prices(\location_i)=p^-$, then there are two cases. In case
  $\location_i\in\minlocations$ or $t_i>0$ or $\valuation^-_i>
  M_{k}+\eta$, in which case we let
  $t^-_i=M_{k+1}-\eta/2^{i+1}-\valuation^-_i$. Otherwise, i.e., if
  $\location_i\in\maxlocations$ and $t_i=0$ and $\valuation^-_i\leq
  M_{k}+\eta$, then we let $t^-_i=
  \max(M_k+\eta/2^{i+1}-\valuation^-_i,0)$.
\end{itemize}

Once again, we will verify in the next lemma that $t^-_i$ is always
non-negative, ensuring that we indeed construct a valid play. Once
defined $t^-_i$, we let $\location^-_{i+1}=\location_{i+1}$, and
valuation $\valuation^-_{i+1}$ is defined to be consistent with the
semantics $\sem\arena$ of the game.

\begin{lemma}\label{lem:eta-convergent-runs}
  Let $r$ be a play that starts from some state $\eta$-close to a
  border. Then, the play $r^-$ constructed before verifies the following
  properties:
  \begin{enumerate}
  \item\label{item:near-convergent} if for some $j$, there exists $k$
    such that $|\valuation_j+t_j-M_k|\leq \eta/2^{j+1}$, then
    $\valuation^-_j+t^-_j=\valuation_j+t_j$;
  \item\label{item:far-convergent} if for some $j$, there exists $k$
    such that $\valuation_j+t_j\in
    (M_k+\eta/2^{j+1},M_{k+1}-\eta/2^{j+1})$, then
    \begin{enumerate}
    \item $\valuation^-_j+t^-_j\in\{M_k+\eta/2^{j+1},M_{k+1}-\eta/2^{j+1},
      \valuation^-_{j}\}$;
    \item $\valuation^-_j+t^-_j\leq M_{k+1}-\eta/2^{j+1}$;
    \item if $t_j^-=0$ and $\valuation^-_j\leq M_k+\eta$, then
      $\valuation^-_j\leq \valuation_j+t_j$;
    \end{enumerate}
  \item $r^-$ is an $\eta$-convergent play;
  \item $r^-\sim r$;
  \item $\Weight(r^-)\leq \Weight(r)$.
  \end{enumerate}
\end{lemma}
\begin{remark}
  As for Lemma~\ref{lem:uniform}, notice that if the clock is not
  reset at step $j$, we have $\valuation_{j+1}=\valuation_j+t_j$, and
  $\valuation^-_{j+1}=\valuation^-_j+t^-_j$. In particular, in that
  case, property~\ref{item:near-convergent} implies that if
  $|\valuation_{j+1}-M_k|\leq \eta/2^{j+1}$, then
  $\valuation^-_{j+1}=\valuation_{j+1}$. Similarly,
  property~\ref{item:far-convergent} implies that if
  $\valuation_{j+1}\in (M_k+\eta/2^{j+1},M_{k+1}-\eta/2^{j+1})$, then
  $\valuation^-_{j+1}\in \{M_k+\eta/2^{j+1},M_{k+1}-\eta/2^{j+1},
  \valuation^-_{j}\}$. Moreover, $\valuation^-_{j+1}\leq
  M_{k+1}-\eta/2^{j+1}$ and if $t_j^-=0$ and $\valuation^-_j\leq
  M_k+\eta$, then $\valuation^-_j\leq \valuation_{j+1}$. Properties
  ($b$) and ($c$) will be useful in the proof by induction of
  subsequent properties.
\end{remark}

\begin{proof}
  All properties are shown by induction on the length of the play. All
  properties are clearly true in case of a play $r$ reduced to a
  single state (which is assumed to be $\eta$-close to a border). We
  now consider a play $r=(\location_0,\valuation_0),\allowbreak
  (t_0,a_0),\ldots, (\location_{i+1},\valuation_{i+1})$ with $i+1\geq
  1$ and we prove the properties for the play $r^-$ constructed
  before.
  \begin{enumerate}
  \item The first property is true directly by construction.
  \item For the second property, by induction hypothesis, it is
    sufficient to prove the property for $j=i$.
    Hence, suppose that $\valuation_i+t_i\in
    (M_k+\eta/2^{i+1},M_{k+1}-\eta/2^{i+1})$.

    In case $\prices(\location_i)=p^+$, we have
    $t_i^-\in\{M_{k}+\eta/2^{i+1}-\valuation^-_i,0\}$. In case
    $\prices(\location_i)=p^-$, we have $t_i^-
    \in\{M_{k}+\eta/2^{i+1}-\valuation^-_i,
    M_{k+1}-\eta/2^{i+1}-\valuation_i^-,0\}$.  This implies that
    $\valuation^-_i+t^-_i\in\{M_k+\eta/2^{i+1},M_{k+1}-\eta/2^{i+1},
    \valuation^-_{i}\}$, i.e., property~($a$). Property~($b$), namely
    $\valuation^-_i+t^-_i\leq M_{k+1}-\eta/2^{i+1}$, needs only to be
    proved when $\valuation^-_i+t^-_i=\valuation^-_i$ (otherwise the
    property is trivially verified since
    $M_k+\eta/2^{i+1}<M_{k+1}-\eta/2^{i+1}$). In that case, there are
    two possibilities. If $\valuation_i$ is $\eta/2^i$-close to
    borders, since $\valuation_i\leq \valuation_i+t_i$, we have
    $\valuation_i\leq M_{k+1}-\eta/2^{i+1}$, and by
    property~\ref{item:near-convergent}, $\valuation^-_i=\valuation_i
    \leq M_{k+1}-\eta/2^{i+1}$ (as $\valuation^-_i=\valuation^-_{i-1}
    + t^-_{i-1}=\valuation_{i-1} + t_{i-1}=\valuation_i$). If
    $\valuation_i$ is not $\eta/2^i$-close to borders, then
    $\valuation_i\in(M_{k'}+\eta/2^i,M_{k'+1}-\eta/2^i)$, with $k'\leq
    k$. By induction, $\valuation^-_i\leq M_{k'+1}-\eta/2^i$, so that
    $\valuation^-_i+t^-_i=\valuation^-_i\leq M_{k+1}-\eta/2^{i+1}$.

    Finally, let us prove property~($c$). Assume that $t_i^-=0$ and
    $\valuation_i^-\leq M_k+\eta$. We now prove that
    $\valuation_i^-\leq \valuation_i$. Whatever the value of
    $\prices(\location_i)$, we have that $\valuation^-_i\geq
    M_k+\eta/2^{i+1}$. There are then four cases.
    \begin{itemize}
    \item In a first case, we have $\valuation_i\leq M_k$. Whatever
      $\valuation_i$ is $\eta/2^i$-close to a border or not, we
      obtain (by property~\ref{item:near-convergent} or by induction),
      that $\valuation^-_i\leq M_k\leq \valuation_i+t_i$. 
    \item In a second case, we have $\valuation_i\in
      (M_k,M_k+\eta/2^i]$. By property~\ref{item:near-convergent}, we
      deduce that $\valuation^-_i=\valuation_i\leq
      \valuation_i+t_i$.
    \item The third case corresponds to
      $\valuation_i\in(M_k+\eta/2^{i},M_{k+1}-\eta/2^i)$, which
      implies that $\valuation_{i-1}+t_{i-1}=
      \valuation_i\in(M_k+\eta/2^i,M_{k+1}-\eta/2^i)$. By induction,
      we obtain that $\valuation^-_i=\valuation^-_{i-1}+t^-_{i-1}\in
      \{M_k+\eta/2^i,M_{k+1}-\eta/2^i,\valuation^-_{i-1}\}$. If
      $\valuation^-_i=M_k+\eta/2^i<\valuation_i\leq\valuation_i+t_i$,
      we conclude directly. The case
      $\valuation^-_i=M_{k+1}-\eta/2^i>M_k+\eta$ leads to a
      contradiction. Finally, if $\valuation^-_i=\valuation^-_{i-1}$,
      since we assume $\valuation^-_{i-1}\leq M_k+\eta$, we obtain by
      induction that $\valuation^-_{i-1}\leq
      \valuation_{i-1}+t_{i-1}=\valuation_i\leq \valuation_i+t_i$.
    \item The fourth case is
      $\valuation_i\in[M_{k+1}-\eta/2^i,M_{k+1}-\eta/2^{i+1})$, but
      then $\valuation^-_i\leq M_k+\eta \leq M_{k+1}-\eta/2^i\leq
      \valuation_i\leq \valuation_i+t_i$.
    \end{itemize}

  \item We now prove that $r^-$ is indeed a play. The only non-trivial
    property is that $t^-_i$ is a non negative delay, especially when
    $t^-_i=\valuation_i+t_i-\valuation^-_i$ or
    $t^-_i=M_{k+1}-\eta/2^{i+1}-\valuation^-_i$. Consider first the
    case $t^-_i=\valuation_i+t_i-\valuation^-_i$. It has to be shown
    that if there exists $k$ such that $|\valuation_i+t_i-M_k|\leq
    \eta/2^{i+1}$, then $\valuation^-_i$ is not greater than
    $\valuation_i+t_i$. Notice that $\valuation_i\leq
    \valuation_i+t_i$. Hence, if $\valuation_i$ is $\eta/2^{i}$-close
    to a border, then by property \ref{item:near-convergent},
    $\valuation^-_i=\valuation_i\leq \valuation_i+t_i$, and we are
    done. Otherwise (i.e. if $\valuation_i$ is not $\eta/2^{i}$-close
    to a border), since $\valuation_i\leq \valuation_i+t_i\leq
    M_k+\eta/2^{i+1}$, we even know that
    $\valuation_i<M_k-\eta/2^{i}$. Since no reset has been possibly
    performed during action $a_{i-1}$ (otherwise, $\valuation_i=0$ is
    $\eta/2^i$-close to a border), we have
    $\valuation_i=\valuation_{i-1}+t_{i-1}
    \in(M_{k'}+\eta/2^{i},M_{k'+1}-\eta/2^{i})$ with $k'< k$. By
    property \ref{item:far-convergent}, this implies that
    $\valuation^-_i=\valuation^-_{i-1}+t^-_{i-1}\leq
    M_{k'+1}-\eta/2^{i}$. In consequence, $\valuation_i^-\leq
    M_k-\eta/2^{i}\leq \valuation_i+t_i$.

    Consider then the case
    $t^-_i=M_{k+1}-\eta/2^{i+1}-\valuation^-_i$, which holds when
    $\valuation_i+t_i\in(M_k+\eta/2^{i+1},M_{k+1}-\eta/2^{i+1})$ and
    $\prices(\location_i)=p^-$. Then, $\valuation_i\leq
    \valuation_i+t_i<M_{k+1}-\eta/2^{i+1}$. Hence, either
    $\valuation_i$ is $\eta/2^{i}$-close to a border (in particular
    when the clock has just been reset), in which case, by property
    \ref{item:near-convergent}, we have
    $\valuation^-_i=\valuation_i<M_{k+1}-\eta/2^{i+1}$, so that
    $t^-_i>0$. Or $\valuation_i\in(M_{k'}+\eta/2^i,M_{k'+1}-\eta/2^i)$
    with $k'\leq k$. By property \ref{item:far-convergent}, this
    implies that $\valuation^-_i\leq M_{k'+1}-\eta/2^{i}<
    M_{k+1}-\eta/2^{i+1}$, once again implying that $t^-_i>0$.

    The fact that $r^-$ is an $\eta$-convergent play is then directly
    implied by properties \ref{item:near-convergent} and
    \ref{item:far-convergent}.

  \item Property $r^-\sim r$ is also a direct consequence of
    properties \ref{item:near-convergent} and
    \ref{item:far-convergent}.

  \item It only remains to prove that $\Weight(r^-)\leq
    \Weight(r)$. Notice that, by induction, only matters the weight
    since the last index $ j$ where plays $r$ and $r^-$ have
    synchronized, i.e., where $\valuation_{ j}+t_{ j}=\valuation^-_{
      j}+t^-_{ j}$. For the sake of simplicity, we suppose in the
    following that $j=0$. In particular, we may suppose that $r$ does
    not contain reset transitions or positions $j$ such that
    $\valuation_j$ is $\eta/2^{j}$-close to borders, except
    possibly the very last one. Since there are no resets, we have
    $\valuation_j=\valuation_{j-1}+t_{j-1}$ for every $j>0$.

    We now consider separately the possibility of sets $\{p^-,p^+\}$.
    \begin{itemize}
    \item As a first case, consider that $p^-=-1$ and $p^+=+1$. We
      prove by induction over $0\leq j \leq i$ that
      \[\Weight(r^-[j+1])\leq
      \Weight(r[j+1])-|\valuation_{j+1}-\valuation^-_{j+1}|\,.\] To
      simplify the notations, we forget the weights of the actions in
      this proof, but notice that the same weights occur in $r$ and
      $r^-$ since these two plays are equivalent.
    
      \begin{itemize}
      \item If $j=0$, then we have supposed that $\valuation_0$ is
        $\eta$-close to a border, so that
        $\valuation^-_0=\valuation_0$. Let $k$ be such that
        $\valuation_1\in(M_k+\eta/2,M_{k+1}-\eta/2)$. Then, if
        $\prices(\location_0)=+1$, $t^-_0=\max
        (M_{k}+\eta/2-\valuation_0,0)$. If $\valuation_0\geq
        M_k+\eta/2$, then $\valuation^-_1=\valuation^-_0=\valuation_0$
        so that \[\Weight(r^-[1]) = 0 = \Weight(r[1])
        -|\valuation_1-\valuation_0| =
        \Weight(r[1])-|\valuation_1-\valuation^-_1|\,.\] %
        If $\valuation_0< M_k+\eta/2$, $\valuation^-_1=M_k+\eta/2\leq
        \valuation_1$ so that
        \[\Weight(r^-[1]) = \valuation^-_1-\valuation_0 =
        \Weight(r[1]) -\valuation_1+\valuation^-_1 =
        \Weight(r[1])-|\valuation_1-\valuation^-_1|\,.\] 

        Consider then the case where $\prices(\location_0)=-1$. In
        case $\valuation^-_1=M_{k+1}-\eta/2\geq \valuation_1$, we have
        \begin{align*}
          \Weight(r^-[1]) = -(\valuation^-_1-\valuation_0) &=
          \Weight(r[1]) -(\valuation^-_1-\valuation_1) \\
          &=\Weight(r[1])-
          |\valuation_1-\valuation^-_1|\,.
        \end{align*} 
        Otherwise, we know that $t_0=0$ and $\valuation^-_0\leq
        M_k+\eta$. Since $\valuation_0$ is not $\eta/2$-close from a
        border, but $\valuation_1=\valuation_0$ should be $\eta$-close
        from a border, we know that
        $\valuation_0\in(M_k+\eta/2,M_k+\eta]\cup[M_{k+1}-\eta,M_{k+1}-\eta/2)$.
        Since $\valuation^-_0=\valuation_0\leq M_k+\eta$, we know that
        $\valuation_0\in(M_k+\eta/2,M_k+\eta]$. Then, we obtain
        $\valuation^-_1=M_k+\eta/2\leq \valuation_0=\valuation_1$, so
        that
        \[\Weight(r^-[1]) = -(\valuation^-_1-\valuation_0) =
        \Weight(r[1])-|\valuation^-_1-\valuation_1|\,.\]

      \item Let us suppose that the property is proved for all indices
        less than or equal to $j$, and prove it for $j+1$. We let $k$
        be such that $\valuation_{j+1}=\valuation_j+t_j\in
        (M_k+\eta/2^{j+1},M_{k+1}-\eta/2^{j+1})$. We will distinguish
        four possible cases depending on $\prices(\location_j)$ and
        the relative order between $\valuation^-_{j+1}$ and
        $\valuation_{j+1}$.

        \begin{enumerate}
        \item We first suppose that $\prices(\location_j)=-1$. Then,
          \begin{align}
            &\Weight(r^-[j+1]) = \Weight(r^-[j])
            -(\valuation^-_{j+1}-\valuation^-_j) \notag\\
            &\hspace{1cm}\leq \Weight(r[j]) - |\valuation_j-\valuation^-_j| -
            \valuation^-_{j+1}+\valuation^-_j \quad\text{(Ind. Hyp.)}\notag\\
            &\Weight(r^-[j+1]) \leq \Weight(r[j+1])
            +(\valuation_{j+1}-\valuation_j)\notag \\ 
            &\hspace{4cm} -
            |\valuation_j-\valuation^-_j| -
            \valuation^-_{j+1}+\valuation^-_j\,.\label{eq:conv-1}
          \end{align}
          \begin{enumerate}
          \item In the case where $\valuation^-_{j+1}>
            \valuation_{j+1}$, we have
            $|\valuation_{j+1}-\valuation^-_{j+1}|=
            \valuation^-_{j+1}-\valuation_{j+1}$ so that
            \eqref{eq:conv-1} becomes
            \begin{multline*}
              \Weight(r^-[j+1]) \leq \Weight(r[j+1]) -
              |\valuation_{j+1}-\valuation^-_{j+1}|\\
              +(\valuation^-_j-\valuation_j)-|\valuation^-_j-\valuation_j|
            \end{multline*}
            which is less than or equal to $\Weight(r[j+1]) -
            |\valuation_{j+1}-\valuation^-_{j+1}|$ since
            $\valuation^-_j-\valuation_j\leq
            |\valuation^-_j-\valuation_j|$.

          \item Similarly, in the case where
            $\valuation^-_{j+1}<\valuation_{j+1}$, we have
            $|\valuation_{j+1}-\valuation^-_{j+1}|=
            \valuation_{j+1}-\valuation^-_{j+1}$. Notice that this
            necessarily implies that $\location_i\in\maxlocations$ and
            $t_j=0$, $\valuation^-_j\leq M_{k}+\eta$ and
            $t^-_j=\max(M_k+\eta/2^{j+1}-\valuation^-_j,0)$:
            otherwise, we would have $\valuation^-_{j+1} =
            M_{k+1}-\eta/2^{i+1}>\valuation_{j+1}$ that contradicts
            the hypothesis. If $t^-_{j}=0$, this implies that
            $\valuation^-_{j}=\valuation^-_{j+1}<
            \valuation_{j+1}=\valuation_{j}$,
              so that \eqref{eq:conv-1} can be rewritten
            \begin{align*}
            \Weight(r^-[j+1]) &\leq \Weight(r[j+1]) -
            \valuation_j+\valuation^-_j \\ &= \Weight(r[j+1]) -
            |\valuation_{j+1}+\valuation^-_{j+1}| \,.
            \end{align*}
            Otherwise, $t^-_j>0$ and we have
            $t^-_j=M_k+\eta/2^{j+1}-\valuation^-_j$. This is possible
            only if $\valuation^-_j\leq
            M_k+\eta/2^{j+1}<\valuation_{j+1}=\valuation_j$. Then,
            $|\valuation^-_j-\valuation_j| =
            \valuation_j-\valuation^-_j$ so that
            \begin{align*}
              \Weight(r^-[j+1]) &\leq \Weight(r[j+1]) -
              |\valuation_{j+1}+\valuation^-_{j+1}| -
              2|\valuation^-_j-\valuation_j| \\
              &\leq \Weight(r[j+1]) -
              |\valuation_{j+1}+\valuation^-_{j+1}| \,.
            \end{align*}
          \end{enumerate}
      
          \item Suppose then that $\prices(\location_j)=+1$. Then, a
            similar calculation gives
            \begin{align}
              &\Weight(r^-[j+1]) = \Weight(r^-[j])
              +\valuation^-_{j+1}-\valuation^-_j \notag\\
              &\hspace{1cm}\leq \Weight(r[j]) - |\valuation_j-\valuation^-_j| +
              \valuation^-_{j+1}-\valuation^-_j \quad\text{(Ind. Hyp.)}\notag\\
              &\Weight(r^-[j+1])  \leq \Weight(r[j+1])
              -(\valuation_{j+1}-\valuation_j) \notag\\
              &\hspace{4cm} -
              |\valuation_j-\valuation^-_j| +
              \valuation^-_{j+1}-\valuation^-_j\,.\label{eq:conv+1}
            \end{align}
            \begin{enumerate}
            \item Once again, if
              $\valuation^-_{j+1}<\valuation_{j+1}$, we have
              $|\valuation_{j+1}-\valuation^-_{j+1}|=
              \valuation_{j+1}-\valuation^-_{j+1}$ so that
              \eqref{eq:conv+1} is rewritten
            \[ \Weight(r^-[j+1]) \leq \Weight(r[j+1]) -
            |\valuation_{j+1}-\valuation^-_{j+1}|+
            \valuation_j-\valuation^-_j-|\valuation_j-\valuation^-_j|\]
            which is less than or equal to $\Weight(r[j+1]) -
            |\valuation_{j+1}-\valuation^-_{j+1}|$ since
            $\valuation_j-\valuation^-_j\leq
            |\valuation_j-\valuation^-_j|$.  

          \item Similarly, if $\valuation^-_{j+1}>\valuation_{j+1}$,
            we know by property \ref{item:far-convergent} that
            $\valuation^-_{j+1}\in\{M_{k+1}-\eta/2^{j+1},\valuation^-_j\}$. If
            $\valuation^-_{j+1}=M_{k+1}-\eta/2^{j+1}$, since
            $t^-_j=\max(M_k+\eta/2^{j+1}-\valuation^-_j,0)$, we know
            that $t^-_j=0$, i.e., in all case
            $\valuation^-_{j+1}=\valuation^-_j$.  Then,
            $\valuation^-_{j+1}=\valuation^-_j>\valuation_{j+1}\geq
            \valuation_j$. Knowing that
            $|\valuation_{j+1}-\valuation^-_{j+1}|=
            \valuation^-_{j+1}-\valuation_{j+1}$, \eqref{eq:conv+1}
            becomes
            \begin{align*}
              \Weight(r^-[j+1]) &\leq \Weight(r[j+1]) -
              |\valuation_{j+1}-\valuation^-_{j+1}|\\
              &\qquad + 2(\valuation^-_{j+1}-\valuation_{j+1})+2
              (\valuation_j-\valuation^-_j)\\
              &= \Weight(r[j+1]) -
              |\valuation_{j+1}-\valuation^-_{j+1}|+
              2(\valuation_j-\valuation_{j+1}) \\
              &\leq \Weight(r[j+1]) -
              |\valuation_{j+1}-\valuation^-_{j+1}|\,.
            \end{align*}
          \end{enumerate}
        \end{enumerate}
      \end{itemize}

      We finally have proved the property by induction. Notice in
      particular that this shows that $\Weight(r^-[j+1])\leq
      \Weight(r[j+1])$ for every $j$ with
      $\valuation^-_{j+1}\neq\valuation_{j+1}$. To conclude the proof
      of $\Weight(r^-)\leq \Weight(r)$, it remains to deal with the
      case of a possible last transition ending with
      $\valuation^-_{i+1}=\valuation_{i+1}$. Unless $i=0$, in which
      case we have $\Weight(r^-)=\Weight(r)$, we know by hypothesis
      that $\valuation^-_{i}\neq \valuation_{i}$. By the previous
      property, we have $\Weight(r^-[i])\leq
      \Weight(r[i])-|\valuation_i-\valuation^-_i|$. Then,
      $\Weight(r^-)=\Weight(r^-[i])+
      \prices(\location_i)(\valuation^-_{i+1}-\valuation^-_i)$ and
      $\Weight(r)=\Weight(r[i])+
      \prices(\location_i)(\valuation_{i+1}-\valuation_i)$. In the
      overall, we get
      \[\Weight(r^-)\leq \Weight(r)
      +\prices(\location_i)(\valuation_i-\valuation^-_i)-
      |\valuation_i-\valuation^-_i|\,.\] In all cases, we verify that
      $\prices(\location_i)(\valuation_i-\valuation^-_i)\leq
      |\valuation_i-\valuation^-_i|$, so that we have proved that
      $\Weight(r^-)\leq \Weight(r)$.

    \item We now consider the case where $p^-=0$ and $p^+=+1$ (the
      case $p^-=-1$ and $p^+=0$ is very similar, and not explained in
      details here). We prove another inequality by induction over
      $0\leq j \leq i$, namely that
      \[\Weight(r^-[j+1])\leq
      \Weight(r[j+1])-\max(\valuation_{j+1}-\valuation^-_{j+1},0)\,.\]
      The proof is very similar to the previous case, and we conclude
      as previously. \qed
    \end{itemize}
  \end{enumerate}
\end{proof}

We now go to the proof of Lemma~\ref{lem:convergent}. In case,
$\lowervalue(s)=+\infty$ nothing has to be done. We then consider the
case $\lowervalue(s)<+\infty$. Let
$\maxstrategy'\in\convergentmaxstrategies^\eta$. We now explain how to
construct a strategy $\maxstrategy\in\maxstrategies$ such that for all
state $s$
\[\inf_{\minstrategy'\in\convergentminstrategies^\eta}
\Weight(\outcomes(s,\minstrategy',\maxstrategy'))\leq
\inf_{\minstrategy\in\minstrategies}
\Weight(\outcomes(s,\minstrategy,\maxstrategy))\,.\] To prove such an
inequality, we will consider any strategy
$\minstrategy\in\minstrategies$ and construct a strategy
$\minstrategy'\in\convergentminstrategies^\eta$ such that
\[\Weight(\outcomes(s,\minstrategy',\maxstrategy'))\leq
\Weight(\outcomes(s,\minstrategy,\maxstrategy))\,.\]

Strategy $\maxstrategy$ follows $\maxstrategy'$ in case of
$\eta$-convergent plays. We must however extend it to deal with the
other plays faithfully. Let
$r=(\location_0,\valuation_0),(t_0,a_0),\ldots,
(\location_{i},\valuation_{i})$ be any finite play ending in a
location $\location_{i}$ of player 2, and
$r^-=(\location^-_0,\valuation^-_0),(t^-_0,a^-_0), \ldots,\allowbreak
(\location^-_{i},\valuation^-_{i})$ the play constructed as before. By
Lemma~\ref{lem:eta-convergent-runs}, we know that $r^-$ is an
$\eta$-convergent play. Hence, $\maxstrategy'(r^-)=(t'_i,a)$, for some
$t'_i\in\Rpos$, and there exists $k$ such that either
$\valuation^-_i+t'_i
\in\{M_k+\eta/2^{i+1}\}\cup[M_k-\eta/2^{i+1},M_k)$, or $t'_i=0$ and
$\valuation^-_i\in (M_k+\eta/2^{i+1},M_k+\eta]$. We let $t_i=
\max(\valuation^-_i+t'_i-\valuation_i,0)$ and $\maxstrategy(r) =
(t_i,a)$. Let $\tilde r$ (respectively, $r'$) be the play $r$
(respectively, $r^-$) extended with the step prescribed by
$\maxstrategy$ (respectively, $\maxstrategy'$). Then, we prove that
$r'$ matches the construction above starting from the run $\tilde r$,
i.e., $\tilde r^-=r'$. By construction, we only have to verify that
the value of $t'_i$ is consistent with the previous constructions,
i.e., $t'_i=t^-_i$.

\begin{lemma}\label{lem:consistency-convergent}
  We have $t'_i=t^-_i$.
\end{lemma}
\begin{proof}
The proof considers several cases.
  \begin{itemize}
  \item Suppose first that there exists $k'$ such that
    $|\valuation_i+t_i-M_{k'}|\leq \eta/2^{i+1}$. In case
    $\valuation^-_i+t'_i=\valuation_i+t_i$, we have $t'_i =
    \valuation_i+t_i-\valuation^-_i$ which is equal to $t^-_i$ (since
    $\valuation_i+t_i$ is $\eta/2^{i+1}$-close to a border implying
    that the first rule of the construction of $\tilde r^-$
    applies). Otherwise (i.e. when
    $\valuation^-_i+t'_i\ne\valuation_i+t_i$), we know that $t_i=0$
    (by definition of $t_i$ as
    $\max(\valuation^-_i+t'_i-\valuation_i,0)$) and that
    $\valuation_i>\valuation^-_i+t'_i\geq \valuation^-_i$. In
    particular, we have $\valuation_i\neq \valuation^-_i$. However,
    since $t_i=0$, $|\valuation_i-M_{k'}|\leq \eta/2^{i+1}\leq
    \eta/2^i$ so that we should have $\valuation^-_i=\valuation_i$,
    causing a contradiction.

  \item Suppose then that there exists $k'$ such that
    $\valuation_i+t_i\in(M_{k'}+\eta/2^{i+1},M_{k'+1}-\eta/2^{i+1})$.

    \begin{itemize}
    \item In case $\valuation_i+t_i=\valuation^-_i+t'_i$, since
      $\maxstrategy'$ is $\eta$-convergent and $\valuation^-_i+t'_i$
      not $\eta/2^{i+1}$-close to a border, we have $t'_i=0$, and
      $\valuation^-_i \in (M_{k}+\eta/2^{i+1},M_{k}+\eta]$: in
      particular, $k'=k$ and $\valuation_i+t_i=\valuation^-_i\in
      (M_{k}+\eta/2^{i+1},M_{k}+\eta]$. We know that
      $\valuation_i\sim\valuation^-_i$, hence there are two
      possibilities for the position of $\valuation_i\in(M_k,
      M_k+\eta]$.
      \begin{itemize}
      \item The first case is $\valuation_i\in(M_k,M_k+\eta/2^{i})$:
        by Lemma~\ref{lem:eta-convergent-runs}-1 (applied on
        $\valuation_i=\valuation_{i-1}+t_{i-1}$ since no reset
        transition may have been taken), we know that
        $\valuation^-_i=\valuation_i$. Since
        $\valuation^-_i=\valuation_i+t_i$, we have $t_i=0$. This shows
        that $t'_i=0$ is consistent with the construction of $\tilde
        r^-$ that sets $t^-_i=\valuation_i+t_i-\valuation^-_i=0$ in
        that case (since $\valuation_i+t_i\in(M_k,M_k+\eta/2^i)$).

      \item The second case is
        $\valuation_i\in(M_k+\eta/2^i,M_k+\eta]$: by
        Lemma~\ref{lem:eta-convergent-runs}-2-($a$), we know that
        $\valuation^-_i\in\{M_k+\eta/2^i, M_{k+1}-\eta/2^i,
        \valuation^-_{i-1}\}$. It is not possible neither that
        $\valuation^-_i=M_k+\eta/2^i<\valuation_i$ (because
        $\valuation^-_i=\valuation_i+t_i\geq \valuation_i$), nor that
        $\valuation^-_i=M_{k+1}-\eta/2^i>M_k+\eta$ (because we know
        that $\valuation^-_i\leq M_k+\eta$). Hence, we have
        $\valuation^-_i=\valuation^-_{i-1}$ (and thus
        $t^-_{i-1}=0$). Moreover, by
        Lemma~\ref{lem:eta-convergent-runs}-\ref{item:far-convergent}-($c$),
        since $t^-_{i-1}=0$ and $\valuation^-_{i-1}=\valuation^-_i\leq
        M_k+\eta$, we have $\valuation^-_{i}=\valuation^-_{i-1}\leq
        \valuation_{i-1}+t_{i-1}=\valuation_i$. Since, we also have
        $\valuation_i\leq \valuation_i+t_i=\valuation^-_i$, we obtain
        that $\valuation_i=\valuation^-_i$, and
        $t_i=\valuation^-_i+t'_i-\valuation_i=0$. In case
        $\prices(\location_i)=p^+$, this shows that $t'_i=0$ is
        consistent with the construction of $\tilde r^-$ that sets
        $t^-_i=\max(M_k+\eta/2^{i+1}-\valuation^-_i,0)=0$ (since
        $\valuation^-_i=\valuation_i>M_k+\eta/2^{i}>M_k+\eta/2^{i+1}$). In
        case $\prices(\location_i)=p^-$, since we are in the case
        where $\location_i\in \maxlocations$ (since $\location_i$ is a
        location where $\maxstrategy$ needs to be defined), $t_i=0$
        and $\valuation^-_i\leq M_k+\eta$, the choice $t'_i=0$ is also
        consistent with the construction of $\tilde r^-$ which sets
        $t^-_i=\max(M_k+\eta/2^{i+1}-\valuation^-_i,0)=0$ (once again,
        because $\valuation^-_i>M_k+\eta/2^{i+1}$).
      \end{itemize}

    \item The last case is when $\valuation_i+t_i\neq
      \valuation^-_i+t'_i$, implying that $t_i=0$ (by definition of
      $t_i$ as $\max(\valuation^-_i+t'_i-\valuation_i,0)$). This
      implies that $\valuation_i>\valuation^-_i+t'_i\geq
      \valuation^-_i$. It is not possible that
      $\valuation_i\in(M_{k'}+\eta/2^{i+1},M_{k'}+\eta/2^i]\cup
        [M_{k'+1}-\eta/2^i,M_{k'+1}-\eta/2^{i+1})$, since otherwise we
          would have $\valuation^-_i=\valuation_i$, by
          Lemma~\ref{lem:eta-convergent-runs}-1 (applied to
          $\valuation^-_i=\valuation^-_{i-1}+t^-_{i-1}=
          \valuation_{i-1}+t_{i-1}=\valuation_{i}$). Hence, we have
          $\valuation_i\in(M_{k'}+\eta/2^i, M_{k'+1}-\eta/2^i)$. In
          particular, by Lemma~\ref{lem:eta-convergent-runs}-2-($a$),
          we know that $\valuation^-_i\in\{M_{k'}+\eta/2^i,
          M_{k'+1}-\eta/2^i, \valuation^-_{i-1}\}$. There are two
          possibilities now, depending on $t'_i$ taken from the fact
          that $\maxstrategy'$ is $\eta$-convergent.
      \begin{itemize}
      \item The first possibility is $\valuation^-_i+ t'_i
        \in\{M_k+\eta/2^{i+1}\}\cup[M_k-\eta/2^{i+1},M_k)$. Notice
        that $\valuation^-_i\sim \valuation_i$ and
        $\valuation^-_i+t'_i\in[\valuation^-_i,\valuation_i)$, so that
        $\valuation^-_i+t'_i\sim \valuation_i$. Hence,
        $\valuation^-_i+t'_i\in (M_{k'},M_{k'+1})$. Knowing that
        $\valuation^-_i+ t'_i <\valuation_i<M_{k'+1}-\eta/2^{i}$, and
        that $|\valuation^-_i+t'_i-M_k|\leq \eta/2^{i+1}$, we conclude
        that $k'=k$ and $\valuation^-_i+ t'_i
        \in(M_{k'},M_{k'}+\eta/2^{i+1}]$. It is therefore only
        possible that $\valuation^-_i+t'_i=M_{k'}+\eta/2^{i+1}$, i.e.,
        $t'_i=M_{k'}+\eta/2^{i+1}-\valuation^-_i$. This choice is
        consistent with the construction of $\tilde r^-$, whatever the
        price of $\location_i$, which sets $t^-_i =
        \max(M_{k'}+\eta/2^{i+1}-\valuation^-_i,0)=
        M_{k'}+\eta/2^{i+1}-\valuation^-_i$ (since $\valuation^-_i\leq
        M_{k'}+\eta/2^{i+1}$): in particular if
        $\prices(\location_i)=p^-$, we are indeed in the case
        $\location_i\in\maxlocations$, $t_i=0$ and $\valuation^-_i\leq
        M_{k'}+\eta$.

      \item The second possibility is that $t'_i=0$ and
        $\valuation^-_i\in(M_k+\eta/2^{i+1},M_k+\eta]$, in which case
        we again deduce from $\valuation_i\sim \valuation^-_i$ that
        $k'=k$. We have
        $t^-_i=\max(M_{k'}+\eta/2^{i+1}-\valuation^-_i,0)=0=t'_i$
        (since $\valuation^-_i\geq M_{k'}+\eta/2^{i+1}$), whatever the
        price of $\location_i$: once again, if
        $\prices(\location_i)=p^-$, we are indeed in the case
        $\location_i\in\maxlocations$, $t_i=0$ and $\valuation^-_i\leq
        M_{k'}+\eta$.\qed
      \end{itemize}
    \end{itemize}
  \end{itemize}
\end{proof}

We now consider any strategy $\minstrategy\in\minstrategies$, and
construct a strategy $\minstrategy'\in\convergentminstrategies^\eta$ such
that $\outcomes(s,\minstrategy',\maxstrategy')=
\outcomes(s,\minstrategy,\maxstrategy)^-$ for every state $s$
$\eta$-close to a border. From Lemma~\ref{lem:eta-convergent-runs},
we will then get
\[\Weight(\outcomes(s,\minstrategy',\maxstrategy'))\leq
\Weight(\outcomes(s,\minstrategy,\maxstrategy))\,,\] which will enable
us to conclude. We let $\outcomes(s,\minstrategy,\maxstrategy)^-=
(\location^-_0,\valuation^-_0),(t^-_0,a^-_0), \ldots,\allowbreak
(\location^-_{n},\valuation^-_{n}),\ldots$ with
$s=(\location^-_0,\valuation^-_0)$ $\eta$-close to a border. Then,
we first define $\minstrategy'$ over the finite plays
$\outcomes(s,\minstrategy,\maxstrategy)^-[n]$ with $n\in \N$, by
letting
\[\minstrategy'(\outcomes(s,\minstrategy,\maxstrategy)^-[n]) =
(t^-_n,a^-_n)\,.\] Notice first that this strategy verifies
$\outcomes(s,\minstrategy',\maxstrategy')[n]=
\outcomes(s,\minstrategy,\maxstrategy)[n]^-$ for every state $s$
$\eta$-close to a border, by induction on $n\in\N$. In fact, in case
$\outcomes(s,\minstrategy,\maxstrategy)^-[n]$ ends with a state of
player 2, the equation holds by construction of $\maxstrategy$, and in
case it ends with a state of player 1, by construction of
$\minstrategy'$.

Once built on these finite plays, it is possible to extend
$\minstrategy'$ as an $\eta$-convergent strategy defined over every
play: in particular, the $\eta$-region-uniformity is possible, since,
if $\outcomes(s,\minstrategy,\maxstrategy)^-[n]\sim_\eta
\outcomes(s',\minstrategy,\maxstrategy)^-[n]$ (with $s$ and $s'$
$\eta$-close to a border), we have
$\minstrategy'(\outcomes(s,\minstrategy,\maxstrategy)^-[n])
\sim_\eta\minstrategy'(\outcomes(s',\minstrategy,\maxstrategy)^-[n])$
(induced by Lemma~\ref{lem:eta-convergent-runs}-4) validating the
definition of $\eta$-region-uniform strategies. The $\eta$-convergence
is ensured by Lemma~\ref{lem:eta-convergent-runs}-3.

This concludes the proof of $\lowervalue(s)\geq
\uniformlowervalue^\eta(s)$ in case $\lowervalue(s)>-\infty$.

\medskip
Finally, in case $\lowervalue(s)=-\infty$, we have to show that
$\convergentlowervalue^\eta(s)=-\infty$ too. Notice that
$\lowervalue(s)=-\infty$ means that for all strategy
$\maxstrategy\in\maxstrategies$ of player 2, we have
\[\inf_{\minstrategy\in\minstrategies}
\Weight(\outcomes(s,\minstrategy,\maxstrategy))=-\infty\,,\] i.e.,
player 1 has a sequence of strategies ensuring the reachability of the
goal with smaller and smaller prices. Hence, let
$\maxstrategy'\in\convergentmaxstrategies^\eta$ be an
$\eta$-convergent strategy for player 2, and $M\in\R$ be any
constant. We construct as previously a strategy
$\maxstrategy\in\maxstrategies$ for player 2. From the fact that
$\inf_{\minstrategy\in\minstrategies}
\Weight(\outcomes(s,\minstrategy,\maxstrategy))<M$, we know the
existence of a strategy $\minstrategy\in\minstrategies$ so that
$\Weight(\outcomes(s,\minstrategy,\maxstrategy))<M$. In particular,
this price is finite so that the previous construction allows us to
obtain an $\eta$-convergent strategy $\minstrategy'$ verifying that
\[\Weight(\outcomes(s,\minstrategy',\maxstrategy')) \leq
\Weight(\outcomes(s,\minstrategy,\maxstrategy))<M\,.\] This proves
that $\convergentlowervalue(s)=-\infty$.

\subsection{Proof of Lemma~\ref{lem:translation}}
\label{app:lem:translation}

The proof uses the fact that we can translate a strategy in
$\abstrarena$ into $\eta$-region-uniform/$\eta$-convergent
strategies in the original game $\arena$ and vice versa.

Since we assume that $\Value_{\abstrarena}((\location,\{M_k\}))$ is
finite, let $\minfstrategy^*$ and $\maxfstrategy^*$ the optimal
strategies for both players provided by
Theorem~\ref{thm:optimal-strategy}. Strategy $\minfstrategy^*$ uses a
finite memory whereas $\maxfstrategy^*$ is memoryless. They verify,
for all $\location$, and $0\leq k\leq K$:
\[\Value_{\abstrarena}((\location,\{M_k\})) =
\Weight(\outcomes((\location,\{M_k\}),
\minfstrategy^*,\maxfstrategy^*)) \,.\]

\medskip We first show inequality
$\uniformuppervalue^\eta_\arena((\location,M_k))-\varepsilon\leq
\Value_{\abstrarena}((\location,\{M_k\}))$. In case
$\uniformuppervalue^\eta_\arena((\location,M_k))=-\infty$, the
inequality is trivially verified. Since
$\Value_{\abstrarena}((\location,\{M_k\}))$ is finite for every
$\location$, and $0\leq k\leq K$, we know that every play where player
1 follows strategy $\minfstrategy^*$ reaches the target set of
states. Let $L\in \N$ be the maximum $\StepsToReach(r)$ for every such
play $r$: notice that $L$ is finite since the game arena is finite and
that $\minfstrategy^*$ has finite memory. For all $\varepsilon>0$, we
let $\eta$ be a positive rational number less than $\varepsilon/(2L)$
and we explain how to construct an $\eta$-region-uniform strategy
$\minstrategy^{(\varepsilon)}$ of $\arena$ so that
\begin{equation}\label{eq:translation-sigma-tau}
  \sup_{\maxstrategy\in\uniformmaxstrategies^\eta(\arena)}
  \Weight(\outcomes((\location,M_k),\minstrategy^{(\varepsilon)},
  \maxstrategy)) \leq
  \Weight(\outcomes((\location,\{M_k\}),\minfstrategy^*,\maxfstrategy^*)) 
  +\varepsilon 
\end{equation} 
Since $\minstrategy^{(\varepsilon)}$ will only have to play against an
$\eta$-region-uniform strategy $\maxstrategy$, we may define only
$\minstrategy^{(\varepsilon)}$ on finite plays $r$ only visiting
$\eta$-regions in $\intervals^\eta_\arena$, i.e., that only stops at
distance at most $\eta$ from the integers $M_k$. Such a play $r$ can
indeed be translated in a play $\rho$ of $\abstrarena$, by simply
translating each state $(\location_i,\valuation_i)$ (with
$\valuation_i$ at distance at most $\eta$ from some $M_k$ by
assumption) it encounters into $(\location_i,I_i)$ with
$I_i\in\intervals^\eta_\arena$ the interval containing
$\valuation_i$. Then, we let $(J,a)=\minfstrategy^*(\rho)$. Let also
denote by $(\location',\valuation)$ the last state of $r$, as well as
$I$ the interval of $\intervals^\eta_\arena$ containing
$\valuation$. Notice that $I\leqinterval J$. In case $J=I$, we let
$\minstrategy^{(\varepsilon)}(r)=(0,a)$ forcing player 1 to play
immediately. In case $J$ is a singleton $\{M_k\}$, we let
$\minstrategy^{(\varepsilon)}(r) = (M_k-\valuation,a)$. In case
$J=[M_k-\eta,M_k)$, we let $\minstrategy^{(\varepsilon)}(r) =
(M_k-\eta-\valuation,a)$. Finally, in case $J=(M_k,M_k+\eta]$, we let
$\minstrategy^{(\varepsilon)}(r) = (M_k+\eta-\valuation,a)$. Notice
that in all case, the fact that $\valuation\in I$ implies that the
delay prescribed in $\minstrategy^{(\varepsilon)}(r)$ is not less than
$0$.

We now prove \eqref{eq:translation-sigma-tau}. For that purpose, we
consider a strategy
$\maxstrategy\in\uniformmaxstrategies^\eta(\arena)$. We reconstruct
from it a strategy $\maxfstrategy\in\maxstrategies(\abstrarena)$ such
that $\outcomes((\location,\{M_k\}),\minfstrategy^*,\maxfstrategy)$ is
the sequence of $\eta$-regions visited by the play
$\outcomes((\location,M_k),\minstrategy^{(\varepsilon)},\maxstrategy)$,
which stays $\eta$-close to borders, since
$\minstrategy^{(\varepsilon)}$ and $\maxstrategy$ are
$\eta$-region-uniform strategies. 
The following lemma compares the weight of these two plays.

\begin{lemma}\label{lem:weights-abstraction}
  Let $r$ be a play of $\arena$ staying $\eta$-close to borders and
  $\tilde r$ its abstraction in terms of $\eta$-regions. Then,
  \[|\Weight(\tilde r)-\Weight(r)|\leq 2\eta\,\StepsToReach(\tilde
  r)\,.\]
\end{lemma}
\begin{proof}
  Nothing has to be done in case $\Weight(r)$ (or equivalently
  $\Weight(\tilde r)$) is equal to $+\infty$. Hence, we can suppose
  that these two runs have a finite length.  We show by induction on
  $0\leq n\leq \StepsToReach(\tilde r)=\StepsToReach(r)$ that
  \[|\Weight(\tilde r[n])-\Weight(r[n])|\leq 2\eta\, n\,.\]

  For $n=0$, nothing has to be proved since $\Weight(\tilde
  r[0])=\Weight(r[0])=0$.

  Suppose that the property holds until step $n<\StepsToReach(\tilde
  r)=\StepsToReach(r)$. We now prove it for index $n+1$. Let
  $(\location,\valuation)$ and $(\location,I)$ be the last states of
  $r[n]$ and $\tilde r[n]$, respectively. We denote by $t$ the delay
  taken in $r$ at step $n$, and $J$ the choice of successor in $\tilde
  r$ at step $n$. Since the same action occurs at step $n$, we have
  \[\Weight(\tilde r[n+1])-\Weight(r[n+1]) = \Weight(\tilde
  r[n])-\Weight(r[n]) + \prices(\location)(d(I,J)-t)\,.\]
  We clearly have $|d(I,J)-t|\leq 2\eta$ so that, by induction
  hypothesis, 
  \[|\Weight(\tilde r[n+1])-\Weight(r[n+1])|\leq 2\eta\, n + 2\eta =
  2\eta\,(n+1)\,.\]

  Hence, the property is proved by induction.\qed
\end{proof}

By Lemma~\ref{lem:weights-abstraction}, we have 
\begin{align*}
  \Weight(\outcomes((\location,M_k),
  \minstrategy^{(\varepsilon)},\maxstrategy)) &\leq
  \Weight(\outcomes((\location,\{M_k\}),
  \minfstrategy^*,\maxfstrategy))+{}\\
  &\hspace{1cm}2\eta\, \StepsToReach(\outcomes((\location,\{M_k\}),
  \minfstrategy^*,\maxfstrategy))\\ & \leq
  \Weight(\outcomes((\location,\{M_k\}),
  \minfstrategy^*,\maxfstrategy))+2\eta\, L \\
  &\leq \Weight(\outcomes((\location,\{M_k\}),
  \minfstrategy^*,\maxfstrategy^*))+\varepsilon
\end{align*}
where the last inequality comes from the definition of $\eta$, and
$\maxfstrategy^*$ is the optimal strategy for player 2 in
$\abstrarena$. In particular, notice that this shows that
$\Weight(\outcomes((\location,M_k),
\minstrategy^{(\varepsilon)},\maxstrategy))$, and hence
$\uniformuppervalue^\eta_\arena((\location,M_k))$, is less than
$+\infty$.

\medskip We then show inequality
$\Value_{\abstrarena}((\location,\{M_k\})) \leq
\convergentlowervalue^\eta_\arena((\location,M_k))+\varepsilon$.  In case
$\convergentlowervalue^\eta_\arena((\location,M_k))=+\infty$, the
inequality is trivially verified. For all $\varepsilon>0$, we let
$\eta$ be a positive rational number less than $\varepsilon/3$ and we
explain how to construct (from $\maxfstrategy^*$) an
$\eta$-convergent strategy $\maxstrategy^{(\varepsilon)}$ of
$\arena$ so that
\begin{equation}\label{eq:lower-value}
  \Weight(\outcomes((\location,\{M_k\}),\minfstrategy^*,\maxfstrategy^*)) 
  -\varepsilon \leq \inf_{\minstrategy\in\convergentminstrategies^\eta(\arena)}
  \Weight(\outcomes((\location,M_k),\minstrategy,\maxstrategy^{(\varepsilon)}))\,.
\end{equation} 

The construction of $\maxstrategy^{(\varepsilon)}$ is inspired from
the previous case, but special care has to be made to prevent player 1
from decreasing its cost in $\arena$ by accumulating errors made by
player 2. 
Once again, since $\maxstrategy^{(\varepsilon)}$ will be evaluated
against an $\eta$-convergent play of player 1, we only define it on finite
plays $r$ that are $\eta$-convergent. Let
$(\location',\valuation)$ be the last configuration of $r$, $n$ its
length and $I\in\intervals^\eta_\arena$ the interval containing
$\valuation$. Let $(J,a)=\maxfstrategy^*(\location',I)$ (remember that
$\maxfstrategy^*$ is memoryless so that only the last configuration of
the play is useful). Once again, we necessarily have $I\leqinterval
J$. In case $J$ is a singleton $\{M_k\}$, we let
$\maxstrategy^{(\varepsilon)}(r)=(M_k-\valuation,a)$. In case
$J=[M_k-\eta,M_k)$, we let
$\maxstrategy^{(\varepsilon)}(r)=(\max(M_k-\eta/2^{n+1}-\valuation,0),a)$. In
case $J=(M_k,M_k+\eta]$, we let
$\maxstrategy^{(\varepsilon)}(r)=(\max(M_k+\eta/2^{n+1}-\valuation,0),a)$. The
use of a maximum operator to define the delay in the two last cases
permits to delay $0$ in case $\valuation$ is either too close or too
far from $M_k$ to go exactly at distance $\eta/2^{n+1}$ from $M_k$.

Then, to prove inequality \eqref{eq:lower-value}, we consider any
$\eta$-convergent strategy $\minstrategy$ of player 1. As in the
previous case, we can construct from $\minstrategy$ a strategy
$\minfstrategy\in\minstrategies(\abstrarena)$ verifying that
$\outcomes((\location,\{M_k\}),\minfstrategy,\maxfstrategy^*)$ is the
sequence of $\eta$-regions visited by the play
$\outcomes((\location,M_k),\minstrategy,\maxstrategy^{(\varepsilon)})$.
The following lemma permits to compare the weight of these two plays.

\begin{lemma}
  Let $r=(\location_0,\valuation_0),(t_0,a_0),\ldots,
  (\location_{i+1},\valuation_{i+1}),\ldots$ be a play of $\arena$,
  such that $\valuation_0=M_{k'}$ for some $k'$, and $\tilde r$ its
  abstraction in terms of $\eta$-regions. Suppose that $r$ is
  $\eta$-convergent. Then, 
  \[|\Weight(\tilde r)-\Weight(r)|\leq 3\eta\,.\]
\end{lemma}
\begin{proof}
  We define the \emph{latency} $\lambda_n$ of an index $n\leq
  \StepsToReach(\tilde r)=\StepsToReach(r)$ as the greatest index
  $0<m\leq n$ such that $t_{m-1}\neq 0$ or the clock has been reset by
  action $a_{m-1}$ in location $\location_{m-1}$, or $0$ if such an
  index $m$ does not exist. Notice that if $n<n'$ and
  $\lambda_n=\lambda_{n'}$, then we have $t_m=0$ for every $m\in
  \{n,\ldots,n'-1\}$, and $\valuation_n=\valuation_{n'}$. We also know
  that $\valuation_{\lambda_n}$ is $\eta/2^{\lambda_n+1}$-close to a
  border: indeed, either $t_{\lambda_n-1}\neq 0$ (which permits to
  conclude by definition of $\eta$-convergent plays), or the clock has
  been reset by action $a_{\lambda_n-1}$, in which case
  $\valuation_{\lambda_n}=0$.

  We then prove by induction on $0\leq n\leq \StepsToReach(\tilde
  r)=\StepsToReach(r)$ that
  \[|\Weight(\tilde r[n])-\Weight(r[n])|\leq 3\eta\left(1- \frac
    1{2^{\lambda_n}}\right)\,.\]

  The property is trivially verified for $n=0$, since $\lambda_n=0$,
  and $\Weight(\tilde r[n])=\Weight(r[n])=0$.

  Suppose now that the property holds for $n<\StepsToReach(\tilde
  r)=\StepsToReach(r)$. We now prove it for $n+1$. We split our study
  with respect to the value of $\lambda_{n+1}$. 
  \begin{itemize}
  \item If $\lambda_{n+1}=n+1$, then we know that $t_n\neq 0$ or that
    the clock has just been reset. Denote by $I$ the $\eta$-region
    containing $\valuation_n$ and $J$ the $\eta$-region containing
    $\valuation_n+t_n$. 
    \begin{itemize}
    \item If $t_n=0$, then $d(I,J)=0$, and
      \begin{align*}
        |\Weight(\tilde r[n+1])-\Weight(r[n+1])| &= |\Weight(\tilde
        r[n])-\Weight(r[n])|\\
        &\leq 3\eta\left(1-\frac
          1{2^{\lambda_n}}\right) \quad \text{(Ind. Hyp.)}\\
        &\leq 3\eta\left(1-\frac 1{2^{\lambda_{n+1}}}\right)
      \end{align*}
      since $\lambda_{n+1}=n+1>n\geq\lambda_n$. 
    \item Otherwise, $t_n\neq 0$ so that, $\valuation_{n}+t_n$ is
      $\eta/2^{n+1}$-close to a border (since the play is supposed to
      be $\eta$-convergent). By reasoning on the transition from step
      $\lambda_n-1$ to step $\lambda_n$, we also know that
      $\valuation_n$ is $\eta/2^{\lambda_n}$-close to a border. By
      definition of $d(I,J)$, we obtain $|d(I,J)-t_n|\leq
      \eta(1/2^{n+1} + 1/2^{\lambda_n})$. In the overall, we get
      \begin{align*}
        &|\Weight(\tilde r[n+1])-\Weight(r[n+1])| \\
        &\hspace{2cm}\leq |\Weight(\tilde
        r[n])-\Weight(r[n])|+|\prices(\location_n)(d(I,J)-t_n)| \\
        &\hspace{2cm}\leq 3\eta\left(1-\frac 1{2^{\lambda_n}}\right) +
        \eta\left(\frac 1{2^{n+1}} + \frac 1{2^{\lambda_n}}\right)\\
        &\hspace{2cm}=\eta \left(3 - \frac 1{2^{\lambda_n-1}} + \frac
          1{2^{n+1}}\right)\\
        &\hspace{2cm}\leq \eta\left(3-\frac 3
          {2^{n+1}}\right)=3\eta\left(1-\frac 1
          {2^{\lambda_{n+1}}}\right)
      \end{align*}
      the last inequality coming from the fact that $\lambda_n\leq n$,
      so that $1/2^{n-1}\leq 1/2^{\lambda_n-1}$.
    \end{itemize}

  \item If $\lambda_{n+1}<n$, then we know that
    $\lambda_{n+1}=\lambda_n$ and $t_n=0$, so that 
    \begin{align*}
      |\Weight(\tilde r[n+1])-\Weight(r[n+1])| &= |\Weight(\tilde
      r[n])-\Weight(r[n])|\\
      &\leq 3\eta\left(1-\frac
        1{2^{\lambda_n}}\right) \quad \text{(Ind. Hyp.)}\\
      &\leq 3\eta\left(1-\frac 1{2^{\lambda_{n+1}}}\right)\,.
    \end{align*}
  \end{itemize}
  
  The property is proved by induction, and we conclude as in
  Lemma~\ref{lem:weights-abstraction}.\qed
\end{proof}

Since
$\outcomes((\location,M_k),\minstrategy,\maxstrategy^{(\varepsilon)})$
is an $\eta$-convergent play by construction, this lemma permits to
conclude that
\begin{equation*}
  \Weight(\outcomes((\location,\{M_k\}),\minfstrategy,\maxfstrategy^*)) 
  -\varepsilon \leq 
  \Weight(\outcomes((\location,M_k),
  \minstrategy,\maxstrategy^{(\varepsilon)}))\,. 
\end{equation*} 
Once again, notice that this shows that
$\Weight(\outcomes((\location,M_k),
\minstrategy,\maxstrategy^{(\varepsilon)}))$, and hence also
$\convergentlowervalue^\eta_\arena((\location,M_k))$, is greater than
$-\infty$.

\subsection{Proof of the infinite case of Corollary~\ref{cor:pseudo}}
\label{app:lem:infinite}

If $\Value_{\abstrarena}((\location,\{M_k\}))=-\infty$, then it is
sufficient to prove that
$\uppervalue_\arena((\location,M_k))=-\infty$. To do so, let $M\in\R$
and $\minfstrategy$ a strategy of player 1 such that
$\pricestrategy((\location,\{M_k\}),\minfstrategy)\leq M$. From the
construction in \cite{priced-games}, we know that we can choose
$\minfstrategy$ with the following structure: the strategy follows the
mean-payoff strategy, storing in its memory the cost accumulated so
far, and stops the mean-payoff execution to reach a target state
whenever the cost is low enough. In particular, the length of this
computation is bounded by a linear function over $M$ of the shape $CM$
with $C$ only depending on $\arena$. Following the same reconstruction
as the one given in Lemma~\ref{lem:translation}, with $\eta\leq
\varepsilon/(CM\max|\edgeweights|)$, we can map strategy
$\minfstrategy$ into an $\eta$-region-uniform strategy
$\minstrategy^{(\varepsilon)}$ of $\arena$. In the overall, whatever
strategy $\maxstrategy$ of player 2, we can build $\maxfstrategy$ in
$\abstrarena$ that mimicks the play following the profile
$(\minstrategy,\maxstrategy)$ of strategies in $\arena$. However, the
length of this play will necessarily be bounded by $CM$, so that the
difference of weight between the two plays is bounded by
$\varepsilon$. This suffices to prove that
$\uniformuppervalue^\eta_\arena((\location,M_k))\leq M$. By
Lemma~\ref{lem:uniform}, this shows that
$\uppervalue_\arena((\location,M_k))\leq M$. Since this holds for
every $M$, we conclude that
$\uppervalue_\arena((\location,M_k))=-\infty$.

Finally, if $\Value_{\abstrarena}((\location,\{M_k\}))=+\infty$, then
we prove that
$\convergentlowervalue^\eta_\arena((\location,M_k))=+\infty$, that is
sufficient for the result by Lemma~\ref{lem:convergent}. We will
indeed show that player 2 has a strategy $\maxstrategy$ to ensure a
value $+\infty$, i.e., ensuring that player 1 cannot reach the target
set of states. We again use our knowledge on the strategy
$\maxfstrategy$ ensuring a cost $+\infty$ in $\abstrarena$:
following the construction of Theorem~\ref{thm:overall}, we know that
a memoryless strategy is enough. Similarly to the previous cases, we
can reconstruct a strategy $\maxstrategy$ in $\arena$. In particular,
it automatically ensure that the target set of states is not reachable
by any strategy of player 1, which permits to conclude easily.

\section{Detailed undecidability proofs}
\subsection{Counter machines}
A two-counter machine $M$ is a tuple $(L, C)$ where ${L = \set{\ell_0,
    \ell_1, \ldots, \ell_n}}$ is the set of instructions---including a
distinguished terminal instruction $\ell_n$ called HALT---and ${C =
  \set{c_1, c_2}}$ is the set of two \emph{counters}.  The
instructions $L$ are one of the following types:
\begin{enumerate}
\item (increment $c$) $\ell_i : c := c+1$;  goto  $\ell_k$,
\item (decrement $c$) $\ell_i : c := c-1$;  goto  $\ell_k$,
\item (zero-check $c$) $\ell_i$ : if $(c >0)$ then goto $\ell_k$
  else goto $\ell_m$,
\item (Halt) $\ell_n:$ HALT.
\end{enumerate}
where $c \in C$, $\ell_i, \ell_k, \ell_m \in L$.
A configuration of a two-counter machine is a tuple $(l, c, d)$ where
$l \in L$ is an instruction, and $c, d$ are natural numbers that specify the value
of counters $c_1$ and $c_2$, respectively.
The initial configuration is $(\ell_0, 0, 0)$.
A run of a two-counter machine is a (finite or infinite) sequence of
configurations $\seq{k_0, k_1, \ldots}$ where $k_0$ is the initial
configuration, and the relation between subsequent configurations is
governed by transitions between respective instructions.
The run is a finite sequence if and only if the last configuration is
the terminal instruction $\ell_n$.
Note that a two-counter  machine has exactly one run starting from the initial
configuration. 
The \emph{halting problem} for a two-counter machine asks whether 
its unique run ends at the terminal instruction $\ell_n$.
It is well known~(\cite{Min67}) that the halting problem for
two-counter machines is undecidable.

\subsection{Constrained-price reachability} 

\begin{theorem}
  Deciding the existence of a strategy for constrained-reachability
  objective $\ReachOb({\bowtie} 1)$ with
  ${\bowtie}\in\{\leq,<,=,>,\geq\}$ is undecidable for PTGs with two
  clocks or more.
\end{theorem}

\begin{proof}
  We prove that the existence of a strategy for
  constrained-reachability objective $\ReachOb({=}1)$ is
  undecidable. The proofs for other objectives follow a similar
  approach: we outline the changes at the end of this proof. In order
  to obtain the undecidability result, we use a reduction from the
  halting problem for two counter machines. Our reduction uses a PTG
  with arbitrary price-rates, and zero prices on labels, along with
  two clocks $x_1, x_2$.
  
  We specify a module for each instruction of the counter machine.  On
  entry into a module for increment/decrement/zero check, we always
  have that $x_1=\frac{1}{5^{c_1}7^{c_2}}$ and $x_2=0$ where $c_1$
  (resp. $c_2$) is the value of counter $C_1$ (resp. $C_2$).  Given a
  two counter machine, we construct a PTG $\arena$ whose building
  blocks are the modules for instructions.  The role of Player 1 will
  be to simulate faithfully the machine by choosing appropriate delays
  to adjust the clocks to reflect changes in counter values.  Player 2
  will have the opportunity to verify that Player~1 did not cheat
  while simulating the machine.  We shall now present the modules for
  decrement, increment and zero-check instructions of the two counter
  machine.
  
\smallskip
  \paragraph{Simulate decrement instruction} 
  Fig.~\ref{fig:equalone} gives the complete module for the
  instruction to decrement $C_1$.  Let us denote by
  $x_{old}=\frac{1}{5^{c_1}7^{c_2}}$ the value of $x_1$ while entering
  the module. At the location $\ell_{k+1}$ of the module, $x_1 = x_{new}$
  should be $5x_{old}$ to correctly decrement counter $C_1$.

  At location $\ell_k$, Player 1 spends a nondeterministic amount of
  time $k= x_{new} - x_{old}$ such that $x_{new} =
  5x_{old}+\varepsilon$. To correctly decrement $C_1$, $\varepsilon$
  should be 0, and $k$ must be $\frac{4}{5^{c_1}7^{c_2}}$.  At
  location Check, Player 2 could choose to go to $Go$ (in order to
  continue the simulation of the machine) or to go to the widget
  $\text{WD}_1$, if he suspects that $\varepsilon \neq 0$. If Player 2
  spends $t$ time in the location Check before proceeding to
  $\ell_{k+1}$, then Player 1 can enter the location Abort from $Go$
  (to abort the simulation), spend $1+t$ time in location Abort and
  reach a target $T_1$ with cost $=1$ (and thus achieve his
  objective). However, if $t=0$ then entering location Abort will make
  the cost to be $>1$ (which is losing for Player 1). In this case,
  player 1 will prefer entering $\ell_{k+1}$ from $Go$.
   
  If player 2 spends $t$ time in location Check, and enters widget
  $\text{WD}_1$, then the cost upon reaching the target in the widget
  $\text{WD}_1$ is $1+\varepsilon$ which is $1$ iff $\varepsilon=0$
  (see Table~\ref{tab1}).
  
  \begin{table}[tbp]
    \caption{Cost Incurred in $\text{WD}_1$}
    \label{tab1}
    \centering
    \begin{tabular}{|x{.2\linewidth}|x{.2\linewidth}
        |x{.1\linewidth}|x{.1\linewidth}|x{.35\linewidth}|} 
      \hline
      Location $\rightarrow$ & $L$ & $M$ & $~~N~~$ & $T$ \\\hline  
      $x_1$ while entering& $x_{new} + t$ delay in
      Check is $t$ & $0$ &   &  \\ 
      \hline
      $x_2$ while entering&  $x_{new}-x_{old} + t$ & $1-x_{old}$ &
      $0$ & 1 \\\hline 
      Time elapsed at & $1 - x_{new}-t$ & $x_{old}$ & $1$ & -\\\hline
      Cost incurred at &  $-1 + x_{new} + t$ & $-5x_{old}$ & $2$ & - \\\hline
      \newline Total cost upon reaching target &  &
       & & cost of $-t$ due to
      delay in Check and $x_{new} = 5x_{old}+\varepsilon$ gives a total
      cost $-t+1+x_{new}+t -5x_{old}= 1 + \varepsilon = 1$ iff
      $\varepsilon=0$ \\\hline 
    \end{tabular}
  \end{table}

  Let us summarize the construction. Let us assume that when entering
  $\ell_k$ (see Fig.~\ref{fig:equalone}) the value of $x_{1}$
  (resp. $x_2$) is $\frac{1}{5^{c_1}7^{c_2}}$ (resp. $0$).  First, let
  us consider the case where Player 1 simulates correctly the machine
  (i.e., if Player 1 spends $\frac{4}{5^{c_1}7^{c_2}}$ time units in
  $\ell_k$). In this case, Player 2 has three possibilities: $(i)$ either
  Player 2 goes immediately to $Go$, in this case Player 1 has no
  incentive to go to the location Abort and thus the simulation of the
  machines goes on; $(ii)$ or Player 2 goes to $Go$ after
  delaying some time $t>0$ in the location Check, in this case Player
  1 will go to the location Abort in order to immediately achieve his
  objective; $(iii)$ or Player 2 goes to the widget $\text{WD}_1$, in
  this case, as Player 1 has correctly simulated the machine, Player 1
  achieves his objective. Let us now turn to the case where Player 1
  does not simulate the machine faithfully. In this case, Player 2 has
  the possibility to reach the target location with a cost different
  from $1$ by using the widget $\text{WD}_1$. In conclusion, if Player
  1 simulates the machine properly either the simulation continues or
  he achieves his objective immediately; and if Player 1 does not
  simulate the machine properly, then Player 2 has a strategy to reach
  a target location with a cost different from $1$.

  \smallskip
  \paragraph{Simulate Increment instruction}
  Fig.~\ref{fig_undec_ek_inc} gives the complete module for increment
  instruction. The construction is similar to that as decrement.
\begin{figure}[t]
  \begin{center}
\begin{tikzpicture}[->,>=stealth',shorten >=1pt,auto,node distance=1cm,
       semithick,scale=0.8]
       \node[initial,initial text={}, player1] at (-1.3,0) (lk) {$0$} ;
       \node()[above of=lk,node distance=5mm,color=gray]{$\ell_k$};
       \node[player1] at (0.5,0) (i) {$0$} ;
       \node()[above of=i,node distance=5mm,color=gray]{$I$};

       \node[player2] at (3,0) (chk){$-5$} ;
       \node()[above of=chk,node distance=5mm,color=gray]{$\text{Check}$};

       \node[player1] at (6,0) (lk1){$0$ } ;  
       \node()[above of=lk1,node distance=5mm,color=gray]{Go};

       \node[fill=gray!20,rounded corners,fill opacity=0.9] at (12,-0.65)(nxt){$\ell_{k+1}$}; 

       \node[player1] at (9, 0) (B) {$1$};
       \node()[above of=B,node distance=5mm,color=gray]{$\text{Abort}$};

       \node[accepting, player1] at (12, 0) (T1) {};
       \node()[above of=T1,node distance=5mm,color=gray]{$T_1$};
       
       \path (lk) edge node {$x_1{=}1$} node[below]{$\set{x_1}$}(i);
       \path (i) edge node {$0{<}x_1{<}1$} (chk);
       \path (chk) edge node {$\set{x_2}$} (lk1);
       \path (lk1) edge node {$x_2{=}0$} (B);
       \path(lk1) edge[bend right=7] node[below, near start] {$x_2{=}0$}(nxt);
       \path (B) edge node {$x_2 {>} 1$} (T1);
       \node[player1] (A) at (3, -1.5) {$-5$};
       \node()[below of=A,node distance=5mm,color=gray]{$L$};
       
       \node[player1] at (6, -1.5) (B){$-1$};
       \node()[below of=B,node distance=5mm,color=gray]{$M$};
       
       \node[player1] at (9, -1.5) (C){$6$};
       \node()[below of=C,node distance=5mm,color=gray]{$N$};

       \node[accepting, player1] at (12, -1.5) (T) {};
       \node()[below of=T,node distance=5mm,color=gray]{$T$};
       
       \path (A) edge node {$x_1{=}1$} node[below] {$\set{x_1}$} (B);
       \path (B) edge node {$x_2 {=} 2$} node[below] {$\set{x_2}$} (C);
       \path (C) edge node {$x_2 {=} 1$} node[below]{$\set{x_2}$} (T);
     
       \path (chk) edge node[yshift=1mm] {$x_1{\leq} 1$} (A);

       \node[rotate=90,color=gray] at (1.8, -1.5) (N) {$\text{WI}_1$};

       \draw[dashed,draw=gray,rounded corners=10pt] (2.3,-1) rectangle (12.6,-2.4);
     \end{tikzpicture}

\caption{$\ReachOb({=}1)$ : Widget to simulate increment counter $C_1$ instructions}
\label{fig_undec_ek_inc}
\end{center}
\end{figure}
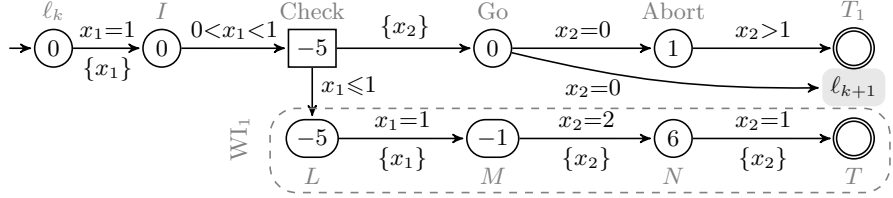

  \smallskip
  \paragraph{Simulate Zero-check instruction}
  The module for the Zero check instruction ($\ell_k:$ if $C_2=0$ goto
  $\ell_{k+1}$ else goto $\ell_{k+2}$) is depicted in
  Fig.~\ref{fig_undec_ek_zeroCheck} and its widgets are in
  Fig.~\ref{fig_undec_ek_wze}~and~\ref{fig_undec_ek_wzne}.
\begin{figure}[tbp]
\centering
    \begin{tikzpicture}[->,>=stealth',shorten >=1pt,auto,node distance=1cm,
      semithick,scale=0.8]
      \node[initial,initial text={}, player1] at (-.5,0) (A) {$0$};
      \node()[above of=A,node distance=5mm,color=gray]{$\ell_k$};

      \node[player2] at (2,1) (B){$0$} ;
       \node()[below of=B,node distance=5mm,color=gray]{$\text{Check}_{c_2=0}$};
      \node[widget,color=gray] at (2,2.5) (B1){$\text{W}_2^{= 0}$ };

      \node[player2] at (2,-1) (C){$0$} ;
      \node()[above of=C,node distance=5mm,color=gray]{$\text{Check}_{c_2\neq 0}$};
      \node[widget,color=gray] at (2,-2.5) (C1){$\text{W}_2^{\neq 0}$};

      \node[fill=gray!20,rounded corners,fill opacity=0.9] at (4.5,1)(D){$\ell_{k+1}$};

         \node[fill=gray!20,rounded corners,fill opacity=0.9] at (4.5,-1)(E){$\ell_{k+2}$};

      \path (A) edge[bend left=10] node[above left,xshift=3mm] {$x_2=0$} (B);
      \path (A) edge[bend right=10] node[below left,xshift=3mm]  {$x_2=0$} (C);
      \path (B) edge node {$x_2=0$} (B1);
      \path (C) edge node {$x_2=0$} (C1);
      \path (B) edge  node {$x_2=0$} (D);
      \path (C) edge  node {$x_2=0$} (E);
    \end{tikzpicture}
    \caption{Widget $\text{WZ}_2$ simulating zero-check for $C_2$}
    \label{fig_undec_ek_zeroCheck}
\end{figure}
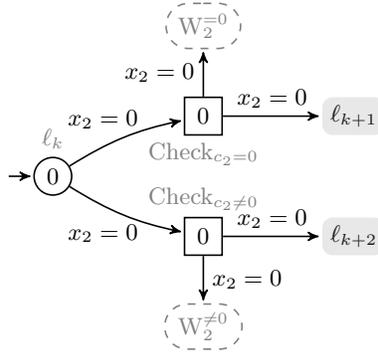

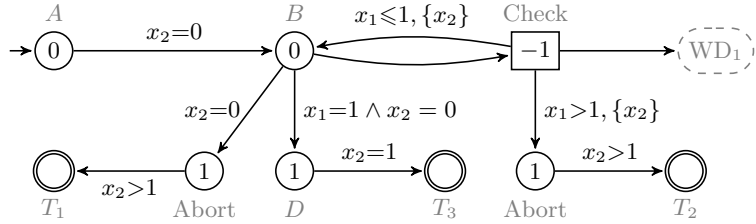
\begin{figure}[tbp]
\centering
\begin{tikzpicture}[->,>=stealth',shorten >=1pt,auto,node distance=1cm,
  semithick,scale=0.8]
  \node[initial,initial text={}, player1] (a) {$0$};
  \node()[above of=a,node distance=5mm,color=gray]{$A$};
  \node[player1] at (4,0) (b) {$0$};
  \node()[above of=b,node distance=5mm,color=gray]{$B$};
  \node[player2] at (8,0) (chk) {$-1$};
  \node()[above of=chk,node distance=5mm,color=gray]{Check};

  \node[widget,color=gray] at (11,0) (wd){$\text{WD}_1$};
  
  \node[player1] at (8,-2) (c) {$1$};
  \node()[below of=c,node distance=5mm,color=gray]{Abort};
  \node[player1] at (4,-2) (d) {$1$};
  \node()[below of=d,node distance=5mm,color=gray]{$D$};
  \node[accepting,player1] at (10.5,-2) (t) {};
  \node()[below of=t,node distance=5mm,color=gray]{$T_2$};
  \node[accepting,player1] at (6.5,-2) (t1) {};
  \node()[below of=t1,node distance=5mm,color=gray]{$T_3$};
  
  \path (a) edge node {$x_2{=}0$} (b);
  \path (b) edge[bend right=10] node {} (chk);
  \path (chk) edge [bend right=10] node[above] {$x_1 {\leq} 1,\{x_2\}$} (b);  
  \path (chk) edge node {} (wd);
  \path (b) edge node {$x_1{=}1\land x_2=0$} (d);
  \path (chk) edge node {$x_1{>}1,\{x_2\}$} (c);
  \path (c) edge node {$x_2{>}1$} (t);
  \path (d) edge node {$x_2{=}1$} (t1);
  
  \node[player1] at (2.5,-2) (C) {$1$};
  \node()[below of=C,node distance=5mm,color=gray]{Abort};
  \node[accepting, player1] at (0,-2) (T) {};
  \node()[below of=T,node distance=5mm,color=gray]{$T_1$};
  
  \path (b) edge[left] node {$x_2{=}0$} (C);
  \path (C) edge node {$x_2{>}1$} (T);
\end{tikzpicture}
\caption{Widget $\text{W}_2^{=0}$. Note that if the delay in Check is
  $>0$ then Player 1 will go to Abort and reach a target with cost
  $=1$. Widget $\text{WD}_1$ is the same as in the decrement
  module. Delay at $B$ is $k = \frac{4\times 5^i}{5^{c_1}7^{c_2}} +
  \varepsilon$. If $\varepsilon \neq 0$ then Player 2 will enter the
  widget $\text{WD}_1$ and the target in $\text{WD}_1$ is reached with
  cost $\neq 1$.}
\label{fig_undec_ek_wze}
\end{figure}

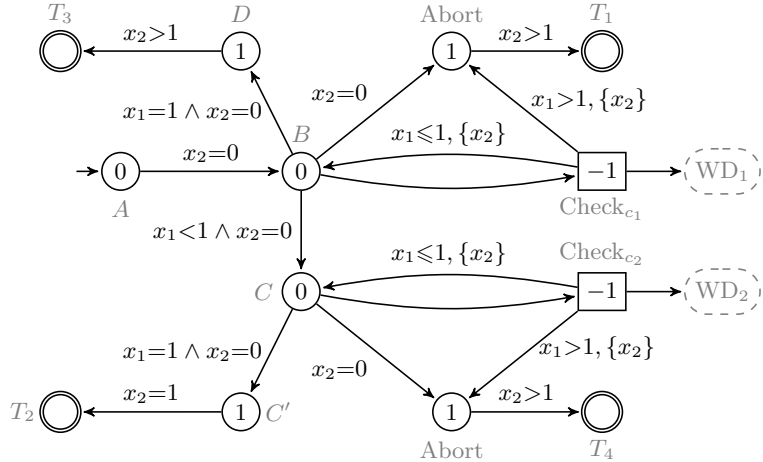
\begin{figure}[tbp]
\centering
\begin{tikzpicture}[->,>=stealth',shorten >=1pt,auto,node distance=1cm,
  semithick,scale=0.8]
  \node[initial,initial text={}, player1] (a) {$0$};
  \node()[below of=a,node distance=5mm,color=gray]{$A$};
  \node[player1] at (3,0) (b) {$0$};
  \node()[above of=b,node distance=5mm,color=gray]{$B$};

  \node[player2] at (8,0) (chk) {$-1$};
  \node()[below of=chk,node distance=5mm,color=gray]{$\text{Check}_{c_1}$};
  \node[widget,color=gray] at (10,0) (wl){$\text{WD}_1$};
  
  \node[player2] at (8,-2) (chk2) {$-1$};
  \node()[above of=chk2,node distance=5mm,color=gray]{$\text{Check}_{c_2}$};
  \node[widget,color=gray] at (10,-2) (wl2){$\text{WD}_2$};
  
  \node[player1] at (3,-2) (c) {$0$};
  \node()[left of=c,node distance=5mm,color=gray]{$C$};
  \node[player1] at (2,2) (d) {$1$};
  \node()[above of=d,node distance=5mm,color=gray]{$D$};
  \node[player1] at (2,-4) (c1) {$1$};
  \node()[right of=c1,node distance=5mm,color=gray]{$C'$};
    
  \node[accepting, player1] at (-1,-4) (t) {};
  \node()[left of=t,node distance=5mm,color=gray]{$T_2$};
  \node[accepting, player1] at (-1,2) (t1) {};
  \node()[above of=t1,node distance=5mm,color=gray]{$T_3$};
  
  \path (a) edge node {$x_2{=}0$} (b);
  \path (b) edge[bend right=10] node {} (chk);
  \path (chk) edge [bend right=10] node[above] {$x_1{\leq}1,\{x_2\}$} (b);  
  \path (chk) edge node {} (wl);
  
  \path (c) edge[bend right=10] node {} (chk2);
  \path (chk2) edge [bend right=10] node[above] {$x_1{\leq}1,\{x_2\}$} (c);  
  \path (chk2) edge node {} (wl2);
  
  \path (b) edge [left] node {$x_1{=}1\land x_2{=}0$} (d);
  \path (b) edge node[left] {$x_1{<}1\land x_2{=}0$} (c);
  \path (c) edge node[left] {$x_1{=}1\land x_2{=}0$} (c1);
  
  \path (c1) edge node[above] {$x_2{=}1$} (t);
  \path (d) edge node[above] {$x_2{>}1$} (t1);
  
  \node[player1] at (5.5,2) (C) {$1$};
  \node()[above of=C,node distance=5mm,color=gray]{Abort};
  \node[accepting, player1] at (8,2) (T) {};
  \node()[above of=T,node distance=5mm,color=gray]{$T_1$};

  \path (b) edge node[above left] {$x_2{=}0$} (C);
  \path (C) edge node {$x_2{>}1$} (T);
  \path (chk) edge node[right,yshift=1mm] {$x_1{>}1,\{x_2\}$} (C);
  
  \node[player1] at (5.5,-4) (C1) {$1$}; 
  \node()[below of=C1,node distance=5mm,color=gray]{Abort};

  \node[accepting, player1] at (8,-4) (T1){};
  \node()[below of=T1,node distance=5mm,color=gray]{$T_4$};

  \path (c) edge node[below left] {$x_2{=}0$} (C1);
  \path (C1) edge node {$x_2{>}1$} (T1);
  \path (chk2) edge node[right,yshift=1mm,xshift=1mm] {$x_1{>}1,\{x_2\}$} (C1);
\end{tikzpicture}
\caption{Widget $\text{W}_2^{\neq 0}$. Widget $\text{WD}_2$ is similar
  to $\text{WD}_1$ in the decrement module shown earlier, except that
  the prices are adjusted to verify decrement of counter $C_2$ instead
  of $C_1$.}
\label{fig_undec_ek_wzne}
\end{figure}
  
  Player 1 guesses whether $C_2$ is zero or not and Player 2 has the
  possibility to check if the guess of Player 1 is correct by entering
  the corresponding widget. In the widget depicted in Fig.
  \ref{fig_undec_ek_wze}, in order to verify if $C_2$ is indeed zero,
  $C_1$ is repeatedly decremented (using a construction similar to
  Decrement module in Fig.~\ref{fig:equalone}) until $C_1$ becomes 0.
    At that point, if $x_1$ is not 1 then clearly $C_2$ was nonzero, and
  Player 1 has made an error in guessing. In this case, the total cost
  incurred will be greater than $1$. However, if indeed $C_2$ was
  zero, then the total cost incurred is 1 for the following reasons:
  \begin{itemize}
  \item If some time has elapsed at location Check, and still $x_1$ is
    less than 1, Player~1 can go to Abort from location $B$.  However,
    if time has elapsed at Check and $x_1>1$, the only option is to go
    to Abort from Check; then Abort goes to $T_2$. If Player~1 cannot
    control $x_1$ not exceeding 1, then in that case, Player~2 will
    spend 0 units of time in Check, and then Player~1 will pay a cost
    $>1$ on reaching $T_2$.
  \end{itemize}
  In a similar manner, if Player~1 guesses that $C_2 \neq 0$, then
  Player~2 verifies it by entering the widget in
  Fig.~\ref{fig_undec_ek_wzne}.  In this case, $C_1$ is first
  decremented (the decrement module may be called several times) until
  it reaches 0, and then $C_2$ is decremented until it also reaches 0
  (the decrement module is called at least once):
  \begin{itemize}  
  \item Entry to $C$ happens with $x_1<1$ and $x_2=0$. To get to $C'$,
    some time elapse is needed at $C$. However, then $x_2 >0$, so
    directly there is no access to $C'$, without going to
    $\text{Check}_{c_2}$ (ensuring that the decrement module for $C_2$
    is taken at least once). To not get punished through
    $\text{Check}_{c_2}$, Player~1 has to elapse a time of the form
    $\frac 6 {5^i7^j}$: hence, if $C_2=0$, this time elapse must be
    $\frac 6{5^i}$, if $C$ was entered with $x_1=\frac{1}{5^i}$. Due
    to $C_2=0$, at some point, $\frac 6{5^i}$ will exceed 1; at that
    time, Player~2 will spend 0 unit of time in Check, and go to
    Abort. That will punish Player~1.
  \end{itemize}  
  Note that Player~2 can not delay in the zero check module due to
  clock constraints. In the widgets $\text{W}_1^{=0}$ and
  $\text{W}_1^{\neq 0}$ however, Player~2 can delay at the Check
  locations. As in the decrement module, we offer to Player~1 the
  option to abort (via the Abort location) : if $t$ units of time was
  spent at location Check, then the cost incurred will be $1+p-t$
  where $p>0$ is chosen by Player~1. Table~\ref{tab2} shows the cost
  incurred in the widget $\text{W}_2^{\neq 0}$ when Player~1 guesses
  $C_2\neq 0$.

  \begin{table}[tbp]
    \caption{Cost incurred in widget $\text{W}_2^{\neq 0}$}
    \label{tab2}
    \centering
    \begin{tabular}{|x{.15\linewidth}|x{.08\linewidth}|x{.13\linewidth}|
        x{.165\linewidth}|x{.05\linewidth}|x{.13\linewidth}|
        x{.15\linewidth}|x{.05\linewidth}|} 
      \hline
      Location $\rightarrow$ & $A$ & $B$ & $D$ & $T_3$ & $C$ & $C'$ & $T_2$\\\hline  
      $x_1$ while entering & $\frac{1}{5^{c_1}7^{c_2}}$ &
      $\frac{5^i}{5^{c_1}7^{c_2}}$ \newline loop taken $i$ times &
      $1$ \newline  $\frac{5^i}{5^{c_1}7^{c_2}} = 1$ \newline
      $\Rightarrow i=c_1 \land{} $\newline $c_2=0$ &  &
      $\frac{7^j}{7^{c_2}}$  \newline  loop taken $j$ times  &
      $1$ \newline $\frac{5^i7^j}{5^{c_1}7^{c_2}} = 1$ \newline
      $\Rightarrow i=c_1 \land{}$\newline  $j=c_2>0$ &  
      \\\hline
  
       $x_2$ while entering & 0 & 0 & 0 & $>1$ & 0 & 0 & $1$ \\\hline
  
       Time elapsed at& 0 & $\frac{4\times
         5^i}{5^{c_1}7^{c_2}}+\varepsilon$ & $>1$ & 0 &
       $\frac{6\times 7^j}{7^{c_2}}+\varepsilon$  & $1$ &0 \\\hline 
  
       Cost incurred at& 0 & 0 & $>1$ & 0 & 0 & $1$ &  \\\hline
  
       Total cost &  &  &  & $>1$ &  & & $1$ \\\hline
     \end{tabular}
   \end{table}

  To summarize the zero check instruction for $C_2$: assume that at
  location $\ell_k$, we enter with $x_1=\frac{1}{5^i}$ (hence $C_2=0$)
  and $x_2=0$.  First let us consider the case when Player~1 guesses
  correctly that $C_2$ is zero: in this case, Player~2 has 2
  possibilities: ($i$) continue with the next instruction, ($ii$)
  simulate the widget $\text{W}_2^{=0}$.  In both cases, no time
  elapse is possible at location $\text{Check}_{c_2}$.  Player~1
  achieves his objective even on entering widget $\text{W}_2^{=0}$ due
  to his correct guess. Let us now turn to the case when Player~1 made
  a wrong guess. Then, Player~2 has the capability to invoke widget
  $\text{W}_2^{\neq 0}$ and in this case, Player~1 will incur a cost
  different from $1$. A similar analysis applies to the case when we
  start in $\ell_k$ with $x_1=\frac{1}{5^i7^j}$.
   
  \smallskip
  \paragraph{Correctness of the construction}
  On entry into the location $\location_n$ (encoding the HALT
  instruction of the two-counter machine), we reset clock $x_1$ to 0;
  $\location_n$ has cost 1, and the edge coming out of $\location_n$
  goes to a Goal location, with guard $x_1{=}1$.
  \begin{enumerate}
  \item Assume that the two counter machine halts. If Player~1
    simulates all the instructions correctly, he will incur a cost 1,
    by either reaching the goal location after $\location_n$, or by
    entering a widget (the second case only occurs if Player~2 decides
    to check whether Player~1 simulates the machine faithfully. If
    Player~1 makes an error in his computation, Player~2 can always
    enter an appropriate widget, making the cost different from $1$.
    In summary, if the two counter machine halts, Player~1 has a
    strategy to achieve his goal (i.e., reaching a target location
    with a cost equal to $1$).
  \item Assume that the two counter machine does not halt.
    \begin{itemize}
    \item If Player~1 simulates all the instructions correctly, and if
      Player~2 never enters a widget, then Player~1 incurs cost
      $\infty$ as the path never reaches a target. 
    \item Suppose now that Player~1 makes an error. In this case,
      Player~2 always has the capability to reach a target set with a
      cost different from $1$.
     \end{itemize}
     In summary, if the two counter machine does not halt, Player~1
     does not have a strategy to achieve his goal.
   \end{enumerate}

   Thus, Player~1 incurs a cost 1 iff he chooses the strategy of
   faithfully simulating the two counter machine, when the machine
   halts. When the machine does not halt, the cost incurred by
   Player~1 is different from $1$ if Player 1 made a simulation error
   and Player~2 entered a widget. Else if a widget is not entered then
   the run does not end and cost is $+\infty$.  \qed
\end{proof}

We briefly discuss the difference with the proofs for other
constrained-price reachability objectives.
\begin{itemize}
\item $\ReachOb({<}1)$. We use 2 clocks $x_1, x_2$, and a module for
  each instruction of the two counter machine. On entry into a module
  for increment/decrement/zero check, $x_1=\frac{1}{5^{c_1}7^{c_2}}$
  and $x_2=0$.
 
  We discuss the case of decrementing counter $C_1$.
  Fig.~\ref{fig_undec_lk_dec} gives the complete module for the
  instruction to decrement $C_1$. To simulate the decrementation
  correctly, Player 1 has to elapse $\frac{4}{5^{c_1}7^{c_2}}$ units
  of time at $\ell_k$. At the location Check, Player 2 can either
  continue to $\ell_{k+1}$ after elapsing $t \geq 0$ units of time, or
  enter a widget to check the choice of Player 1. Consider the case
  when Player 2 proceeds to $\ell_{k+1}$ after elapsing a time $t
  >0$. This incurs a cost $-t$.  Then Player 1 can go to location
  Abort, and spend 1 unit of time there, reaching a target with cost
  $1-t < 1$.  Assume now that Player 1 spends
  $\frac{4}{5^{c_1}7^{c_2}}+\varepsilon$ units of time in $\ell_k$, and
  Player 1 spends $t>0$ units of time at Check. Since Player 1 has
  elapsed more time than what he should have, Player 2 enters the
  widget $\text{WD}_1^{>}$.  The cost incurred so far is $-t$.  On
  entry into $\text{WD}_1^{>}$, we have
  $x_1=\frac{5}{5^{c_1}7^{c_2}}+\varepsilon+t$,
  $x_2=\frac{4}{5^{c_1}7^{c_2}}+\varepsilon+t$.  Then, the total
  accumulated cost becomes $-1+\frac{5}{5^{c_1}7^{c_2}}+\varepsilon$
  on coming out of $L$, and then becomes $-1+\varepsilon$ on coming
  out of $M$; this further becomes $1+\varepsilon$ on entering $O$. At
  location $O$, time $0<p<1$ has to be spent, making the total cost to
  be $1-p+\varepsilon$. Time $p$ is chosen by Player 2. Thus, if
  $\varepsilon=0$ (hence, Player 1 made no error), the cost incurred
  is less than $1$; however, when $\varepsilon >0$, $p$ can always be
  chosen to be at most $\varepsilon$, thereby making the total cost
  at least $1$.
 
  A similar analysis can be done when Player 1 incurs a delay
  $\frac{4}{5^{c_1}7^{c_2}}-\varepsilon$, $\varepsilon>0$ at location
  $\ell_{k}$. In this case, Player 2 enters widget $\text{WD}_1^{<}$.
 
  The increment and zero-check instructions are obtained by a similar
  approach.
   
 \begin{figure}[tbp]
\begin{center}
\begin{tikzpicture}[->,>=stealth',shorten >=1pt,auto,node distance=1cm,
  semithick,scale=0.9]
	\node[initial,initial text={}, player1] at (0.5,0) (lk) {$0$} ;
       \node()[above of=lk,node distance=5mm,color=gray]{$\ell_k$};
       
       \node[player2] at (3,0) (chk){$-1$} ;
       \node()[above right of=chk,node distance=7mm,color=gray]{$\text{Check}$};

       \node[player1] at (6,0) (lk1){$0$ } ;  
       \node()[above of=lk1,node distance=5mm,color=gray]{Go};

       \node[fill=gray!20,rounded corners,fill opacity=0.9] at
       (12,-0.65)(nxt){$\ell_{k+1}$}; 

       \node[player1] at (9, 0) (B) {$1$};
       \node()[above of=B,node distance=5mm,color=gray]{$\text{Abort}$};

       \node[accepting, player1] at (12, 0) (T) {};
       \node()[above of=T,node distance=5mm,color=gray]{$T$};
       
       \path (lk) edge node {$0{<}x_1{<}1$} (chk);
       \path (chk) edge node {$\set{x_2}$} (lk1);
       \path (lk1) edge node {$x_2{=}0$} (B);
       \path(lk1) edge[bend right=7] node[below, near start] {$x_2{=}0$}(nxt);
       \path (B) edge node {$x_2 {=} 1$} (T);
        
       \node[player1] (A) at (2, -1.5) {$1$};
       \node()[below of=A,node distance=5mm,color=gray]{$A$};
       
       \node[player1] at (4.5, -1.5) (B){$5$};
       \node()[below of=B,node distance=5mm,color=gray]{$B$};
       
       \node[player1] at (7, -1.5) (C){$-1$};
       \node()[below of=C,node distance=5mm,color=gray]{$C$};

       \node[accepting, player1] at (9.5, -1.5) (T1) {};
       \node()[below of=T1,node distance=5mm,color=gray]{$T_1$};
       
       \path (A) edge node {$x_1{=}1$} node[below] {$\set{x_1}$} (B);
       \path (B) edge node {$x_2 {=} 1$} node[below] {$\set{x_2}$} (C);
       \path (C) edge node {$0 {<} x_2 {<} 1$} (T1);
     
       \path (chk) edge node {} (A);

       \node[rotate=90,color=gray] at (0.8, -1.5) (N) {$\text{WD}_1^<$};

       \draw[dashed,draw=gray,rounded corners=10pt] (1.3,-1) rectangle (10.1,-2.4);
       
       \node[player1] (L) at (2, 1.5) {$-1$};
       \node()[above of=L,node distance=5mm,color=gray]{$L$};
       
       \node[player1] at (4.5, 1.5) (M){$-5$};
       \node()[above of=M,node distance=5mm,color=gray]{$M$};
       
       \node[player1] at (7, 1.5) (N){$2$};
       \node()[above of=N,node distance=5mm,color=gray]{$N$};
       
       \node[player1] at (9.5, 1.5) (O){$-1$};
       \node()[above of=O,node distance=5mm,color=gray]{$N$};

       \node[accepting, player1] at (12, 1.5) (T2) {};
       \node()[above of=T2,node distance=5mm,color=gray]{$T_2$};
       
       \path (L) edge node {$x_1{=}1$} node[below] {$\set{x_1}$} (M);
       \path (M) edge node {$x_2 {=} 1$} node[below] {$\set{x_2}$} (N);
       \path (N) edge node {$x_2 {=} 1$} node[below] {$\set{x_2}$} (O);
       \path (O) edge node {$0 {<} x_2 {<} 1$} (T2);
     
       \path (chk) edge node {} (L);

       \node[rotate=90,color=gray] at (0.8, 1.5) (N) {$\text{WD}_1^>$};

      \draw[dashed,draw=gray,rounded corners=10pt] (1.3,1) rectangle (12.6,2.4);
      \end{tikzpicture}

      \caption{$\ReachOb({<}1)$: Simulation to decrement counter
        $C_1$}
\label{fig_undec_lk_dec}
\end{center}
\end{figure}

\item $\ReachOb({\leq }1)$.  We use 2 clocks $x_1, x_2$, and a module
  for each instruction of the two counter machine. On entry into a
  module for increment/decrement/zero check,
  $x_1=\frac{1}{5^{c_1}7^{c_2}}$ and $x_2=0$ where $c_1$($c_2$) is the
  value of counter $C_1$($C_2$).
 
  We discuss the case of decrementing counter
  $C_1$. Fig.~\ref{fig_undec_leqk_dec} gives the complete module for
  the instruction to decrement $C_1$.  As in the case of the objective
  $\ReachOb({<}1)$, Player 1 has to elapse a time
  $\frac{4}{5^{c_1}7^{c_2}}$ at $\ell_k$.  If Player 2 elapses $t>0$
  units of time in Check and proceeds to $\ell_{k+1}$, then Player 1 has
  the option to go to Abort, and reach a target with a total cost of
  $-t+1+k$, where $1+k$ is the time elapsed at Abort. $k$ can be
  chosen by Player 1 so that $1+k-t \leq 1$.

  Let us consider the case when Player 1 spends
  $\frac{4}{5^{c_1}7^{c_2}}-\varepsilon$ units of time in $\ell_k$, with
  $\varepsilon \geq 0$ and Player 1 spends $t >0$ units of time at Check.
  Player 2 then will proceed to the widget $\text{WD}_1^{<}$. On entry into
  $\text{WD}_1^{<}$, we have $x_1=\frac{5}{5^{c_1}7^{c_2}}-\varepsilon+t$,
  $x_2=\frac{4}{5^{c_1}7^{c_2}}-\varepsilon+t$, and an incurred cost of
  $-t$.  On entering $M$, we have an accumulated cost of
  $1-\frac{5}{5^{c_1}7^{c_2}}+\varepsilon-2t$, and further on entering
  $N$, the accumulated cost becomes $1+\varepsilon-2t$.  Finally, when
  $T$ is entered, the total cost is $1+\varepsilon-2t-p$, where $0<p<1$
  is a delay chosen by Player 2.  Clearly, $1+\varepsilon-2t-p < 1$ if
  $\varepsilon=0$. However, if $\varepsilon \neq 0$, Player 2 can adjust the
  values of $p,t$ in such a way that $2t+p < \varepsilon$, there by
  making the total cost $>1$.

  A similar analysis can be done when Player 1 incurs a delay
  $\frac{4}{5^{c_1}7^{c_2}}+\varepsilon$, $\varepsilon >0$ at location
  $\ell_{k}$. In this case, Player 2 enters widget $\text{WD}_1^{>}$.
 
  The increment and zero-check instructions are obtained by a similar
  approach.

 \begin{figure}[tbp]
\begin{center}
\begin{tikzpicture}[->,>=stealth',shorten >=1pt,auto,node distance=1cm,
  semithick,scale=0.9]
	\node[initial,initial text={}, player1] at (0,0) (lk) {$0$} ;
       \node()[above of=lk,node distance=5mm,color=gray]{$\ell_k$};
       
       \node[player2] at (3,0) (chk){$-1$} ;
       \node()[above of=chk,node distance=5mm,color=gray]{$\text{Check}$};

       \node[player1] at (6,0) (lk1){$0$ } ;  
       \node()[above of=lk1,node distance=5mm,color=gray]{Go};

       \node[fill=gray!20,rounded corners,fill opacity=0.9] at
       (11,-0.65)(nxt){$\ell_{k+1}$}; 

       \node[player1] at (9, 0) (B) {$1$};
       \node()[above of=B,node distance=5mm,color=gray]{$\text{Abort}$};

       \node[accepting, player1] at (11, 0) (T) {};
       \node()[above of=T,node distance=5mm,color=gray]{$T$};
       
       \path (lk) edge node {$0{<}x_1{<}1$} (chk);
       \path (chk) edge node {$\set{x_2}$} (lk1);
       \path (lk1) edge node {$x_2{=}0$} (B);
       \path(lk1) edge[bend right=7] node[below, near start] {$x_2{=}0$}(nxt);
       \path (B) edge node {$x_2 {>} 1$} (T);
        
       \node[player1] (A) at (0, -1.5) {$1$};
       \node()[below of=A,node distance=5mm,color=gray]{$A$};
       
       \node[player1] at (3, -1.5) (B){$5$};
       \node()[below of=B,node distance=5mm,color=gray]{$B$};
       
       \node[player1] at (6, -1.5) (C){$-1$};
       \node()[below of=C,node distance=5mm,color=gray]{$C$};

       \node[accepting, player1] at (9, -1.5) (T1) {};
       \node()[below of=T1,node distance=5mm,color=gray]{$T_1$};
       
       \path (A) edge node {$x_1{=}1$} node[below] {$\set{x_1}$} (B);
       \path (B) edge node {$x_2 {=} 1$} node[below] {$\set{x_2}$} (C);
       \path (C) edge node {$0 {<} x_2 {<} 1$} (T1);
     
       \path (chk) edge node {} (A);

       \node[rotate=90,color=gray] at (-1.2, -1.5) (N) {$\text{WD}_1^<$};

       \draw[dashed,draw=gray,rounded corners=10pt] (-0.7,-1) rectangle (9.6,-2.4);
       
       \node[player1] (L) at (0, 1.5) {$-1$};
       \node()[above of=L,node distance=5mm,color=gray]{$L$};
       
       \node[player1] at (3, 1.5) (M){$-5$};
       \node()[above of=M,node distance=5mm,color=gray]{$M$};
       
       \node[player1] at (6, 1.5) (N){$2$};
       \node()[above of=N,node distance=5mm,color=gray]{$N$};
       
       \node[player1] at (9, 1.5) (O){$-1$};
       \node()[above of=O,node distance=5mm,color=gray]{$N$};

       \node[accepting, player1] at (11.5, 1.5) (T2) {};
       \node()[above of=T2,node distance=5mm,color=gray]{$T_2$};
       
       \path (L) edge node {$x_1{=}1$} node[below] {$\set{x_1}$} (M);
       \path (M) edge node {$x_2 {=} 1$} node[below] {$\set{x_2}$} (N);
       \path (N) edge node {$x_2 {=} 1$} node[below] {$\set{x_2}$} (O);
       \path (O) edge node {$0 {<} x_2 {<} 1$} (T2);
     
       \path (chk) edge node {} (L);

       \node[rotate=90,color=gray] at (-1.2, 1.5) (N) {$\text{WD}_1^>$};

      \draw[dashed,draw=gray,rounded corners=10pt] (-0.7,1) rectangle (12.1,2.4);
      \end{tikzpicture}
 
\caption{$\ReachOb({\leq}1)$: Simulation to decrement counter $C_1$ }
\label{fig_undec_leqk_dec}
\end{center}
\end{figure}

\item $\ReachOb({> }1)$. We use 2 clocks $x_1, x_2$, and a module for
  each instruction of the two counter machine. On entry into a module
  for increment/decrement/zero check, $x_1=\frac{1}{5^{c_1}7^{c_2}}$
  and $x_2=0$ where $c_1$($c_2$) is the value of counter $C_1$($C_2$).
 
  We discuss the case of decrementing counter
  $C_1$. Fig.~\ref{fig_undec_gk_dec} gives the complete module for the
  instruction to decrement $C_1$.  As in the previous two objectives,
  Player 1 has to spend a time $\frac{4}{5^{c_1}7^{c_2}}$ at $\ell_k$.
  If Player 2 elapses $t>0$ units of time in Check and proceeds to
  $\ell_{k+1}$, then Player 1 has the option to go to Abort, and reach a
  target with a total cost of $t+k$, where $k<1$ is the time elapsed
  at Abort by Player 1. $k$ can be chosen by Player 1 so that $t+k>1$.

  Let us consider the case when Player 1 spends
  $\frac{4}{5^{c_1}7^{c_2}}-\varepsilon$ units of time in $\ell_k$, with
  $\varepsilon \geq 0$ and Player 1 spends $t >0$ units of time at Check.
  Player 2 then will proceed to the widget $\text{WD}_1^{<}$. On entry into
  $\text{WD}_1^{<}$, we have $x_1=\frac{5}{5^{c_1}7^{c_2}}-\varepsilon+t$,
  $x_2=\frac{4}{5^{c_1}7^{c_2}}-\varepsilon+t$, and an incurred cost of
  $t$.  On entering $M$, we have an accumulated cost of
  $-1+\frac{5}{5^{c_1}7^{c_2}}-\varepsilon+2t$, and further on entering
  $N$, the accumulated cost becomes $-1-\varepsilon+2t$.  On entering
  $O$, the accumulated cost becomes $1-\varepsilon+2t$, and finally, on
  entering $T$, the total cost is $1-\varepsilon+2t+p$, where $0<p<1$ is
  a delay chosen by Player 2.  Clearly, $1-\varepsilon+2t+p > 1$ if
  $\varepsilon=0$. However, if $\varepsilon \neq 0$, Player 2 can adjust the
  values of $p,t$ in such a way that $2t+p \leq \varepsilon$, there by
  making the total cost $\leq 1$.

 \begin{figure}[tbp]
\begin{center}
\begin{tikzpicture}[->,>=stealth',shorten >=1pt,auto,node distance=1cm,
  semithick,scale=0.9]
	\node[initial,initial text={}, player1] at (0,0) (lk) {$0$} ;
       \node()[above of=lk,node distance=5mm,color=gray]{$\ell_k$};
       
       \node[player2] at (3,0) (chk){$1$} ;
       \node()[above of=chk,node distance=5mm,color=gray]{$\text{Check}$};

       \node[player1] at (6,0) (lk1){$0$ } ;  
       \node()[above of=lk1,node distance=5mm,color=gray]{Go};

       \node[fill=gray!20,rounded corners,fill opacity=0.9] at
       (11,-0.65)(nxt){$\ell_{k+1}$}; 

       \node[player1] at (9, 0) (B) {$1$};
       \node()[above of=B,node distance=5mm,color=gray]{$\text{Abort}$};

       \node[accepting, player1] at (11, 0) (T) {};
       \node()[above of=T,node distance=5mm,color=gray]{$T$};
       
       \path (lk) edge node {$0{<}x_1{<}1$} (chk);
       \path (chk) edge node[above] {$x_1 \leq 1$} node[below]{$\set{x_2}$} (lk1);
       \path (lk1) edge node {$x_2{=}0$} (B);
       \path(lk1) edge[bend right=7] node[below, near start] {$x_2{=}0$}(nxt);
       \path (B) edge node {$x_2 {<} 1$} (T);
        
       \node[player1] (A) at (0, -1.5) {$1$};
       \node()[below of=A,node distance=5mm,color=gray]{$A$};
       
       \node[player1] at (3, -1.5) (B){$5$};
       \node()[below of=B,node distance=5mm,color=gray]{$B$};
       
       \node[player1] at (6, -1.5) (C){$-1$};
       \node()[below of=C,node distance=5mm,color=gray]{$C$};

       \node[accepting, player1] at (9, -1.5) (T1) {};
       \node()[below of=T1,node distance=5mm,color=gray]{$T_1$};
       
       \path (A) edge node {$x_1{=}1$} node[below] {$\set{x_1}$} (B);
       \path (B) edge node {$x_2 {=} 1$} node[below] {$\set{x_2}$} (C);
       \path (C) edge node {$0 {<} x_2 {<} 1$} (T1);
     
       \path (chk) edge node {} (A);

       \node[rotate=90,color=gray] at (-1.2, -1.5) (N) {$\text{WD}_1^>$};

       \draw[dashed,draw=gray,rounded corners=10pt] (-0.7,-1) rectangle (9.6,-2.4);
       
       \node[player1] (L) at (0, 1.5) {$-1$};
       \node()[above of=L,node distance=5mm,color=gray]{$L$};
       
       \node[player1] at (3, 1.5) (M){$-5$};
       \node()[above of=M,node distance=5mm,color=gray]{$M$};
       
       \node[player1] at (6, 1.5) (N){$2$};
       \node()[above of=N,node distance=5mm,color=gray]{$N$};
       
       \node[player1] at (9, 1.5) (O){$-1$};
       \node()[above of=O,node distance=5mm,color=gray]{$N$};

       \node[accepting, player1] at (11.5, 1.5) (T2) {};
       \node()[above of=T2,node distance=5mm,color=gray]{$T_2$};
       
       \path (L) edge node {$x_1{=}1$} node[below] {$\set{x_1}$} (M);
       \path (M) edge node {$x_2 {=} 1$} node[below] {$\set{x_2}$} (N);
       \path (N) edge node {$x_2 {=} 1$} node[below] {$\set{x_2}$} (O);
       \path (O) edge node {$0 {<} x_2 {<} 1$} (T2);
     
       \path (chk) edge node {} (L);

       \node[rotate=90,color=gray] at (-1.2, 1.5) (N) {$\text{WD}_1^<$};

      \draw[dashed,draw=gray,rounded corners=10pt] (-0.7,1) rectangle (12.1,2.4);
      \end{tikzpicture}
 
\caption{$\ReachOb(> 1)$ : Simulation to Decrement Counter $C_1$ }
\label{fig_undec_gk_dec}
\end{center}
\end{figure}

\item $\ReachOb({\geq }1)$. It can be seen that a proof along similar
  lines as in the case of $\ReachOb({> }1)$ can be given.
\end{itemize}

\subsection{Bounded-time reachability objective} 

\begin{lemma}
  \label{lem_undec_timeBound}
  The existence of a strategy for bounded-time reachability objective
  $\BReachOb(18,40)$ is undecidable for PTGs with price-rates taken
  from $\{0,1\}$ and 5 clocks or more.
\end{lemma}
\begin{proof}
  We prove that the existence of a strategy for bounded-time
  reachability objective ensuring a cost at most $18$ within 40 time
  units of total elapsed time is undecidable. In order to obtain the
  undecidability result, we use a reduction from the halting problem
  for two counter machines.  Our reduction uses PTGs with price-rates
  in $\{0,1\}$, zero prices of labels, and 6 clocks.

  We specify a module for each instruction of the two counter machine.
  On entry into a module for the $(k+1)$th instruction, we have one of
  the clocks $x_1, \, x_2$ having the value as
  $\frac{1}{2^{k+c_1}{3^{k+c_2}}}$, and the other one having value
  0. Values $c_1$ and $c_2$ represent the values of the two counters,
  after simulation of $k$ instructions. A clock $z$ keeps track of the
  total time elapsed during simulation of an instruction: we always
  have $z=1-\frac{1}{2^k}$ at the end of simulating $k$ instructions.
  Thus, a time of $\frac{1}{2}$ is spent simulating the first
  instruction, a time of $\frac{1}{4}$ is spent simulating the second
  instruction and so on, so that the total time spent in simulating
  the main modules corresponding to increment/decrement/zero check is
  less than $1$ at any point of time.  Two clocks $a$ and $b$ are used
  for rough work, and for enforcing urgency in some locations.

  Again, the role of Player 1 is to simulate the machine faithfully by
  choosing appropriate delays to adjust clock values in order to
  reflect the changes in counter values, and also to reflect the total
  time elapsed.  Player 2 will have the opportunity to check if Player
  cheated while simulating the machine.  We now present the modules
  for decrement, increment and zero-check instructions.

  \medskip

  \paragraph{Simulate increment instruction}
  Fig. \ref{fig_undec_tb_inc_module} gives the complete increment
  module with respect to counter $C_1$.  Assume this is the $(k+1)$th
  instruction that we are simulating. Also, as mentioned earlier, one
  of $x_1,x_2$ will have the value of $\frac{1}{2^{k+c_1}{3^{k+c_2}}}$ on entry, while the
  other will be zero. Without loss of generality, assume $x_1=
  \frac{1}{2^{k+c_1}3^{k+c_2}}$, while $a=b=x_2=0$, and
  $z=1-\frac{1}{2^k}$.  At the end of the module we want $x_2 =
  \frac{1}{2^{k+1+c_1+1}3^{k+1+c_2}}$ and $x_1 = 0$ and $z = 1 -
  \frac{1}{2^{k+1}}$.
 
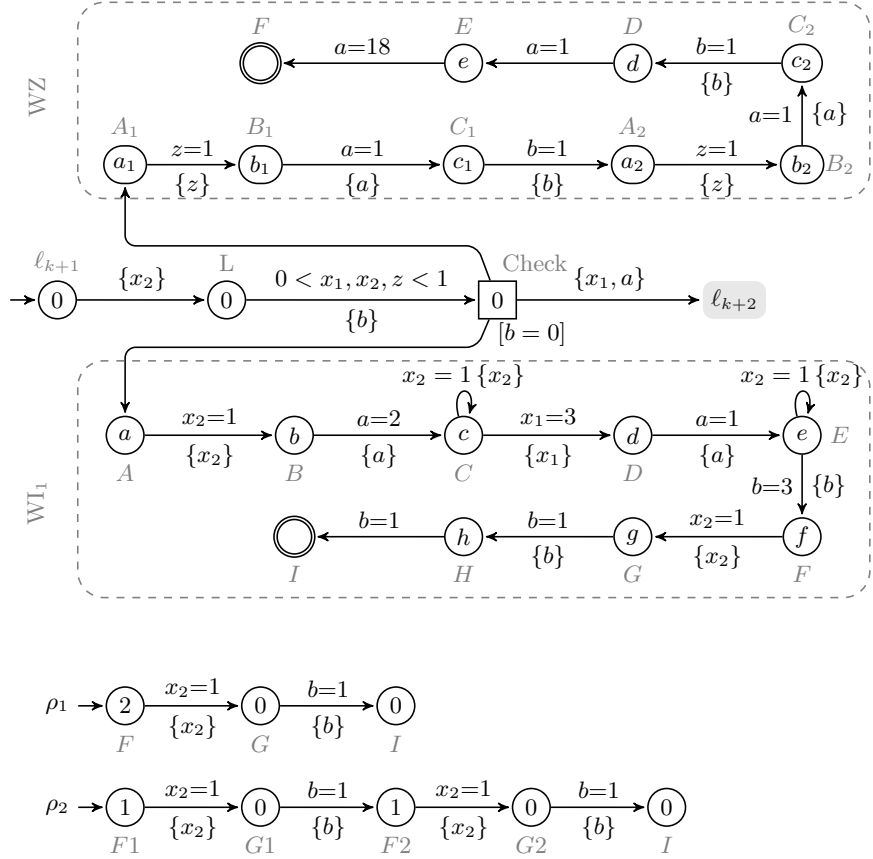
\begin{figure}[tbp]
\begin{center}

\begin{tikzpicture}[->,>=stealth',shorten >=1pt,auto,node distance=1cm,
  semithick,scale=0.9]
  
  	\node[initial,initial text={}, player1] at (0,0) (lk) {$0$} ;
       \node()[above of=lk,node distance=5mm,color=gray]{$\ell_{k+1}$};
       
       \node[player1] at (2.5,0) (L){$0$} ;
       \node()[above of=L,node distance=5mm,color=gray]{$\text{L}$};

       \node[player2] at (6.5,0) (chk){$0$} ;
       \node()[above right of=chk,node distance=7mm,color=gray]{$\text{Check}$};
        \node () [below right of=chk,node distance=6.5mm] {$[b=0]$};

       \node[fill=gray!20,rounded corners,fill opacity=0.9] at (10,0)(nxt){$\ell_{k+2}$}; 

       \path (lk) edge node {$\set{x_2}$} (L);
       \path (L) edge node[above] {$0<x_1,x_2,z<1$} node[below] {$\set{b}$} (chk);
       \path (chk) edge node {$\set{x_1,a}$} (nxt);
       
       \node[player1] (A1) at (1, 2) {$a_1$};
       \node()[above of=A1,node distance=5mm,color=gray]{$A_1$};
       
       \node[player1] (B1) at (3, 2) {$b_1$};
       \node()[above of=B1,node distance=5mm,color=gray]{$B_1$};
       
        \node[player1] (C1) at (6, 2) {$c_1$};
        \node()[above of=C1,node distance=5mm,color=gray]{$C_1$};
        
        \node[player1] (A2) at (8.5, 2) {$a_2$};
        \node()[above of=A2,node distance=5mm,color=gray]{$A_2$};
       
        \node[player1] (B2) at (11, 2) {$b_2$};
        \node()[right of=B2,node distance=5mm,color=gray]{$B_2$};
        
        \node[player1] (C2) at (11, 3.5) {$c_2$};
        \node()[above of=C2,node distance=5mm,color=gray]{$C_2$};
        
        \node[player1] (D) at (8.5, 3.5) {$d$};
        \node()[above of=D,node distance=5mm,color=gray]{$D$};
        
        \node[player1] (E) at (6, 3.5) {$e$};
        \node()[above of=E,node distance=5mm,color=gray]{$E$};
        
        \node[accepting,player1] (F) at (3, 3.5) {};
        \node()[above of=F,node distance=5mm,color=gray]{$F$};
       
       \path (A1) edge node {$z{=}1$} node[below] {$\set{z}$} (B1);
       \path (B1) edge node {$a{=}1$} node[below] {$\set{a}$} (C1);
       \path (C1) edge node {$b{=}1$} node[below] {$\set{b}$} (A2);       
       \path (A2) edge node[above] {$z{=}1$} node[below] {$\set{z}$} (B2);
       \path (B2) edge node[left] {$a{=}1$} node[right] {$\set{a}$} (C2);
       \path (C2) edge node[above] {$b{=}1$} node[below] {$\set{b}$} (D);
       \path (D) edge node [above] {$a{=}1$} (E);
       \path (E) edge node [above] {$a{=}18$} (F);
       
        \draw[rounded corners] (chk) -- (6.2,0.8) -- (1,0.8) -- (A1);

       \node[rotate=90,color=gray] at (-0.3, 3) (N) {$\text{WZ}$};

       \draw[dashed,draw=gray,rounded corners=10pt] (0.3,1.5) rectangle (12,4.4);
       
       \node[player1] (A) at (1, -2) {$a$};
       \node()[below of=A,node distance=5mm,color=gray]{$A$};
       
       \node[player1] (B) at (3.5, -2) {$b$};
       \node()[below of=B,node distance=5mm,color=gray]{$B$};
       
        \node[player1] (C) at (6, -2) {$c$};
        \node()[below of=C,node distance=5mm,color=gray]{$C$};
        
        \node[player1] (D) at (8.5, -2) {$d$};
        \node()[below of=D,node distance=5mm,color=gray]{$D$};
       
        \node[player1] (E) at (11, -2) {$e$};
        \node()[right of=E,node distance=5mm,color=gray]{$E$};
        
        \node[player1] (F) at (11, -3.5) {$f$};
        \node()[below of=F,node distance=5mm,color=gray]{$F$};
        
        \node[player1] (G) at (8.5, -3.5) {$g$};
        \node()[below of=G,node distance=5mm,color=gray]{$G$};        
        
        \node[player1] (H) at (6, -3.5) {$h$};
        \node()[below of=H,node distance=5mm,color=gray]{$H$};        
        
        \node[accepting,player1] (I) at (3.5, -3.5) {};
        \node()[below of=I,node distance=5mm,color=gray]{$I$};
       
       \path (A) edge node {$x_2{=}1$} node[below] {$\set{x_2}$} (B);
       \path (B) edge node {$a{=}2$} node[below] {$\set{a}$} (C);
       \path (C) edge node {$x_1{=}3$} node[below] {$\set{x_1}$} (D);       
       \path (C) edge [loop above] node[above =-0.7mm] {$x_2=1 \set{x_2}$} (C);
       
       \path (D) edge node[above] {$a{=}1$} node[below] {$\set{a}$} (E);
       \path (E) edge node[left] {$b{=}3$} node[right] {$\set{b}$} (F);
       \path (E) edge [loop above] node[above =-0.7mm] {$x_2=1 \set{x_2}$} (E);
       
       \path (F) edge node[above] {$x_2{=}1$} node[below] {$\set{x_2}$} (G);
       \path (G) edge node [above] {$b{=}1$} node[below] {$\set{b}$} (H);
       \path (H) edge node [above] {$b{=}1$} (I);
       
        \draw[rounded corners] (chk) -- (6.2,-0.7) -- (1,-0.7) -- (A);

       \node[rotate=90,color=gray] at (-0.3, -3) (N) {$\text{WI}_1$};

       \draw[dashed,draw=gray,rounded corners=10pt] (0.3,-0.9) rectangle (12,-4.4);

       \node(r1) at (0,-6) {$\rho_1$};
       
       \node[player1] (f) at (1, -6) {$2$};
       \node()[below of=f,node distance=5mm,color=gray]{$F$};
       
       \node[player1] (g) at (3, -6) {$0$};
       \node()[below of=g,node distance=5mm,color=gray]{$G$};
       
        \node[player1] (i) at (5, -6) {$0$};
        \node()[below of=i,node distance=5mm,color=gray]{$I$};

        \path (r1) edge (f);
        \path (f) edge node [above] {$x_2{=}1$} node[below] {$\set{x_2}$} (g);
        \path (g) edge node [above] {$b{=}1$} node[below] {$\set{b}$} (i);
        \node(r2) at (0,-7.5) {$\rho_2$};
       
       \node[player1] (f1) at (1, -7.5) {$1$};
       \node()[below of=f1,node distance=5mm,color=gray]{$F1$};
       
       \node[player1] (g1) at (3, -7.5) {$0$};
       \node()[below of=g1,node distance=5mm,color=gray]{$G1$};
       
       \node[player1] (f2) at (5, -7.5) {$1$};
       \node()[below of=f2,node distance=5mm,color=gray]{$F2$};
       
       \node[player1] (g2) at (7, -7.5) {$0$};
       \node()[below of=g2,node distance=5mm,color=gray]{$G2$};
       
        \node[player1] (i2) at (9, -7.5) {$0$};
        \node()[below of=i2,node distance=5mm,color=gray]{$I$};
       
      \path (r2) edge (f1);
       \path (f1) edge node [above] {$x_2{=}1$} node[below] {$\set{x_2}$} (g1);
        \path (g1) edge node [above] {$b{=}1$} node[below] {$\set{b}$} (f2);
        \path (f2) edge node [above] {$x_2{=}1$} node[below] {$\set{x_2}$} (g2);
        \path (g2) edge node [above] {$b{=}1$} node[below] {$\set{b}$} (i2);
\end{tikzpicture}

\caption{$\BReachOb(18, 40)$: Simulation of instruction : $\ell_{k+1}$:
  increment $C_1$. WZ is a template for two widgets $\text{WZ}^<$ and
  $\text{WZ}^>$, based on the actual values of price-rate
  parameters. $\text{WZ}^>$ has prices $b_1 = c_1 = c_2 =e= 1$ and
  rest are 0, while $\text{WZ}^<$ has $a_1 = d =e= 1$ and rest are
  zero. Similarly, widget $\text{WI}_1$ is template for two widgets
  $\text{WI}_1^{<}$ ($a=d=1,f=11,h=6$) and $\text{WI}_1^{>}$
  ($c=1,g=12,h=17$). Path $\rho_1$ is the shorthand notation with
  larger prices for longhand notation of path $\rho_2$ using prices
  0,1 only.}
\label{fig_undec_tb_inc_module}
\end{center}
\end{figure}

  Let $t_1$ and $t_2$ be respectively the time spent at locations
  $\ell_{k+1}$ and $L$. We want to check that $t_1+t_2
  =\frac{1}{2^{k+1}}$, $t_2 =\frac{1}{2^{k+1+c_1+1}{3^{k+1+c_2}}}$.
  This will ensure that the clocks keep track of the total time
  elapsed, as well as the increment of $C_1$.  Player 2 has two
  widgets at his disposal to check each of these:
  \begin{enumerate}
  \item $\text{WZ}$ is a widget (that can be used by Player 2) in
    order to check that the total time elapsed in any module
    corresponding is correct. More precisely Player 2 has the
    opportunity to check (by means of the module $\text{WZ}$) that the
    execution of the module corresponding to the $(k+1)$th instruction
    takes exaclty $\frac{1}{2^{k+1}}$ time units (recall that this
    time is recorded in clock $z$).
  \item $\text{WI}_1$ is a widget (that can be used by Player 2) in
    order to check that counter $C_1$ is indeed incremented properly.
  \end{enumerate}

  \noindent Upon entering the location Check, the values of clocks are
  $a = t_1 + t_2$, $z = 1 - \frac{1}{2^k} + t_1 +t_2$, $x_1 =
  t_2$, $x_2 = \frac{1}{2^{k+c_1}3^{k+c_2}} + t_1 + t_2 $ and $b=0$.

  \begin{itemize} 
  \item Widget $\text{WZ}$: The role of widget $\text{WZ}$ is to check
    if the value of the clock $z$ is $1 - \frac{1}{2^{k+1}}$ when the
    location Check is reached.  The PTG corresponding to $\text{WZ}$
    is depicted in Fig. \ref{fig_undec_tb_inc_module}, and a cost
    analysis of $\text{WZ}$ is presented in
    Table~\ref{tab3}.\footnote{Notice that in the table, the line
      ``time elapse'' represents time elapsed in the current location,
      and not the time elapsed before reaching the current location.}
    $\text{WZ}$ is a general template for two widgets, based on the
    actual values of the price-rates. These two widgets are
    $\text{WZ}^<$ and $\text{WZ}^>$. In widget $\text{WZ}^>$,
    price-rates $b_1 = c_1 = c_2 =e= 1$ and the other ones are zero,
    while in widget $\text{WZ}^<$, price-rates $a_1 = d =e= 1$ and the
    other ones are zero.
 
    It can be seen that if $t_1+t_2 < \frac{1}{2^{k+1}}$, then the
    total cost incurred in $\text{WZ}^<$ is strictly more than 18;
    similarly, if $t_1+t_2 > \frac{1}{2^{k+1}}$, then the total cost
    incurred in $\text{WZ}^>$ is strictly more than 18. Player 2 can
    choose one these widgets appropriately. In case $t_1+t_2=
    \frac{1}{2^{k+1}}$, then the total cost incurred in either widget
    is exactly 18.
 
    \begin{table}[th]
      \caption{$\BReachOb(18, 40)$ : Cost Incurred in $\text{WZ}$: Total time elapsed= $19-t<20$}
      \label{tab3}
      \begin{tabular}{|c|c|c|c|c|c|}
        \hline
        $\begin{array}{c} \mbox{ Location }\rightarrow \\  i \in
          \{1,2\}\end{array}$ & $A_i$ & $B_i$ & $C_i$ & $D$ & $E$ \\
        \hline 
   
        $\begin{array}{c}z \\ \mbox{on entry }\end{array}$& 
        $\begin{array}{c}1-\frac{1}{2^{k}} +t_1+t_2\\= 1-
          \frac{1}{2^k} + t \\t = t_1+t_2  \end{array}$& 
        0& 
        $1-\frac{1}{2^k}$& 
        -& 
        -
        \\ \hline
   
        $\begin{array}{c}a \\ \mbox{on entry}\end{array}$& 
        $t$ & 
        $\frac{1}{2^k}$& 
        0& 
        1
        \\ \hline
   
        $\begin{array}{c}b \\ \mbox{ on entry}\end{array}$& 
        0& 
        $\frac{1}{2^k}-t$& 
        $1-t$& 
        0& 
        -
        \\ \hline
   
        $\begin{array}{c}\mbox{time} \\ \mbox{elapsed at}\end{array}$& 
        $\frac{1}{2^k}-t$& 
        $1-\frac{1}{2^k}$& 
        $t$& 
        $1-t$& 
        17
        \\ \hline
        $\begin{array}{c}\mbox{Widget $\text{WZ}^>$} \\ \mbox{check $t
            > \frac{1}{2^{k+1}}$} \\ \mbox{prices} \\  b_1,c_1,c_2,e :
          1 \\ \mbox{ rest : 0}\end{array}$ \\ 
        \hline
        $\begin{array}{c}\mbox{cost incurred at} \\  \end{array}$& 
        0& 
        $1-\frac{1}{2^k}$ at $B_1$& 
        $\begin{array}{c} 2t \\ t \mbox{ at } C_1 \\ t \mbox{ at }
          C_2 \end{array}$& 
        0& 
        17 \\ 
        \hline
        $\begin{array}{c}\mbox{Total Cost} \\ \mbox{at target}\end{array}$& 
        -& 
        -& 
        -& 
        -& 
        $\begin{array}{c}1 - \frac{1}{2^k} + 2t + 17\\ =
          18~if~t=\frac{1}{2^{k+1}}\\>18~if~t>\frac{1}{2^{k+1}} \end{array}$
        \\ \hline
        $\begin{array}{c}\mbox{Widget $\text{WZ}^<$} \\ \mbox{check
            $t < \frac{1}{2^{k+1}}$} \\ \mbox{ prices} \\  a_1,d,e :1
          \\ \mbox{ rest : 0}\end{array}$ \\ 
        \hline
        $\begin{array}{c}\mbox{cost incurred at} \\ \end{array}$& 
        $\begin{array}{c}\frac{1}{2^k}-t \\ \mbox{at } A_1 \\ \end{array}$& 
        0& 
        $0$& 
        $1-t$& 
        17\\  
        \hline
   
        $\begin{array}{c}\mbox{Total Cost} \\ \mbox{at target}\end{array}$& 
        -& 
        -& 
        -& 
        -& 
        $\begin{array}{c}18 + \frac{1}{2^k} -2t \\ =
          18~if~t=\frac{1}{2^{k+1}}\\>18~if~t<\frac{1}{2^{k+1}} \end{array}$
        \\ \hline

      \end{tabular}
    \end{table}
    
  \item Widget $\text{WI}_1$: The widget $\text{WI}_1$ in
    Fig.~\ref{fig_undec_tb_inc_module} ensures that upon entering the
    location Check, the value of the clock $x_2 = t_2=
    \frac{1}{2^{k+1+c_1+1}{3^{k+1+c_2}}} = \frac{1}{12} \times
    \frac{1}{2^{k+c_1}{3^{k+c_2}}}=\frac{n}{12}$. This indeed accounts
    for increment of counter $C_1$, and also for reaching the end of
    $k+1^{th}$ instruction while keeping the value of $C_2$ unchanged.

    The widget $\text{WI}_1$ in Fig.~\ref{fig_undec_tb_inc_module} is
    again a general template for two widgets $\text{WI}_1^{<}$ and
    $\text{WI}_1^{>}$, obtained by fixing the prices at various
    locations.  The prices $a,b,c,d,e,f$ are values $>0$. It must be
    noted that we only need prices in $\{0,1\}$ for each location;
    using general prices is a short hand notation for a longer path
    which will use only prices from $\{0,1\}$.  In widget
    $\text{WI}_1^{<}$, (in the shorthand notation), we have prices
    $a=d=1,f=11,h=6$ and the rest are zero.  Likewise, in widget
    $\text{WI}_1^{>}$, we have (in shorthand notation), prices
    $c=1,g=12,h=17$ and the rest of the prices 0.  Player 2 uses
    $\text{WI}_1^{<}$ to check if $t_2 < \frac{n}{12}$, and uses
    $\text{WI}_1^{>}$ to check if $t_2 > \frac{n}{12}$.  Table
    \ref{tab4} runs the reader through the widgets $\text{WI}_1^{<}$
    and $\text{WI}_1^{>}$. While reading the table, keep in mind that
    $n = \frac{1}{2^{k+c_1}3^{k+c_2}}$ and that $t=t_1+t_2$.  As can
    be seen from the table, the total cost incurred is exactly 18 iff
    $t_2=\frac{n}{12}$.
 
    \begin{table}[!ht]
      \caption{$\BReachOb(18, 40)$: Cost incurred in
        $\text{WI}_1$. Recall $t_1+t_2=t$. Also, total 
        time elapsed in $\text{WI}_1^{>}$ and $\text{WI}_1^{<}$ in the
        long hand notation 
        is $\leq$ 3 + 12 + 17 = 32}
      \label{tab4}
      \begin{tabular}{|c|c|c|c|c|c|c|c|c|c|}
        \hline
        Loc $\rightarrow$ & $A$ & $B$ & $C$ & $D$ & $E$ & $F$ & $G$ &
        $H$ & $I$ \\ \hline 
 
        $\begin{array}{c}x_1\\ \mbox{on entry } \end{array}$ &
        $n+t$& 
        $1+n+t_1$& 
        $2+n$& 
        $0$& 
        -& 
        -& 
        -&
        - & 
        - 
        \\ \hline
        
        $\begin{array}{c}x_2\\ \mbox{on entry} \end{array}$ &
        $t_2$& 
        0& 
        $1-t_1$& 
        $1-t_1-n$& 
        $1-t_1$& 
        $t_2$& 
        0 &
        - & 
        - 
        \\ \hline

        $\begin{array}{c} a \\ \mbox{ on entry} \end{array}$ &
        $t$& 
        $1+t_1$& 
        $0$& 
        $1-n$& 
        $0$& 
        -& 
        -&
        - & 
        - 
        \\ \hline

        $\begin{array}{c} b\\ \mbox{on entry } \end{array}$ &
        0& 
        $1-t_2$& 
        $2-t$& 
        $3-t-n$& 
        $3-t$& 
        0& 
        $1-t_2$&
        0 & 
        1 
        \\ \hline

        $\begin{array}{c} \mbox{time }\\ \mbox{elapsed at} \end{array}$ &
        $1-t_2$& 
        $1-t_1$& 
        $1-n$& 
        $n$& 
        $t$& 
        $1-t_2$& 
        $t_2$ & 
        $1$ & 
        $0$ 
        \\ \hline  
    
        $\begin{array}{c} \mbox{$\text{WI}_1^{<}$ }\\ 
          \mbox{checks}\\
          \mbox{if $t_2 {<} \frac{n}{12}$} \\ \mbox{ prices } \\
          a,d{=}1 \\ f{=}11{,} h{=}6\\ \mbox{rest : 0} \end{array}$  
        \\ \hline

        $\begin{array}{c} \mbox{cost} \\\mbox{incurred at}\\  \end{array}$ &
        $1-t_2$& 
        $0$& 
        $0$& 
        $n$& 
        $0$& 
        $11-11t_2$& 
        0&
        6 & 
        - 
        \\ \hline

        $\begin{array}{c} \mbox{Total cost}\\ \mbox {at target } \end{array}$ &
        -& 
        -& 
        -& 
        -& 
        -& 
        -& 
        -&
        - & 
        $\begin{array}{c} 18 -12t_2 +n \\ {=}18~if~t_2 {=}
          \frac{n}{12}\\{>}18~if~t_2{<}\frac{n}{12} \end{array}$
        \\ \hline

        $\begin{array}{c} \mbox{$\text{WI}_1^{>}$}\\ 
          \mbox{checks}\\
          \mbox{if $t_2 {>} \frac{n}{12}$} \\ \mbox{prices} \\
          c{=}1,g{=}12,\\ h{=}17 \\ \mbox{rest : 0} \end{array}$  
        \\ \hline

        $\begin{array}{c} \mbox{cost} \\ \mbox{incurred at}\\  \end{array}$ &
        $0$& 
        $0$& 
        $1-n$& 
        $0$& 
        $0$& 
        $0$& 
        $12t_2$&
        $17$ & 
        - 
        \\ \hline

        $\begin{array}{c} \mbox{Total cost}\\  \mbox {at target } \end{array}$ &
        -& 
        -& 
        -& 
        -& 
        -& 
        -& 
        -&
        - & 
        $\begin{array}{c} 18 +12t_2 -n \\ {=}18~if~t_2
          {=}\frac{n}{12}\\{>}18~if~t_2{>}\frac{n}{12} \end{array}$
        \\ \hline
      \end{tabular}
    \end{table}
  \end{itemize}

  To summarize the simulation of the $(k+1)$th instruction, which is
  an increment $C_1$ instruction: assume we enter $\ell_{k+1}$ with
  values $x_1= \frac{1}{2^{k+c_1}3^{k+c_2}}$, while $a=b=x_2=0$,
  and $z=1-\frac{1}{2^k}$. First let us consider the case when Player
  1 correctly simulates the machine, respecting the time limit: then
  Player 1 spends a total of $\frac{1}{2^{k+1}}$ time across $\ell_{k+1}$
  and $L$, and a time $\frac{1}{2^{k+1+c_1+1}{3^{k+1+c_2}}}$ at
  $L$. Player 2 has 3 possibilities : (i)either Player 2 directly goes
  to the next instruction $\ell_{k+2}$ and thus the simulation of the
  machine goes on, or (ii) Player 2 goes to one of the widgets
  $\text{WZ}^{>}$ or $\text{WZ}^{<}$; in this case, since Player 1 has
  spent the right amount of time in $\ell_{k+1}$ and $L$ together, he
  achieves his objective, or (iii) Player 2 goes to one of the widgets
  $\text{WI}_1^{>}$ or $\text{WI}_1^{<}$; again, in this case, since
  Player 1 has spent the right amount of time at $L$, he achieves his
  objective.  Let us now turn to the case when Player 1 does not spend
  the right amount of time in $\ell_k$ and $L$ together; in this case,
  Player 2 can always enter one of the widgets $\text{WZ}^{>}$ or
  $\text{WZ}^{<}$ and reach a target with a cost $>18$.  Similarly, if
  Player 1 does not spend the right amount of time in $L$ (and
  therefore did not increment $C_1$ properly), then again Player 2 has
  the possibility to reach a target with a cost $>18$ using one of the
  widgets $\text{WI}_1^{>}$ or $\text{WI}_1^{<}$.

  \medskip
  
  \paragraph{Simulate decrement instruction}
  The module to decrement counter $C_1$ is the same as the module to
  increment $C_1$ in Fig.~\ref{fig_undec_tb_inc_module}. We replace
  only the widget $\text{WI}_1$ by the widget $\text{WD}_1$.  This
  widget ensures that the time spent $t_2$ at $L$ (and hence the value
  of $x_2$) is
  $\frac{1}{2^{k+1+c_1-1}3^{k+1+c_2}}=\frac{1}{2^{k+c_1}3^{k+1+c_2}}=
  \frac{1}{3}\frac{1}{2^{k+c_1}{3^{k+c_2}}}$. $\text{WD}_1$ can be
  obtained by a simple modification of the prices in $\text{WI}_1$.

  \medskip

  \paragraph{Simulate Zero-check instruction}
  Consider $\ell_{k+1}$: if $C_1 = 0$ then goto $l^1_{k+2}$ else goto
  $l^2_{k+2}$. During the simulation of this instruction, we need to
  ensure that all the clocks are updated to account for reaching the
  end of $(k+1)$th instruction and the counter values remain
  unchanged.

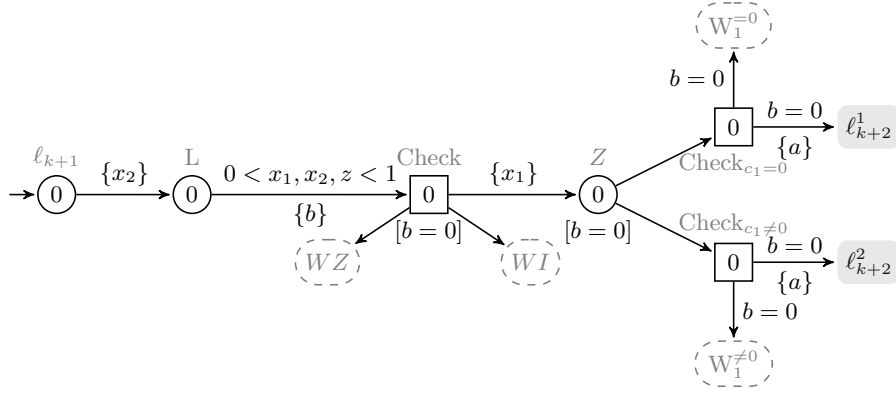
\begin{figure}[tbp]
  \begin{center}
\begin{tikzpicture}[->,>=stealth',shorten >=1pt,auto,node distance=1cm,
  semithick,scale=0.9]  
  
    	\node[initial,initial text={}, player1] at (-1,0) (lk) {$0$} ;
       \node()[above of=lk,node distance=5mm,color=gray]{$\ell_{k+1}$};
       
       \node[player1] at (1,0) (L){$0$} ;
       \node()[above of=L,node distance=5mm,color=gray]{$\text{L}$};

       \node[player2] at (4.5,0) (chk){$0$} ;
       \node()[above of=chk,node distance=5mm,color=gray]{$\text{Check}$};
        \node () [below of =chk,node distance=5mm] {$[b=0]$};

        \node[widget,color=gray] at (3,-1) (D1){$WZ$ } ;
	\node[widget,color=gray] at (6,-1) (D2){$WI$ };
    
       \path (lk) edge node {$\set{x_2}$} (L);
       \path (L) edge node[above] {$0<x_1,x_2,z<1$} node[below] {$\set{b}$} (chk);
       \path (chk) edge (D1);
       \path (chk) edge (D2);
	
      \node[player1] at (7,0) (Z) {$0$};
      \node()[above of=Z,node distance=5mm,color=gray]{$Z$};
       \node () [below of =Z,node distance=5mm] {$[b=0]$};

      \node[player2] at (9,1) (B){$0$} ;
       \node()[below of=B,node distance=5mm,color=gray]{$\text{Check}_{c_1=0}$};
      \node[widget,color=gray] at (9,2.5) (B1){$\text{W}_1^{= 0}$ };

      \node[player2] at (9,-1) (C){$0$} ;
      \node()[above of=C,node distance=5mm,color=gray]{$\text{Check}_{c_1\neq 0}$};
      \node[widget,color=gray] at (9,-2.5) (C1){$\text{W}_1^{\neq 0}$};

       \node[fill=gray!20,rounded corners,fill opacity=0.9] at (11,1)(nxt1){$\ell_{k+2}^1$};
       \node[fill=gray!20,rounded corners,fill opacity=0.9] at (11,-1)(nxt2){$\ell_{k+2}^2$}; 
       
       \path (chk) edge node {$\set{x_1}$} (Z);
       \path (Z) edge (B);
       \path (Z) edge (C);
       \path (B) edge node[above] {$b=0$} node[below] {$\set{a}$}(nxt1);
       \path (B) edge node {$b=0$}(B1);
       \path (C) edge node {$b=0$} node[below] {$\set{a}$}(nxt2);
       \path (C) edge node {$b=0$}(C1);

 \end{tikzpicture}
 \caption{$\BReachOb(18, 40)$: Simulation of instruction zero check
   $C_1=0$. Widget WZ given in Fig.~\ref{fig_undec_tb_inc_module}, and
   WI is similar to $\text{WI}_1$ differing only in prices. WI
   verifies if $x_2 = \frac{1}{6} \frac{1}{2^{k+c_1}3^{k+c_2}}$.}
\label{fig_undec_tb_zeroCheck}
\end{center}
\end{figure}

  The module for zero check is given in
  Fig.~\ref{fig_undec_tb_zeroCheck}.  At the location labelled Check,
  the widgets $\text{WZ}$ and $\text{WI}$ ensure that the clocks $z$
  and $x_2$ are updated to account for reaching the end of the
  $(k+1)$th instruction.  This is similar to the increment module.  If
  the values of clocks on entering $\ell_{k+1}$ are $z=1-\frac{1}{2^k}$,
  $x_1=\frac{1}{2^{k+c_1}3^{k+c_2}}$ then upon entering Check they are
  $z = 1-\frac{1}{2^{k+1}}$,
  $x_2=\frac{1}{2^{k+1+c_1}3^{k+1+c_2}}$. No time elapses at the Check
  location.  At location $Z$, no time elapses, and Player 1 guesses
  the value of counter $C_1$ and goes to either of the locations
  $\text{Check}_{c_1 = 0}$ or $\text{Check}_{c_1 \neq 0}$. Based on
  the choice of Player 1, Player 2 can go to one of the widgets
  $\text{W}_1^{=0}$ (in Fig.~\ref{fig_undec_tb_wz}) or
  $\text{W}_1^{\neq 0}$ (similar to $\text{W}_1^{=0}$), if he suspects
  that Player 1 has made a wrong guess.

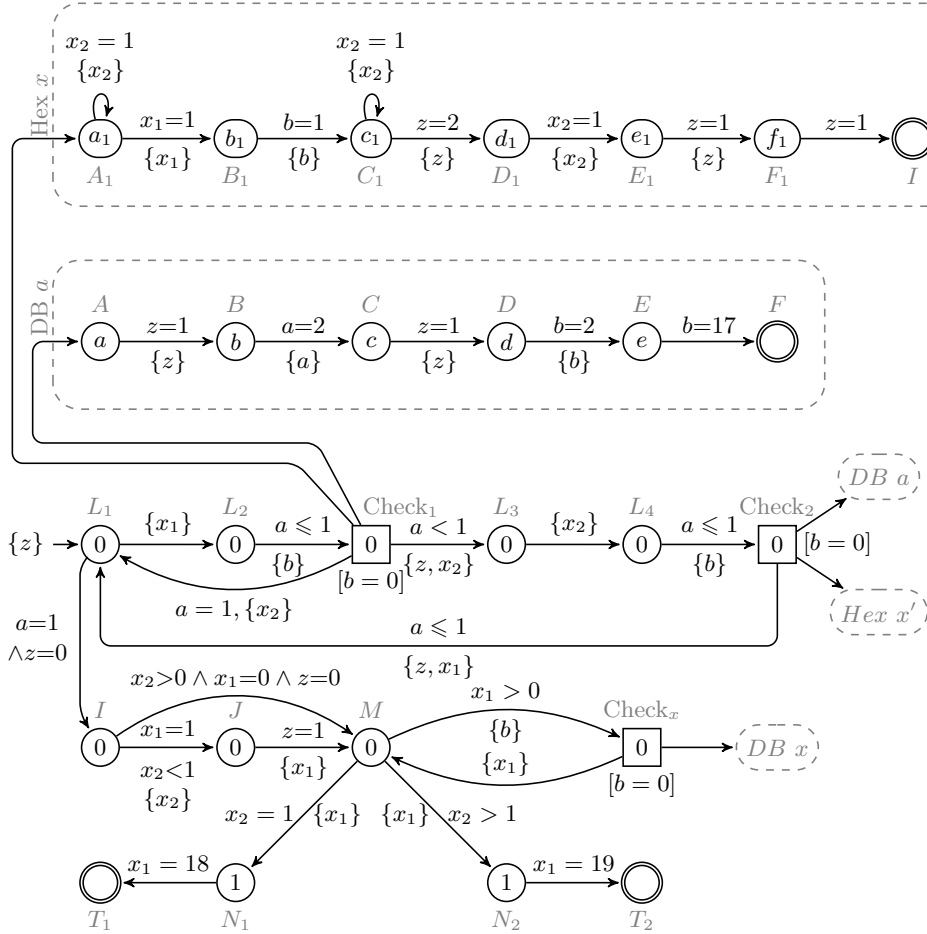
\begin{figure}[tbp]
\begin{center}
\begin{tikzpicture}[->,>=stealth',shorten >=1pt,auto,node distance=1cm,
  semithick,scale=0.9]
       \node[initial,initial text={$\set{z}$}, player1] at (-1,0) (l1) {$0$} ;
       \node()[above of=l1,node distance=5mm,color=gray]{$L_1$};
       
       \node[player1] at (1,0) (l2){$0$} ;
       \node()[above of=l2,node distance=5mm,color=gray]{$L_2$};
       
       \node[player2] at (3,0) (chk1){$0$} ;
       \node()[above right of=chk1,node distance=7mm,color=gray,xshift=-1mm]{$\text{Check}_1$};
        \node () [below of =chk1,node distance=5mm] {$[b=0]$};
        
        \node[player1] at (5,0) (l3) {$0$} ;
       \node()[above of=l3,node distance=5mm,color=gray]{$L_3$};
       
       \node[player1] at (7,0) (l4){$0$} ;
       \node()[above of=l4,node distance=5mm,color=gray]{$L_4$};
       
       \node[player2] at (9,0) (chk2){$0$} ;
       \node()[above of=chk2,node distance=5mm,color=gray]{$\text{Check}_2$};
       \node () [right of=chk2,node distance=8mm] {$[b=0]$};
       
        \node[widget,color=gray] at (10.5, 1) (D1){$DB~a$ } ;
	\node[widget,color=gray] at (10.5,-1) (D2){$Hex~x'$ };

       \path (l1) edge node[above] {$\set{x_1}$} (l2);
       \path (l2) edge node[above] {$a \leq 1$} node[below,xshift=-2mm] {$\set{b}$}(chk1);
       \path (chk1) edge [bend left=30] node {$a=1,\set{x_2}$} (l1);
       
       \path (chk1) edge node[above] {$a < 1$} node[below] {$\set{z,x_2}$}(l3);       
       \path (l3) edge node[above] {$\set{x_2}$} (l4);
       \path (l4) edge node[above] {$a \leq 1$} node[below] {$\set{b}$}(chk2);
       
        \draw[rounded corners] (chk2) -- (9, -1.5) --  node[above] {$a \leq 1$} node[below] {$\set{z,x_1}$} (-1,-1.5) --(l1);
       
       \path (chk2) edge (D1);
       \path (chk2) edge (D2);
      
       \node[player1] at (-1,-3) (i){$0$} ;
       \node()[above of=i,node distance=5mm,color=gray]{$I$};
       
       \node[player1] at (1,-3) (j){$0$} ;
       \node()[above of=j,node distance=5mm,color=gray]{$J$};
       
       \node[player1] at (3,-3) (m){$0$} ;
       \node()[above of=m,node distance=5mm,color=gray]{$M$};
       
       \node[player2] at (7,-3) (chkx){$0$} ;
       \node()[above of=chkx,node distance=5mm,color=gray]{$\text{Check}_x$};
       \node () [below of =chkx,node distance=5mm] {$[b=0]$};
       
       \node[widget,color=gray] at (9, -3) (D3){$DB~x$ } ;
       
       \node[player1] at (1,-5) (n1){$1$} ;
       \node()[below of=n1,node distance=5mm,color=gray]{$N_1$};
       
       \node[accepting,player1] at (-1,-5) (t1){} ;
       \node()[below of=t1,node distance=5mm,color=gray]{$T_1$};
       
       \node[player1] at (5,-5) (n2){$1$} ;
       \node()[below of=n2,node distance=5mm,color=gray]{$N_2$};
       
       \node[accepting,player1] at (7,-5) (t2){} ;
       \node()[below of=t2,node distance=5mm,color=gray]{$T_2$};
       
       \draw[rounded corners] (l1) -- (-1.3, -0.3)   -- node[left] 
       {$\begin{array}{c} \mbox{$a {=} 1$} \\ \wedge z{=}0\end{array}$} (-1.3,-2.5)  --   (i);
       
       \path (i) edge node[above] {$x_1 {=} 1$} node[below] {$\begin{array}{c} \mbox{$x_2 {<}1$}\\ \set{x_2}\end{array}$} (j);       
       \path (j) edge node[above] {$z {=} 1 $} node[below] {$\set{x_1}$} (m);
       \path (i) edge [bend left=35] node[above] {$x_2{>}0\wedge x_1{=}0\wedge z{=}0$}(m);
       
       \path (m) edge [bend left=25] node[above] {$x_1 >0$} node[below] {$\set{b}$} (chkx);
       \path (chkx) edge [bend left=25] node[above] {$\set{x_1}$} (m);
       \path (chkx) edge (D3);
       \path (m) edge node[right] {$x_2 >1$} node[left] {$\set{x_1}$} (n2);
       \path (m) edge node[left] {$x_2 =1$} node[right] {$\set{x_1}$} (n1);
       \path (n1) edge node[above] {$x_1=18$}(t1);
       \path (n2) edge node[above] {$x_1=19$}(t2);
       
       \node[player1] at (-1,3) (A) {$a$} ;
       \node()[above of=A,node distance=5mm,color=gray]{$A$};
       
       \node[player1] at (1,3) (B){$b$} ;
       \node()[above of=B,node distance=5mm,color=gray]{$B$};
       
       \node[player1] at (3,3) (C){$c$} ;
       \node()[above of=C,node distance=5mm,color=gray]{$C$};
        
        \node[player1] at (5,3) (D) {$d$} ;
       \node()[above of=D,node distance=5mm,color=gray]{$D$};
       
       \node[player1] at (7,3) (E){$e$} ;
       \node()[above of=E,node distance=5mm,color=gray]{$E$};
       
       \node[accepting,player1] at (9,3) (F){} ;
       \node()[above of=F,node distance=5mm,color=gray]{$F$};
       
    \draw[rounded corners] (chk1) -- (2.2, 1.5) -- (-2,1.5) -- (-2,3) -- (A);

       \path (A) edge node {$z{=}1$} node[below] {$\set{z}$} (B);
       \path (B) edge node {$a{=}2$} node[below] {$\set{a}$} (C);
       \path (C) edge node {$z{=}1$} node[below] {$\set{z}$} (D);              
       \path (D) edge node[above] {$b{=}2$} node[below] {$\set{b}$} (E);
       \path (E) edge node[above] {$b{=}17$}(F);
       
       \node[rotate=90,color=gray] at (-1.9, 3.5) (N) {$\text{DB}~a$};
       \draw[dashed,draw=gray,rounded corners=10pt] (-1.7,2) rectangle (9.7,4.2);

       \node[player1] at (-1,6) (A1) {$a_1$} ;
       \node()[below of=A1,node distance=5mm,color=gray]{$A_1$};
       
       \node[player1] at (1,6) (B1){$b_1$} ;
       \node()[below of=B1,node distance=5mm,color=gray]{$B_1$};
       
       \node[player1] at (3,6) (C1){$c_1$} ;
       \node()[below of=C1,node distance=5mm,color=gray]{$C_1$};
        
        \node[player1] at (5,6) (D1) {$d_1$} ;
       \node()[below of=D1,node distance=5mm,color=gray]{$D_1$};
       
       \node[player1] at (7,6) (E1){$e_1$} ;
       \node()[below of=E1,node distance=5mm,color=gray]{$E_1$};
       
       \node[player1] at (9,6) (F1){$f_1$} ;
       \node()[below of=F1,node distance=5mm,color=gray]{$F_1$};
       
       \node[accepting,player1] at (11,6) (I){} ;
       \node()[below of=I,node distance=5mm,color=gray]{$I$};
       
       \draw[rounded corners] (chk1) -- (1.9, 1.2) -- (-2.3,1.2) -- (-2.3,6) -- (A1);
       
       \path (A1) edge node {$x_1{=}1$} node[below] {$\set{x_1}$} (B1);
       \path (A1) edge [loop above] node [above] {$\begin{array}{c}x_2=1 \\ \set{x_2}\end{array}$} (A1);
       \path (B1) edge node {$b{=}1$} node[below] {$\set{b}$} (C1);
       \path (C1) edge node {$z{=}2$} node[below] {$\set{z}$} (D1);              
       \path (C1) edge [loop above] node [above] {$\begin{array}{c}x_2=1 \\ \set{x_2}\end{array}$} (C1);
       \path (D1) edge node[above] {$x_2{=}1$} node[below] {$\set{x_2}$} (E1);
       \path (E1) edge node[above] {$z{=}1$} node[below] {$\set{z}$}(F1);
       \path (F1) edge node[above] {$z{=}1$}(I);
       
       \node[rotate=90,color=gray] at (-1.9, 6.5) (N) {$\text{Hex}~x$};
       \draw[dashed,draw=gray,rounded corners=10pt] (-1.7,5) rectangle (11.7,8);

\end{tikzpicture}

\caption{$\BReachOb(18, 40)$: Widget $\text{W}_1^{=0}$ entered with $a
  = \frac{1}{2^{k+1}}=\alpha$ and $x_1=x_3 =
  \frac{1}{2^{k+1+c_1}3^{k+1+c_2}}=\beta$. Widget $\text{DB}_a$ checks
  if upon entry $a=\alpha+t$ and $z=t$ then $t =\alpha$.  Times spent
  : $1-t$ at $A$, $1-\alpha$ at $B$, $\alpha$ at $C$, $t$ at $D$ and
  17 at $E$. $\text{DB}_a$ stands for 2 widgets, to check $\alpha > t$
  ($a,c,e=1$) and to check $\alpha < t$ ($d,b,e=1$). Widget
  $\text{Hex}~x$ checks that if $x_1 =t$, $x_2 = \beta + t+k$ and $z =
  t+k$ then $t = 6\beta$. Again, this stands for two widgets, one when
  $t< 6\beta$ ( $a_1=1,e_1=6,f_1=17$), and the other when $t>6\beta$
  ($d_1=6,b_1=1,f_1=12$). $\text{DB}_x$ is same as $\text{DB}_a$ where
  $a,z$ are replaced by $x_2$ and $x_1$ respectively. 
  Widget $\text{Hex}~x'$ is the same as widget $\text{Hex}~x$ with roles 
  of $x_1$ and $x_2$ reversed.}
\label{fig_undec_tb_wz}
\end{center}
\end{figure}

  \begin{itemize}
  \item Widget $\text{W}_1^{=0}$ given in Fig.~\ref{fig_undec_tb_wz}.
    We have $x_2=\frac{1}{2^{k+1+c_1}3^{k+1+c_2}}$, on
    entering the $L_1$ of Widget $\text{W}_1^{=0}$. To check if $c_1=0$, we first
    convert $x_2$ to be of the form $\frac{1}{2^{c_1}3^{c_2}}$ by
    multiplying $x_2$ by $6^{k+1}$.

    Let $\alpha=\frac{1}{2^{k+1}}$ and
    $\beta=\frac{1}{2^{k+1+c_1}3^{k+1+c_2}}$.  The location $L_1$ is
    entered with $a = \frac{1}{2^{k+1}}=\alpha$ and
    $x_2=\frac{1}{2^{k+1+c_1}3^{k+1+c_2}} = \beta$. Let $t_1,t_2$
    be the times spent at locations $L_1, L_2$ respectively. Then, on
    entering $\text{Check}_1$, we have $a = \alpha + t_1 + t_2$,
    $z=t_1+t_2$, $x_2 = \beta + t_1+t_2$ and $x_1 = t_2$. The
    widget $\text{DB}~a$ (Fig.~\ref{fig_undec_tb_wz}) ensures that
    $t_1+t_2 = \alpha$, i.e; $a$ has been doubled. Similarly, the
    widget $\text{Hex}~x$ (Fig.~\ref{fig_undec_tb_wz}) ensures that
    $t_2 = 6 \beta$. No time is spent at $\text{Check}_1$.

    We repeatedly keep multiplying $a$ by 2 until $a$ becomes equal to
    1.  For this, we once take the path $L_1$ to $\text{Check}_1$,
    using clock $x_1$, and the next time, use the path $L_3$ to
    $\text{Check}_2$, using clock $x_2$.  Note that when $a$
    becomes 1 after $k+1$ iterations, we have also multiplied $\beta$
    with $6^{k+1}$. Due to the alternation of clocks $x_1,x_2$ in paths 
    $L_1$ to $\text{Check}_1$ and $\text{Check}_1$ to $\text{Check}_2$, 
    the value $\beta * 6^{k+1}$ could be either in $x_1$ or $x_2$ while the other clock is 0. 
    We ensure (via locations $I,J$) that 
    $x_2 = 6^{k+1}\beta$ and $x_1=0$ upon entering $M$. 
    Hence, we get after $k+1$ iterations,
    $a=1=2^{k+1}\alpha$ and $x_2= 6^{k+1}\beta=
    6^{k+1}\frac{1}{2^{k+1+c_1} 3^{k+1+c_2}} = \frac{1}{2^{c_1}
      3^{c_2}}$.  At this point, we are at location $M$.

    Now each time the loop between $M$ and $\text{Check}_x$ is taken,
    value in $x_2$ is doubled; the widget $\text{DB}~x$ checks this.
    Repeatedly doubling $x_2$ $i$ times gives $x_2 = 2^i
    \frac{1}{2^{c_1} 3^{c_2}}$.  If this value is 1, then we know
    $i=c_1$, and $c_2=0$.  When this happens we reach location $N_1$,
    from where a target is reached with cost 18.  If $c_2 \neq 0$,
    then after some $j$ iterations of the loop between $Check_x$ and
    $M$, we will obtain $x_2=\frac{1}{3^{c_2}}$.  Note that we can
    neither go to $N_1$ nor $N_2$ at this point, so the only option is
    to continue the loop between $M$ and $\text{Check}_x$. Clearly,
    $x_2$ will never become 1; so the only option is for $x_2$ to grow
    larger than 1. At that point, the transition to $N_2$ is enabled,
    and we reach a target with a cost 19.

  \item Time spent in widget $\text{W}_1^{=0}$: if $L_1$ was entered
    for the first time with $a = \frac{1}{2^{k+1}} = \alpha$, then the
    time spent in $L_1$ and $L_2$ before $\text{Check}_1$ is entered
    is $t=\alpha$.  After visiting $L_1,L_2$, the next time we use
    $L_3,L_4$. Since we enter $L_3$ for the first time with
    $a=2\alpha$, the time spent in $L_3,L_4$ is $2 \alpha$. The next
    time we visit $L_1$, we will be spending $4 \alpha$ and so on.
    Proceeding like this, we know that the total time spent in this
    loop before $M$ is reached is $\alpha + 2 \alpha + 4 \alpha +
    \cdots + \frac{1}{2}$ which is always $<1$. A similar argument
    holds for the time spent in the loop between $M$ and
    $\text{Check}_x$.  This apart, we spend 18 or 19 units of time (at
    $N_1$ or $N_2$) or at most 25 units of time in the widgets
    ($\text{Hex}~x$), thus adding upto a total time that is at most
    $<28$.

  \item The widget $\text{W}_1^{\neq 0}$ is similar to
    $\text{W}_1^{=0}$.  The loop between $M$ and $\text{Check}_x$ is
    retained, as is.  In addition, when $x_2<1$ and $x_1=0$, we go to
    a location $M_1$ from $M$. The idea is to first multiply $x_2$
    repeatedly till we obtain $x_2=\frac{1}{3^{c_2}}<1$, at which
    point of time, we go to $M_1$.  From $M_1$, we have a loop between
    $M_1$ and an urgent Player 2 location $\text{Check}'_x$, and a
    widget $\text{Triple}~x$ is attached to $\text{Check}'_x$. Each
    time we come back to $M_1$ from $\text{Check}'_x$, we reset $x_1$.
    Finally, if we get $x_2=1,x_1=0$, then we go from $M_1$ to a
    location $M_2$ having price 1.  Elapsing 18 units of time in
    $M_2$, we reach a target with cost 18.  However, if $x_2$ exceeds
    1, then with the guard $x_2 > 1, x_1=0?$, we go from $M_1$ to a
    location $M_3$ having price 1. To reach a target from $M_3$, one
    has to elapse 19 units of time, thereby incurring a cost
    19. Clearly, the route via $M_3$ will be needed iff $c_2=0$.

  \item The total time spent in $\text{W}_1^{\neq 0}$ will also be
    less than $28$.
  \end{itemize}

  To summarize the zero check instruction for $C_1$: assume we start
  with $x_1=\frac{1}{2^{k+c_1}3^{k+c_2}}$, $z=\frac{1}{2^k}$ at
  location $\ell_{k+1}$ of Fig.~\ref{fig_undec_tb_zeroCheck}.  Let us
  consider the case when Player 1 correctly simulates the instructions
  within time limits. In this case, location Check is reached with
  $z=\frac{1}{2^{k+1}}$ and $x_2=\frac{1}{2^{k+c_1+1}3^{k+c_2+1}}$.
  Player 2 has the possibility to test if this is indeed the case, by
  visiting widgets $\text{WZ},\text{WI}$; however Player~1 will
  achieve his objective in that case. At location $Z$, Player~1 then
  correctly guesses whether $C_1$ is zero or not, by appropriately
  choosing one of the locations $\text{Check}_{c_1=0}$ or
  $\text{Check}_{c_1 \neq 0}$. Again, Player 2 has the possibility to
  check if Player 1's guess is correct by visiting widgets
  $\text{W}_1^{=0}$ and $\text{W}_1^{\neq 0}$; however, Player 1 will
  achieve his objective here as well.  Now consider the case that
  Player 1 made a mistake: if he did not spend the right amount of
  time in $\ell_{k+1}$ and $T_2$, then Player 2 has the opportunity to
  punish him through the widgets $\text{WI}$ and $\text{WZ}$; if he
  made a wrong guess regarding $C_1$ being zero or non-zero, then
  again Player 2 has a possibility to punish him through the widgets
  $\text{W}_1^{=0}$ and $\text{W}_1^{\neq 0}$.

  \medskip

  \paragraph{Correctness of the construction}
  On entry into the location $\location_n$ (for the HALT instruction),
  we reset clock $x_1$ to 0; $\location_n$ has cost 1, and the edge
  coming out of $\location_n$ goes to a Goal location, with constraint
  $x_1\leq 18$.
  \begin{enumerate}
  \item Assume that the two counter machine halts. If Player 1
    simulates all the instructions correctly, he will incur a cost
    $\leq 18$, by either reaching the goal location after
    $\location_n$, or by entering a widget (the second case only
    occurs if Player~$2$ decides to check whether Player~$1$ simulates
    the machine faithfully.  If Player 1 makes an error in his
    computation, Player 2 can always enter an appropriate widget,
    making the cost greater than $18$.  In summary, if the two counter
    machine halts Player 1 has a strategy to achieve his goal
    (i.e., reaching a target location with a cost at most $18$).
  \item Assume that the two counter machine does not halt.
    \begin{itemize}
    \item If Player 1 simulates all the instructions correctly, and if
      Player 2 never enters a check widget, then Player 1 incurs cost
      $\infty$ as the path is infinite.  Notice that even in this
      case, the total time needed for the computations will be less
      than $1$, due to the strictly decreasing sequence of times
      chosen for simulating successive instructions.  In this case,
      Player 2 will never want to enter a widget, since he gets a
      higher payoff.
    \item Suppose now that Player 1 makes an error. In this case,
      Player 2 always has the possibility to reach a target set with a
      cost greater than $18$.
     \end{itemize}
     In summary, if the two counter machine does not halts Player 1
     does not have a strategy to achieve his goal.
   \end{enumerate}
   Thus, Player 1 incurs a cost at most $18$ iff he chooses the
   strategy of faithfully simulating the two counter machine, when the
   machine halts. When the machine does not halt, the cost incurred by
   Player 1 is greater than $18$ if Player 1 made a simulation error
   and Player 2 entered a widget. Otherwise, if a widget is not
   entered, then the run does not end and cost is $+\infty$.  \qed
 \end{proof}

 \noindent\textbf{Shorthand and Longhand notations used in the proof}: Note
 that in the shorthand notation used in widgets $\text{WI}_1^{>}$ and
 $\text{WI}_1^{<}$,
 \begin{itemize}
 \item we never have consecutive locations with price-rates different
   from $0$;
 \item on entering any location, there is a ``free'' clock with value
   0;
 \item further, all guards are of the form $x=c$ with reset of $x$ on
   all edges.
 \end{itemize}
 The time elapsed at a location is captured in the ``free'' clock
 which had value 0 while entering that location.  Consider, for
 example two consecutive locations $\ell_1$ and $\ell_2$ where
 $\ell_1$ has price $f>0$, $\ell_2$ has price 0, with an edge between
 $\ell_1$ and $\ell_2$ with guard $x=c$ and reset of $x$.  Let $y=0$
 on entering $\ell_1$.  If $t$ units of time was spent at $\ell_1$, we
 get $y=t, x=0$, and the rest of the clocks are incremented by $t$. A
 cost of $ft$ is incurred.  This can be replaced by a series of $2f-1$
 blocks, where a block looks like this: The first block contains a
 copy $\ell_{11}$ of $\ell_1$ and some dummy location $d_1$;
 $\ell_{11}$ has price 1, $d_1$ has price 0, and there is an edge from
 $\ell_{11}$ to $d_1$ with guard $x=c$, and reset of $x$. A cost $t$
 is incurred. In the second block, all locations have price 0. The
 second block begins from $d_1$ and ends in the second copy
 $\ell_{12}$ of $\ell_1$. The price of all copies of $\ell_1$ is same
 as that of $\ell_1$.  In the second block, the clock values are
 adjusted to be the same as they were, when they first entered
 $\ell_1$.  The third block begins with $\ell_{12}$, and is like the
 first block: it ends in a dummy location $d_2$. $d_2$ has price 0.
 The fourth block begins with $d_{2}$, and is like the second block:
 it ends in the 3rd copy $\ell_{13}$ of $\ell_1$ and so on.  The
 $2f-1$th block will end in location $d_{2f-3}$, with price 0.  The
 valuations of all clocks on entering $d_{2f-3}$ is the same as what
 they were on entering $\ell_2$, in the original transition from
 $\ell_1$ to $\ell_2$.  With no time elapse in $d_{2f-3}$, we go to
 $\ell_2$.
   
 The total cost incurred across the $2f-1$ blocks is $ft$.  If the
 time $t$ elapsed at $\ell_1$ is such that $\lceil t \rceil =r$, then the
 total time elapsed across the $2f-1$ blocks is $\leq fr$: $t$ time
 elapsed in the odd numbered blocks, and $r-t$ in the even numbered
 blocks, to restore clock values.  The point to note is that, as long
 as there are sufficiently many clocks, the above trick can be done.
 
 Consider the paths $\rho_1$ and $\rho_2$ in Figure \ref{fig_undec_tb_inc_module}. 
 As explained above, the ``free'' clock with value
   0 is $b$. Here $f$ is $2$ and $t = 1- v$ where $v$ is the value of clock $x_2$. 
   Clearly, the cost accumulated in path $\rho_1$ upon reaching $I$ is $2 * (1-v)$. 
   The clock values upon entering $F1$ of path $\rho_2$ are $x_2 = v$ and $b=0$. 
   Upon entering $G1$ the values are $x_2 = 0$ and $b = 1-v$. 
   Thus upon entering $F2$, $x_2 = v$ and $b=0$ and so on. The costs incurred in this 
   path are $1-v$ at $F1$ and $1-v$ at $F2$. Hence total cost accumulated is $2*(1-v)$ upon reaching $I$.


\subsection{Repeated reachability} 
\begin{lemma}
  \label{lem_undec_cost_buchi}
  The existence of a strategy for the repeated reachability objective
  $\RReachOb(\eta)$, for any $\eta \in \Rpos$ is undecidable for PTGs
  with 3 clocks or more.
\end{lemma}
\begin{proof}
  We prove the existence of a strategy with a repeated reachability
  objective ensuring the cost is within an interval $[-\eta, \eta]$ is
  undecidable, for any choice of $\eta >0$.  In order to obtain the
  undecidability result, we use a reduction from the halting problem
  of 2 counter machines.  Our reduction uses a PTG with 3 clocks, and
  arbitrary location prices, but no edge prices.
      
  We specify a module for each instruction of the two counter machine.
  On entry into a module, we have $x_1 =
  \frac{1}{5^{c_1}7^{c_2}}$,$x_2 = 0$ and $x_3=0$, where $c_1,c_2$ are
  the values of counters $C_1,C_2$. We construct a PTG whose building
  blocks are the modules for instructions. The role of Player 1 is to
  faithfully simulate the two counter machine, by choosing appropriate
  delays to adjust the clocks to reflect changes in counter values.
  Player 2 will have the opportunity to verify that Player 1 did not
  cheat while simulating the machine.  We shall now present modules
  for increment, decrement and zero check instructions.
   
\begin{figure}[tbp]
\centering
\begin{tikzpicture}[->,>=stealth',shorten >=1pt,auto,node distance=1cm,
  semithick,scale=0.9]
  \node[initial,initial text={}, player1] at (0,-.5) (lk) {$0$ } ;
   \node()[above of=lk,node distance=5mm,color=gray]{$\ell_k$};

  \node[player2] at (3,-.5) (chk){$0$} ;
     \node()[above of=chk,node distance=5mm,color=gray]{$\text{Check}$};

  \node () [below right of=chk,node distance=6.5mm,xshift=-3mm] {$[x_3=0]$};

   \node[fill=gray!20,rounded corners,fill opacity=0.9] at (6,-.5)(lk1){$\ell_{k+1}$};

  \node[ player2] at (0,-2)(A) {$-1$};
   \node()[above of=A,node distance=5mm,color=gray]{$A$};

  \node[player2] at (2.5,-2) (B){$4$};
   \node()[above of=B,node distance=5mm,color=gray]{$B$};

  \node[player2] at (5,-2) (C){$-3$};
   \node()[above of=C,node distance=5mm,color=gray]{$C$};

  \node[player2] at (7,-2) (D) {$0$};
   \node()[above of=D,node distance=5mm,color=gray]{$D$};

  \node[player2] at (9,-2) (E) {$0$};
   \node()[above of=E,node distance=5mm,color=gray]{$E$};

  \node[player2,accepting] at (11,-2) (F) {$0$} ;
   \node()[above of=F,node distance=5mm,color=gray]{$F$};

  \node () [below of=C,node distance=7mm] {$x_3=0$};
    
\path (lk) edge node {$x_1 {\leq} 1$} node[below] {$\set{x_3}$}(chk);
\path (chk) edge node[below] {$\set{x_2}$} (lk1);
\path (chk) edge node {} (A);

    \path (A) edge node {$ x_2 {=} 1$} node[below]{$\set{x_2}$} (B);
    \path (B) edge node {$x_1 {=} 2$} node[below]{$\set{x_1}$} (C);
    \path (C) edge node {$x_1 {=} 1$} node[below]{$\set{x_1}$} (D);
    \path (D) edge node {$x_2 {=} 2$} node[below]{$\set{x_2}$} (E);
    \path (E) edge node {$x_3 {=} 3$} node[below]{$\set{x_3}$} (F);
    \draw[rounded corners] (F) -- (11, -3) -- (0, -3) --  (A);
           
    \node[rotate=90,color=gray] at (-.8, -2) (N) {$\text{WD}_1$};
    \draw[dashed,draw=gray,rounded corners=10pt] (-.5,-3.2) rectangle (11.5,-1.25);

 \end{tikzpicture}
\caption{$\RReachOb(\eta)$: Simulation to decrement counter $C_1$}
\label{fig_undec_buchi_dec}
\end{figure}

  \paragraph{Simulation of decrement instruction} : The module to
  simulate the decrement of counter $C_1$ is given in
  Fig.~\ref{fig_undec_buchi_dec}.  We enter location $\ell_k$ with
  $x_1=\frac{1}{5^{c_1}7^{c_2}}, x_2 = 0$ and $x_3=0$. Lets denote by
  $x_{old}$ the value $\frac{1}{5^{c_1}7^{c_2}}$.  To correctly
  decrement $C_1$, Player 1 should choose a delay of $4x_{old}$ in
  location $\ell_k$.  At location $\text{Check}$, there is no time
  elapse.  Player 2 has two possibilities : ($i$) to go to
  $\ell_{k+1}$, or ($ii$) to enter the widget $\text{WD}_1$.  If
  Player 1 makes an error, and delays $4x_{old}+\varepsilon$ at
  $\ell_k$ ($\varepsilon \neq 0$)then Player 2 can enter the widget
  $\text{WD}_1$ and punish Player 1.  When we enter $\text{WD}_1$ for
  the first time, we have $x_1=x_{old}+4x_{old}+\varepsilon$,
  $x_2=4x_{old}+\varepsilon$ and $x_3=0$.  In $\text{WD}_1$, the cost
  of going from location $A$ to $F$ is $\varepsilon$. Also, when we
  get back to $A$ after going through the loop once, the clock values
  with which we entered $\text{WD}_1$ are restored; thus, each time,
  we come back to $A$, we restore the starting values with which we
  enter $\text{WD}_1$.  The third clock is really useful for this
  purpose only.
      
  Since all locations in $\text{WD}_1$ are Player 2 locations, Player
  2 can continue taking this loop as long as he pleases; each time
  incurring a cost $\varepsilon$.  Thus, for any choice of $\eta$,
  Player 2 can incur a cost that is not in $[-\eta, \eta]$ by taking
  the loop from $A$ to $F$ a large number of times.  Note however that
  when $\varepsilon=0$, then Player 1 will always achieve his
  objective: he will visit $F$ infinitely often with a cost $0 \in
  [-\eta,\eta]$ for any choice of $\eta$.
             
\begin{figure}[tbp]
\centering
\begin{tikzpicture}[->,>=stealth',shorten >=1pt,auto,node distance=1cm,
  semithick,scale=0.9]
  \node[initial,initial text={}, player1] at (0,-.5) (lk) {$0$ } ;
   \node()[above of=lk,node distance=5mm,color=gray]{$\ell_k$};

     \node[player2] at (3,-.5) (i){$0$} ;
     \node()[above of=i,node distance=5mm,color=gray]{$I$};
 
  \node[player2] at (6,-.5) (chk){$0$} ;
  \node()[above of=chk,node distance=5mm,color=gray]{$\text{Check}$};
  \node () [below right of=chk,node distance=6.5mm,xshift=1mm] {$[x_3=0]$};

   \node[fill=gray!20,rounded corners,fill opacity=0.9] at (9,-.5)(lk1){$\ell_{k+1}$};
   
\path (lk) edge node {$x_1 {=} 1$} node[below] {$\set{x_1}$}(i);
\path (i) edge node {$x_1 {\leq} 1$} node[below] {$\set{x_3}$}(chk);
\path (chk) edge node[above] {$\set{x_2}$} (lk1);

  \node[ player2] at (0,-2.5)(A) {$5$};
   \node()[above right of=A,node distance=8mm,color=gray]{$A$};

  \node[player2] at (2.5,-2.5) (B){$1$};
   \node()[above of=B,node distance=5mm,color=gray]{$B$};

  \node[player2] at (5,-2.5) (C){$-5$};
   \node()[above of=C,node distance=5mm,color=gray]{$C$};

  \node[player2] at (7,-2.5) (D) {$0$};
   \node()[above of=D,node distance=5mm,color=gray]{$D$};

  \node[player2] at (9,-2.5) (E) {$0$};
   \node()[above of=E,node distance=5mm,color=gray]{$E$};

  \node[player2,accepting] at (11,-2.5) (F) {$0$} ;
   \node()[above of=F,node distance=5mm,color=gray]{$F$};

  \node () [below of=C,node distance=7mm] {$x_3=0$};

  \draw[rounded corners] (chk) -- (6, -1.35) -- (0, -1.35) --  (A);

    \path (A) edge node {$ x_1 {=} 1$} node[below]{$\set{x_1}$} (B);
    \path (B) edge node {$x_2 {=} 2$} node[below]{$\set{x_2}$} (C);
    \path (C) edge node {$x_2 {=} 1$} node[below]{$\set{x_2}$} (D);
    \path (D) edge node {$x_1 {=} 2$} node[below]{$\set{x_1}$} (E);
    \path (E) edge node {$x_3 {=} 3$} node[below]{$\set{x_3}$} (F);
    \draw[rounded corners] (F) -- (11, -3.5) -- (0, -3.5) --  (A);
           
    \node[rotate=90,color=gray] at (-.8, -2.6) (N) {$\text{WI}_1$};
    \draw[dashed,draw=gray,rounded corners=10pt] (-.5,-3.7) rectangle (11.5,-1.75);

 \end{tikzpicture}
 
 \vspace{0.5cm}

Let $x_n = x_{new}$ and $x_o = x_{old}$
\vspace{0.25cm}

  \begin{tabular}{|c|c|c|c|c|}
  \hline
  Loc $\rightarrow$ & $\ell_{k}$ & $I$ & Check & $\ell_{k+1}$ \\\hline  
  $x_1$ on entry& $x_o$ & 0 & $x_n$ & $x_n$ \\\hline
  $x_2$ on entry& 0 & $1-x_o$ & $1 - x_o + x_n$ & $0$ \\\hline
  $x_3$ on entry& 0 & $1-x_o$& 0 & $0$ \\\hline
  Time elapsed & $1-x_o$ & $x_n$ & 0 & 0 \\\hline
  Cost incurred & 0 & 0 & 0 &0 \\\hline  
 \end{tabular}
 
\vspace{0.5cm}

 \begin{tabular}{|c|c|c|c|c|c|c|}
  \hline
  Loc $\rightarrow$ & $A$ & $B$ & $~~C~~$ & $D$ & $E$ & $F$ \\\hline  
  $\begin{array}{c}x_1 \\ \mbox{on entry}\end{array}$& $x_n$ &
  $0$ &
  $x_o$ &
  $1+x_o$ &
  0 &
  $x_n$ \\\hline
  $\begin{array}{c}x_2 \\ \mbox{on entry}\end{array}$& $1-x_o+x_n$ &
  $2-x_o$ &
  0 &
  0 &
  $1-x_o$ &
  $1-x_o+x_n$ \\\hline
  $\begin{array}{c}x_3 \\ \mbox{on entry}\end{array}$& 0 &
  $1-x_n$ &
  $1-x_n+x_o$ &
  $2-x_n+x_o$ &
  $3-x_n$ &
  0 \\\hline
  
  $\begin{array}{c}\mbox{Time}\\\mbox{elapsed}\end{array}$ &
  $1-x_n$ &
  $x_o$ &
  1 &
  $1-x_o$ &
  $x_n$ &
  0 \\\hline
  $\begin{array}{c}\mbox{Cost} \\ \mbox{incurred}\\\mbox{in one pass}\end{array}$ & $5-5x_n$ & $x_o$ & -5 & 0 & 0 & 0 \\\hline
  $\begin{array}{ccc} \mbox{Total cost} \\ \end{array}$ & -&- &- &- &- & 
  $\begin{array}{c}x_n = \frac{x_o}{5}+\varepsilon \\ 
  
  \mbox{one pass}\\ =-5x_n+x_o  \\= -5\varepsilon \\ 
  \mbox{after $i$ passes} \\ =-i\times 5\times \varepsilon\\
  
  \mbox{Total cost } \\ 
  =0~if~\varepsilon=0\\
  >\eta~if~\varepsilon < 0\\
  <-\eta~if~\varepsilon>0\end{array}$ \\\hline
 \end{tabular}

\caption{$\RReachOb(\eta)$: Simulation to increment counter $C_1$}
\label{fig_undec_buchi_inc}
\end{figure}

  \paragraph{Simulation of increment instruction}: The module to
  increment $C_1$ is given in Fig.~\ref{fig_undec_buchi_inc}. Again,
  we start at $\ell_k$ with $x_1=\frac{1}{5^{c_1}7^{c_2}}, x_2 = 0$ and
  $x_3=0$. Again, call $\frac{1}{5^{c_1}7^{c_2}}$ as $x_{old}$.  A
  time of $1-x_{old}$ is spent at $\ell_k$.  Let the time spent at $I$ be
  denoted $x_{new}$.  To correctly increment counter 1, $x_{new}$ must
  be $\frac{x_{old}}{5}$.  No time is spent at $\text{Check}$. Player
  2 can either continue simulation of the next instruction, or can
  enter the widget $\text{WI}_1$ to verify if $x_{new}$ is indeed
  $\frac{x_{old}}{5}$.  Fig.~\ref{fig_undec_buchi_inc} gives a table
  detailing the values of clocks, time elapsed and cost incurred at
  each location of the main module, as well as $\text{WI}_1$. It can
  be seen that the total cost incurred in one pass from location $A$
  to $F$ is $x_{old}-5x_{new}$, which is 0 iff $x_{old}=5x_{new}$. As
  seen in the case of decrement, here also, each time we come back to
  $A$, we restore the clock values with which we enter $\text{WI}_1$;
  clearly, if Player 1 makes an error of the form
  $x_{new}=\frac{x_{old}}{5}+\varepsilon$, the cost incurred in one
  pass from $A$ to $F$ is $-5 \varepsilon$. If $\varepsilon >0$, then
  Player 2 can bring the cost less than $-\eta$ for any choice of
  $\eta$ by taking the loop between $A$ and $F$ a large number of
  times. Similarly, if $\varepsilon <0$, a cost $>\eta$ can be
  incurred for any choice of $\eta$.
  
  \paragraph{Simulation of Zero-check}:
  Fig.~\ref{fig_undec_buchi_zeroCheck} gives the module for zero-check
  instruction for counter $C_2$. $\ell_k$ is a no time elapse location,
  from where, Player 1 chooses one of the locations $\text{Check}_{c_2=0}$ or
  $\text{Check}_{c_2 \neq 0}$. Both these are Player 2 locations, and Player 2
  can either continue the simulation, or can go to the check widgets
  $\text{W}_2^{=0}$ or $\text{W}_2^{\neq 0}$ to verify the correctness of Player 1's
  choice. 
  
\begin{figure}[tbp]
\centering
    \begin{tikzpicture}[->,>=stealth',shorten >=1pt,auto,node distance=1cm,
      semithick,scale=0.8]
      \node[initial,initial text={}, player1] at (-.5,0) (A) {$0$};
      \node()[above of=A,node distance=5mm,color=gray]{$\ell_k$};
      \node () [below right of=A,node distance=6.5mm,xshift=-3mm] {$[x_3=0]$};

      \node[player2] at (2,1) (B){$0$} ;
       \node()[below of=B,node distance=5mm,color=gray]{$\text{Check}_{c_2=0}$};
       
      \node[widget,color=gray] at (2,2.5) (B1){$\text{W}_2^{= 0}$ };

      \node[player2] at (2,-1) (C){$0$} ;
      \node()[above of=C,node distance=5mm,color=gray]{$\text{Check}_{c_2\neq 0}$};
      
      \node[widget,color=gray] at (2,-2.5) (C1){$\text{W}_2^{\neq 0}$};

      \node[player1] at (4.5,1) (D) {$0$};
      \node()[above of=D,node distance=5mm,color=gray]{$\ell_{k+1}^1$};

      \node[player1] at (4.5,-1) (E) {$0$ };
      \node()[above of=E,node distance=5mm,color=gray]{$\ell_{k+1}^2$};

      \path (A) edge (B);
      \path (A) edge (C);
      \path (B) edge node {$x_3=0$} (B1);
      \path (B) edge node[above] {$x_3=0$} (D);
      \path (C) edge node[left] {$x_3=0$} (C1);
      \path (C) edge node[above] {$x_3=0$} (E);
    \end{tikzpicture}
    \caption{$\RReachOb(\eta)$: Widget $\text{WZ}_2$ simulating
      zero-check for $C_2$}

\label{fig_undec_buchi_zeroCheck}
\end{figure}
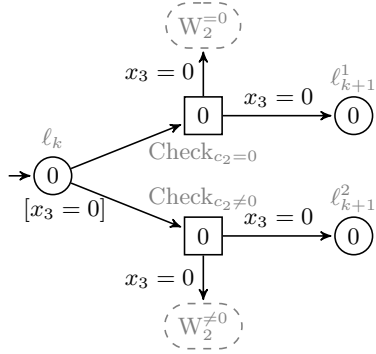

  The widgets $\text{W}_2^{=0}$ and $\text{W}_2^{\neq 0}$ are given in
  Fig.~\ref{fig_undec_buchi_wze} and \ref{fig_undec_buchi_wzne}
  respectively.
  \begin{itemize}
  \item Consider the case when Player 1 guessed that $C_2$ is zero,
    and entered the location $\text{Check}_{c_2=0}$ in
    Fig.~\ref{fig_undec_buchi_zeroCheck}.  Let us assume that Player 2
    verifies Player 1's guess by entering $\text{W}_2^{=0}$
    (Fig.~\ref{fig_undec_buchi_wze}).  No time is spent in the initial
    location $A$ of $\text{W}_2^{=0}$.  We are therefore at $B$ with
    $x_1=\frac{1}{5^{c_1}7^{c_2}}=x_{old}$ and $x_2,x_3=0$.  In case
    $c_1=c_2=0$, we can directly go to the $F$ state, and stay there
    forever, incurring cost 0.  If that is not the case, Player 1 has
    to prove his claim right, by multiplying $x_1$ with 5 repeatedly,
    till $x_1$ becomes 1; clearly, this is possible iff $c_2=0$.  The
    loop between $B$ and $\text{Check}$ precisely does this: each time
    Player 1 spends a time $x_{new}$ in $B$, Player 2 can verify that
    $x_{new}=5 x_{old}$ by going to $\text{WD}_1$, or come back to
    $B$.  No time is elasped in $\text{Check}$. Finally, if $x_1=1$,
    we can go to $F$, and Player 1 achieves his objective. However, if
    $C_2$ was non-zero, then $x_1$ will never reach 1 after repeatedly
    multiplying $x_1$ with 5; in this case, at some point, the edge
    from $\text{Check}$ to $C$ will be enabled. In this case, the
    infinite loop between $C$ and $T$, incurs a cost $+\infty$.
    
  \item Consider now the case when Player 1 guessed that $C_2$ is
    non-zero, and hence entered the location $\text{Check}_{c_2\neq
      0}$ in Fig.~\ref{fig_undec_buchi_zeroCheck}.  Let us assume now
    that Player 2 enters $\text{W}_2^{\neq 0}$
    (Fig.~\ref{fig_undec_buchi_wzne}) to verify Player 1's
    guess. Similar to $\text{W}_2^{=0}$, no time is spent at location
    $A$ of $\text{W}_2^{\neq 0}$, and the clock values at $B$ are
    $x_1=\frac{1}{5^{c_1}7^{c_2}}=x_{old}$ and $x_2,x_3=0$.  If
    $c_1=c_2=0$, then $x_2=1$, in which case, the location $D$ is
    reached, from where, the loop between $D,T$ is taken incurring a
    cost $+\infty$.  There are two possibilities now: ($i$) $B$ can go
    to $C$ or ($ii$) to $\text{Check}_{c_1}$.  In case $B$ goes to
    $\text{Check}_{c_1}$, then Player 1 repeatedly multiplies $x_1$ by
    5 till we obtain $x_1$ as $\frac{1}{7^{c_2}}$.  Note that mistakes
    committed by Player1 in the multiplication by 5, can always be
    caught by Player 2 via $\text{WD}_1$.
      
    If $c_1=0$ already, then we can straightaway go to $C$ spending no
    time at $B$, if $c_2 \neq 0$.  In this case, Player 1 has to
    compulsorily go to $\text{Check}_{c_2}$ from $C$ atleast once,
    since the edge from $C$ to $F$ is not enabled.  The loop between
    $C$ and $\text{Check}_{c_2}$ results in Player 1 multiplying $x_1$
    of the form $\frac{1}{7^{c_2}}$ till $x_1$ becomes 1.  Here again,
    if Player 1 commits a mistake during multiplication by 7, Player 2
    can catch Player 1 by entering the widget
    $\text{WD}_2$. Otherwise, when $x_1$ reaches 1, Player 1 can go
    from $C$ to $F$ achieving his objective.
  \end{itemize}
   
\begin{figure}[tbp]
 \begin{center}
\begin{tikzpicture}[->,>=stealth',shorten >=1pt,auto,node distance=1cm,
  semithick,scale=0.8]
  \node[initial,initial text={}, player1] at (1,0) (a) {$0$};
  \node()[above of=a,node distance=5mm,color=gray]{$A$};
  
  \node[player1] at (4,0) (b) {$0$};
  \node()[above of=b,node distance=5mm,color=gray]{$B$};
  
  \node[player2] at (8,0) (chk) {$0$};
  \node()[above of=chk,node distance=5mm,color=gray]{Check};
  \node()[below right of=chk,node distance=6.5mm]{$[x_3{=}0]$};

  \node[widget,color=gray] at (11,0) (wd){$\text{WD}_1$};
  
  \node[player2] at (8,-2) (c) {$1$};
  \node()[below of=c,node distance=5mm,color=gray]{C};
  
  \node[player2] at (11,-2) (d) {$0$};
  \node()[below of=d,node distance=5mm,color=gray]{$T$};
  
  \node[accepting,player1] at (4,-2) (t) {$0$};
  \node()[below of=t,node distance=5mm,color=gray]{$F$};
  
  \path (a) edge node {$x_2{=}0$} (b);
  \path (b) edge[bend right=10] node[below] {$\set{x_3}$} (chk);
  \path (chk) edge [bend right=10] node[above] {$x_1 {\leq} 1,\{x_2\}$} (b);  
  \path (chk) edge node {} (wd);
  \path (b) edge node[left] {$\begin{array}{c}x_1{=}1\\ \land x_2=0\end{array}$} (t);
  \path (t) edge [loop left] node {$x_2=0$} (t);
  \path (chk) edge node[left]{$x_1{>}1$} node[right]{$\set{x_2}$} (c);
  \path (c) edge[bend right=10] node[below] {$0{<}x_2{<}1$} (d);
  \path (d) edge[bend right=10] node[above] {$\set{x_2}$} (c);
\end{tikzpicture}

\vspace{0.5cm}
Let $\alpha=\frac{1}{5^{c_1}7^{c_2}}$
\vspace{0.25cm}

   \begin{tabular}{|c|c|c|c|c|c|c|c|}
     \hline
     Location $\rightarrow$ & $A$ & $B$ & Check & $C$ & $T$ & $F$ \\\hline  
     $\begin{array}{c} x_1 \\ \mbox{on entry} \end{array}$&
     $\alpha$ &
     $\begin{array}{c} 5^i \times \alpha\\ \mbox{loop} \\
       B \rightarrow \text{Check} \\ \mbox{taken} \\ \mbox{$i$ times}\end{array}$ & $\begin{array}{c} 5^i \times \alpha + k \\  k=  4 (5^i \times \alpha) + \varepsilon\end{array}$  & $\begin{array}{c} >1 \\ 5^i \times \alpha > 1 \\ c_2 \neq 0 \end{array}$  &
     - &
     $\begin{array}{c} 1 \\ 5^i \times \alpha = 1 \\ i=c_1 \wedge c_2=0 \end{array}$ 
     \\\hline
  
  $\begin{array}{c} x_2 \\ \mbox{on entry} \end{array}$& 0 & 0 & $k$  & 0 & $p>0$ & 0 \\\hline
  $\begin{array}{c} x_3 \\ \mbox{on entry} \end{array}$& 0 & 0 & 0  & - & - & 0 \\\hline
  
  $\begin{array}{c} \mbox{Time} \\ \mbox{elapsed} \end{array}$& 0 & $k$ &  $0$ &$p$  &0 &0 \\\hline
  
  $\begin{array}{c} \mbox{Cost} \\ \mbox{incurred} \\ \mbox{ in one pass} \end{array}$& 0 & 0 & 0 & $p$ & 0 & $0$ \\\hline
  
  $\begin{array}{ccc} \mbox{Total cost} \\ \end{array}$ & - & - &  - & $>\eta$ &  - & $0$ \\\hline
 \end{tabular}
 \end{center} 

 \caption{$\RReachOb(\eta)$: Widget $\text{W}_2^{=0}$. Delay at $B$ is
   $k = 4(5^i\times \alpha) + \varepsilon =4(5^i\times
   \frac{1}{5^{c_1}7^{c_2}}) + \varepsilon$ . If $\varepsilon \neq 0$
   then Player 2 will enter the widget $WD_1$ and where the cost
   incurred $\not \in [-\eta,\eta]$ if $\varepsilon\neq 0$.}
\label{fig_undec_buchi_wze}
\end{figure}

\begin{figure}[tbp]
\begin{center}
\begin{tikzpicture}[->,>=stealth',shorten >=1pt,auto,node distance=1cm,
  semithick,scale=0.8]
  \node[initial,initial text={}, player1] (a) {$0$};
  \node()[above of=a,node distance=5mm,color=gray]{$A$};
  
  \node[player1] at (4,0) (b) {$0$};
  \node()[above of=b,node distance=5mm,color=gray]{$B$};
  
  \node[player2] at (8,0) (chk) {$0$};
  \node()[above of=chk,node distance=5mm,color=gray]{$\text{Check}_{c_1}$};
  \node()[below of=chk,node distance=5mm]{$[x_3{=}0]$};

  \node[widget,color=gray] at (11,0) (wd){$\text{WD}_1$};
  
  \node[player2] at (0,-3) (d) {$1$};
  \node()[above of=d,node distance=5mm,color=gray]{$D$};
  
  \node[player2] at (0,-5) (t) {$0$};
  \node()[below of=t,node distance=5mm,color=gray]{$T$};
  
  \node[player1] at (4,-3) (c) {$0$};
  \node()[left of=c,node distance=5mm,color=gray]{C};
  
  \node[player2] at (8,-3) (chk2) {$0$};
  \node()[above of=chk2,node distance=5mm,color=gray]{$\text{Check}_{c_2}$};
  \node()[below of=chk2,node distance=5mm]{$[x_3{=}0]$};

  \node[widget,color=gray] at (11,-3) (wd2){$\text{WD}_2$};
  
  \node[accepting,player1] at (4,-5) (t1) {$0$};
  \node()[below of=t1,node distance=5mm,color=gray]{$F$};
  
  \path (a) edge node {$x_2{=}0$} (b);
  \path (b) edge[bend right=10] node[below] {$\set{x_3}$} (chk);
  \path (chk) edge [bend right=10] node[above] {$x_1 {\leq} 1,\{x_2\}$} (b);  
  \path (chk) edge node {} (wd);
  \path (b) edge node[left] {$\begin{array}{c}x_1{=}1\\ \land x_2=0\end{array}$} (d);
  \path (d) edge[bend left=20] node[right] {$0{<}x_2{<}1$} (t);
  \path (t) edge[bend left=20] node[left] {$\set{x_2}$} (d);
  
  \path (b) edge node[right]{$\begin{array}{c}x_1{<}1\\ \land x_2=0\end{array}$} (c);
  \path (c) edge[bend right=10] node[below] {$\set{x_3}$} (chk2);
  \path (chk2) edge [bend right=10] node[above] {$x_1 {\leq} 1,\{x_2\}$} (c);  
  \path (chk2) edge node {} (wd2);
  \path (c) edge node[left]{$\begin{array}{c}x_1{=}1\\ \land x_2=0\end{array}$} (t1);
  \path (t1) edge [loop right] node {$x_2=0$} (t1);
\end{tikzpicture}

  \vspace{0.55cm}
  Let $\alpha =\frac{1}{5^{c_1}7^{c_2}} $
  \vspace{0.25cm}

  \begin{tabular}{|c|c|c|c|c|c|c|c|}
    \hline
    Loc $\rightarrow$ & $A$ & $B$ & $D$ & $T$ & $C$ & $F$\\\hline  
    $\begin{array}{c} x_1 \\ \mbox{on entry} \end{array}$& $\alpha$ &
    $\begin{array}{c} 5^i \times \alpha\\ \mbox{loop }\\ B \rightarrow \text{Check} \\ \mbox{taken} \\ \mbox{$i$ times}\end{array}$    &
    $\begin{array}{c} 1 \\ 5^i \times \alpha = 1 \\ i=c_1 \wedge c_2=0 \end{array}$ & -& 
    $\begin{array}{c} 7^j \times \frac{1}{7^{c_2}}  \\  \mbox{loop }\\ C \rightarrow \text{Check} \\ \mbox{taken} \\ \mbox{$j$ times}\end{array}$      &
    $\begin{array}{c} 1 \\7^j\times 5^i \times \alpha = 1 \\ i=c_1 \wedge \\j=c_2>0 \end{array}$ 
    \\\hline
  
  $\begin{array}{c} x_2 \\ \mbox{on entry} \end{array}$& 0 & 0 & 0 &  $p>0$ & 0 & 0 \\\hline
  
  $\begin{array}{c} \mbox{Time} \\ \mbox{elapsed} \end{array}$& 0 &
  $4(5^i\times \alpha)+\varepsilon$  & $p$ & 0 & $6(7^j\times\frac{1}{7^{c_2}})+\varepsilon$  & 0 \\\hline
  
  $\begin{array}{c} \mbox{Cost} \\ \mbox{incurred} \\ \mbox{in one pass} \end{array}$& 0 & 0 & $p$ & 0 & 0 & 0 \\\hline
  
  $\begin{array}{ccc} \mbox{Total cost} \\ \end{array}$ & - & - &  $>\eta$ & - & - & 0 \\\hline
 \end{tabular}
\end{center}
 \caption{$\RReachOb(\eta)$: Widget $\text{W}_2^{\neq 0}$. Widget
   $\text{WD}_2$ following $\text{Check}_{c_2}$ is similar to
   $\text{WD}_1$ shown earlier, except that the prices are adjusted to
   verify decrement of counter $C_2$.}
\label{fig_undec_buchi_wzne}
\end{figure}

  \paragraph{Correctness of the construction}
  On entry into the location $\location_n$ (for HALT instruction), we
  reset clock $x_1$ to 0; from $\location_n$, we go to a state $F$
  with price 0, with a self loop $x_1=0$.
  \begin{enumerate}
  \item Assume that the two counter machine halts. If Player 1
    simulates all the instructions correctly, he will incur a cost
    $=0$, by either reaching the $F$ after $\location_n$, or by
    entering a widget (the second case only occurs if Player~$2$
    decides to check whether Player~$1$ simulates the machine
    faithfully.  If Player 1 makes an error in his computation, Player
    2 can always enter an appropriate widget, making the cost as large
    or as small, so as to not fit in $[-\eta,\eta]$ for any choice of
    $\eta$.  In summary, if the two counter machine halts Player 1 has
    a strategy to achieve his goal (visiting $F$ with a cost 0 $\in
    [-\eta,\eta]$ for any $\eta>0$.)
  \item Assume that the two counter machine does not halt.
    \begin{itemize}
    \item If Player 1 simulates all the instructions correctly, and if
      Player 2 never enters a check widget, then Player 1 incurs cost
      $\infty$ as we never reach $F$.  In this case,
      Player 2 will never want to enter a widget, since he gets a
      higher payoff.
    \item Suppose now that Player 1 makes an error. In this case,
      Player 2 always has the possibility to enter a loop, where
      Player 1 will incur cost $\infty$ or $-\infty$.
    \end{itemize}
    In summary, if the two counter machine does not halts Player 1
    does not have a strategy to achieve his goal.
  \end{enumerate}
  Thus, Player 1 incurs a cost in $[-\eta,\eta]$ for any $\eta$ iff he
  chooses the strategy of faithfully simulating the two counter
  machine, when the machine halts. When the machine does not halt, the
  cost incurred by Player 1 is not in $[-\eta,\eta]$ for any chosen
  $\eta$ if Player 1 made a simulation error and Player 2 entered a
  widget. Else if a widget is not entered then the run does not reach
  $F$ and cost is $+\infty$.  \qed
\end{proof}

\end{document}